\documentclass[11pt, numbers, sort&compress]{elsarticle}
\usepackage{epsfig,graphics}
\usepackage{amssymb,amsmath,amsthm,bm, amscd, amsfonts}
\usepackage{booktabs}  
\usepackage{threeparttable}  
\usepackage{multirow}
\usepackage{mathabx}
\usepackage{siunitx}
\usepackage{booktabs}
\usepackage{algorithm}
\usepackage{algpseudocode}
\usepackage{subcaption}
\usepackage[hidelinks]{hyperref}
\usepackage{float}
\usepackage[title]{appendix}

\theoremstyle{plain}
\newtheorem{exam}{Example}[section]

\newtheorem{remark}[exam]{Remark}

\newtheorem{theo}{Theorem}
\newtheorem{lem}{Lemma}

\newtheorem{mdef}{Definition}

\def\gga{\gamma}
\def\gth{\theta}

\def\gO{\Omega}

\def\gs{\sigma}
\def\gl{\lambda}

\def\n{\noindent}

\def\b0{{\bf 0}}
\def\1{{\bf 1}}

\def\wh{\widehat}

\makeatletter
\def\widebreve{\mathpalette\wide@breve}
\def\wide@breve#1#2{\sbox\z@{$#1#2$}%
     \mathop{\vbox{\m@th\ialign{##\crcr
\kern0.08em\brevefill#1{0.99\wd\z@}\crcr\noalign{\nointerlineskip}%
                    $\hss#1#2\hss$\crcr}}}\limits}
\def\brevefill#1#2{$\m@th\sbox\tw@{$#1($}%
  \hss\resizebox{#2}{\wd\tw@}{\rotatebox[origin=c]{90}{\upshape(}}\hss$}
\makeatletter

\begin{document}

\begin{frontmatter}




\title{Time-reassigned synchrosqueezing frequency-domain chirplet transform for multicomponent signals with intersecting group delay curves}
\author[university1]{Shuixin Li}
\ead{lishuixin@zjnu.edu.cn}
\author[university1]{Jiecheng Chen}
\ead{jcchen@zjnu.cn}
\author[university1]{Qingtang Jiang}
\ead{jiangq@zjnu.edu.cn}
\author[university2]{Lin Li}
\ead{lilin@xidian.edu.cn}

\address[university1]{School of Mathematical Sciences, Zhejiang Normal University, Jinhua 321004, China}
\address[university2]{School of Electronic Engineering, Xidian University, Xi'an 710071, China}

\begin{abstract}

To analyze signals with rapid frequency variations or transient components, the time-reassigned synchrosqueezing transform (TSST) and its variants have been recently proposed. 
Unlike the traditional synchrosqueezing transform, TSST squeezes the time-frequency (TF) coefficients along the group delay (GD) trajectories rather than the instantaneous frequency trajectories.
Although TSST methods perform well in analyzing transient signals, they are fundamentally limited in processing multicomponent signals with intersecting GD curves.
This limitation compromises the accuracy of both feature extraction and signal component recovery, thereby significantly reducing the interpretability of time-frequency representations (TFRs). This is particularly problematic in broadband signal processing systems, where the linearity of the phase response is critical and precise measurement of group delay dispersion (GDD) is essential.

Motivated by the superior capability of frequency-domain signal modeling in characterizing rapidly frequency-varying signals, this paper proposes a novel three-dimensional time-frequency-group delay dispersion (TF-GDD) representation based on the frequency-domain chirplet transform. A subsequent time-reassigned synchrosqueezing frequency-domain chirplet transform (TSFCT) is introduced to achieve a sharper TF-GDD distribution and more accurate GD estimation. For mode retrieval, a novel frequency-domain group signal separation operation (FGSSO) is proposed.
The theoretical contributions include a derivation of the approximation error for the GD and GDD reference functions and an establishment of the error bounds for FGSSO-based mode retrieval. Experimental results demonstrate that the proposed TSFCT and FGSSO effectively estimate GDs and retrieve modes--even for modes with intersecting GD trajectories.

\end{abstract}


\end{frontmatter}

\section{Introduction}
Natural-world signals are often composed of multiple, superimposed non-stationary components. 
In these cases, time-frequency analysis (TFA) is an indispensable tool, as it  characterizes both temporal and spectral dynamics, facilitating the decomposition of such multicomponent signals.
Conventional linear TFA techniques, particularly the short-time Fourier transform (STFT) \cite{stankovic2013time} and continuous wavelet transform (CWT) \cite{daubechies1992ten,mallat1999wavelet}, are fundamentally constrained by the Heisenberg uncertainty principle \cite{cohen1995time}, 
which enforces an inescapable trade-off between temporal and spectral resolution, precluding their simultaneous optimization.
While quadratic representations like the Wigner-Ville distribution  \cite{hlawatsch1992linear}  offer theoretically superior resolution, they inevitably generate  cross-terms when analyzing multicomponent signals, severely limiting their practical utility in real-world applications.

Enhancing the time-frequency resolution of TFA methodologies to precisely unveil latent signal features remains a critical requirement in practical signal processing applications.
The reassignment method (RM) \cite{auger1995improving} was developed to sharpen time-frequency representations (TFRs) and does improve energy concentration; however, it cannot reconstruct the original signal.
To address this limitation, the synchrosqueezing transform (SST) was proposed 
in \cite{daubechies1996nonlinear} and further developed in \cite{daubechies2011synchrosqueezed}.
By concentrating time-frequency coefficients along instantaneous frequency (IF) trajectories, SST achieves both enhanced readability and mode retrieval. 
This dual capability makes it a powerful tool for analyzing non-stationary signals and has inspired extensive research into high-resolution TFR techniques, including second-order \cite{oberlin2015second, oberlin2017second, behera2018theoretical} and high-order SST variants \cite{pham2017high}, the synchroextracting transform (SET) \cite{yu2017synchroextracting}, the multisynchrosqueezing transform (MSST) \cite{yu2018multisynchrosqueezing}, 
and the adaptive SST \cite{sheu2017entropy,berrian2017adaptive,li2020adaptive,li2020adaptivestft}. 

Additionally, the chirplet transform (CT) 
and wavelet-chirplet transform (WCT) have been developed for processing multicomponent signals with crossing IF trajectories  \cite{chui2021time,li2022chirplet,chui2023analysis}. 
These transforms generalize the conventional time-frequency and time-scale frameworks to three dimensions, resulting in the time-frequency-chirprate (TFC) and time-scale-chirprate (TSC) domains, respectively.
For brevity, the term  ``chirprate'' is used throughout this paper in place of  ``chirp rate''.
Under certain conditions, these approaches enable effective separation of components in non-stationary multicomponent signals--even in the cases with crossover IFs.  
When synchrosqueezing techniques are applied within these three-dimensional spaces, they yield significantly sharper TFC representations and substantially improve the separation of signal components \cite{zhu2020frequency,chen2023disentangling,chen2024multiple,chen2024composite,jiang2025synchrosqueezed}. 

While these methods perform well for non-stationary signals with smoothly varying IFs, they encounter significant challenges when analyzing signals that exhibit rapid frequency variations or transient components.
This limitation has motivated the development of the time-reassigned synchrosqueezing transform (TSST)  \cite{he2019time,yu2020time,li2022theoretical}. 
Unlike traditional SST, TSST squeezes the time-frequency coefficients along the group delay (GD) trajectories instead of the IF trajectories.
However, TSST still suffers from energy dispersion for strongly frequency-varying signals. Subsequent refinements have included second-order TSST \cite{fourer2019second,he2020gaussian},  
the transient-extracting transform (TET) \cite{yu2019concentrated}, the Newton time-extracting wavelet transform \cite{li2023newton}, second-order TET \cite{yu2021second,he2019second}, the generalized transient squeezing transform \cite{bao2022generalized} and the generalized TET \cite{bao2021generalized}. 
While existing TSST methods demonstrate strong performance in transient signal analysis, they face a fundamental limitation in resolving intersecting GD curves. This compromises both the accuracy of feature extraction and component recovery, thereby significantly degrading the interpretability of TFRs.
Such constraints are particularly detrimental in broadband signal processing systems, where phase linearity is critical for performance and precise group delay dispersion (GDD) measurement is essential.

The frequency-domain chirplet transform (FCT) was introduced in \cite{yang2014frequency} to characterize frequency-varying GD components.
This paper extends the FCT framework into a three-dimensional time-frequency-group delay dispersion (TF-GDD) space, enabling the separation of multicomponent signals—even in challenging scenarios with intersecting GD curves.
Furthermore, drawing inspiration from the three-dimensional synchrosqueezing transform of CT, we construct the GD and GDD reference functions (also called the reassignment operators). This foundation enables the development of the three-dimensional time-reassigned synchrosqueezing frequency-domain chirplet transform (TSFCT).
The TSFCT provides a more concentrated TF-GDD representation and substantially strengthens feature extraction capabilities, thereby boosting the accuracy of transient signal mode recovery.
Additionally, a theorem is derived on the approximation error between the GD/GDD reference functions and their true signal counterparts.
To retrieve transient signal modes that may exhibit intersecting group delay (GD) curves, this paper proposes extending the signal separation operator (SSO) framework (originally rooted in STFT and CWT \cite{chui2016signal,li2022direct,chui2021analysis, chui2021signal}) to the three-dimensional time-frequency-group delay dispersion (TF-GDD) space.
This extension facilitates the effective extraction of modes from multicomponent transient signals, and a theorem on the error bounds for mode retrieval using FGSSO is established.

The remainder of this paper is organized as follows. In Section 2 
we derive the GD and GDD reference functions, based on which the TSFCT is defined. The errors for the approximations of  the GD and GDD reference functions to GD and GDD are presented with proof provided in the Appendix.    
Section 3  presents the frequency-domain group signal separation operation (FGSSO) scheme for mode retrieval and establishes the corresponding error bounds for the reconstructed signal. 
Section 4 provides the implementation details and experimental validation to verify the effectiveness of the proposed method. Finally, Section 5 concludes the paper and summarizes the main contributions.

In this paper, without loss of generality, the multicomponent signal $x(t)$ is represented in the frequency domain by:
\begin{align}
    \widehat{x}(\eta) = \sum_{k=1}^{K} \widehat{x}_k(\eta) = \sum_{k=1}^{K} B_k(\eta) e^{-i2\pi  \theta_k(\eta)},\label{Frequency-domain_model}
\end{align}
where for the $k$-th component, $B_k(\eta)> 0$ and $\theta_k(\eta)$ are its amplitude and phase, respectively, and $\theta'_k(\eta)$ denotes the group delay (GD).

\section{Time-reassigned  synchrosqueezing frequency-domain chirplet transform}

\subsection{Frequency-domain chirplet transform}

Recall the chirplet transform \cite{mann1995chirplet} of a signal $x(t) \in L^2(\mathbb{R})$ with a window function $g(t) \in L^2(\mathbb{R})$, which is defined as
\begin{equation}
Q_{x}^g(t, \eta, \gga) := \int_{\mathbb{R}} x(t + \tau) \, g(\tau) \, e^{-i2\pi  \eta \tau} e^{-i\pi  \gga \tau^2}  d\tau.
\end{equation}

The frequency-domain chirplet transform (FCT) framework was first introduced in \cite{yang2014frequency}. 
To align with the subsequent FGSSO scheme, we adopt the following definition of the FCT for a signal $\widehat{x}(\xi) \in L^2(\mathbb{R})$ and a frequency-domain window $g(\xi) \in L^2(\mathbb{R})$:
\begin{equation}
\mathcal{D}_{x}^g(t, \eta, \gga) := \int_{\mathbb{R}} \widehat{x}(\xi + \eta) \, g(\xi) \, e^{i2\pi  \xi t} e^{i\pi  \gga \xi^2}  d\xi, \label{defFCT_freq}
\end{equation}
where $\widehat{x}(\xi)$ denotes the Fourier transform of $x(t)$.

The CT and FCT exhibit a fundamental duality:
\[
\mathcal{D}_{x}^g(t, \eta, \gga) = {Q}_{\widehat{x}}^g(\eta, -t, -\gga), \quad
Q_{x}^g(t, \eta, \gga) = \mathcal{D}_{\widehat{x}}^{g(-\cdot)}(\eta, -t, -\gga).
\]
This time-frequency duality provides the mathematical foundation for referring to $\mathcal{D}_{x}^g(t, \eta, \gga)$ as the frequency-domain chirplet transform.

\cite{Huang2024horizontal, zhao2025horizontal} proposed horizontal rearrangement techniques based on the FCT scheme. 
However, these methods fundamentally view the FCT as a transform with a paramter $\gga$ and operate within the two-dimensional time-frequency plane. Consequently, like the TSST methods, they fail to resolve signals with intersecting GD curves.

By employing a three-dimensional synchrosqueezing technique, the synchrosqueezing chirplet transform \cite{zhu2020frequency,chen2023disentangling,chen2024multiple} and synchrosqueezing wavelet-chirplet transform \cite{chen2024composite,jiang2025synchrosqueezed} deliver significantly sharper TFC representations than the CT and WCT.
This improvement  allows  for the accurate estimation of IFs and chirprates in multicomponent signals, even with crossing IF trajectories.
In this paper, we address the challenge of analyzing multicomponent signals with modes that exhibit rapidly frequency-varying  characteristics and  probably intersecting GD curves. To this end, by constructing novel GD and GDD reference functions, we propose the time-reassigned  synchrosqueezing frequency-domain chirplet transform (TSFCT), which achieves a sharper TF-GDD representation and enable accurate GD estimation even in the cases with crossing GD trajectories.

\subsection{Time-reassigned  synchrosqueezing frequency-domain chirplet transform}

The derivation of the IF and chirprate reference functions for the synchrosqueezing chirplet transform
in \cite{chen2023disentangling} is rather lengthy and complex. Inspired by the 
work of \cite{bao2021generalized,meignen2022analysis}, we propose a more streamlined approach to obtain concise expressions of the GD and GDD reference functions for TSFCT.  The derivation is presented below.

Differentiating both sides of  \eqref{defFCT_freq} with respect to \(\eta\), we obtain
\begin{align}
	\partial_\eta \mathcal{D}_x^{g}(t,\eta,\gga)=& \int_{\mathbb{R}}  \frac{\partial}{\partial_\xi}\big( \widehat{x}(\xi + \eta)\big) g(\xi)  e^{i2\pi \xi t}  e^{i\pi  \gga \xi^2} \, d\xi  \nonumber\\ 
	  =-& \int_{\mathbb{R}}  \widehat{x}(\xi + \eta) \frac{\partial}{\partial_\xi}\big(g(\xi)   e^{i2\pi \xi t}  e^{i\pi  \gga \xi^2}\big) \, d\xi  \nonumber\\
      =-& i2\pi t \mathcal{D}_x^{g}(t,\eta,\gga) -i2\pi \gga  \mathcal{D}_x^{\xi g}(t,\eta,\gga) -\mathcal{D}_x^{g'}(t,\eta,\gga). \label{partialeta}
	\end{align}
Here and below $\mathcal{D}_x^{g'}$ and $\mathcal{D}_x^{\xi^j g}$ (for a natural number $j$) denote the FCT of $x(t)$ defined by \eqref{defFCT_freq} with $g(\xi)$ replaced by $g'(\xi)$ and $\xi^j g(\xi)$, respectively.  

Additionally, we can  take the partial derivative of \eqref{defFCT_freq} with respect to \(t\) to reach 
\begin{align}
	\label{partial_t}	\partial_t \mathcal{D}_x^{g}(t,\eta,\gga)=i2\pi  \mathcal{D}_x^{\xi g}(t,\eta,\gga).
\end{align}
Assume that \(x(t)\) is a generalized chirp signal in the frequency domain,
 which means 
\begin{equation}
\label{def_chirpf} 
\wh{x}(\eta) =e^{-(p \eta+\frac{1}{2}q \eta^2)} e^{-i2\pi  (c\eta+\frac{1}{2}r \eta^2)}=B(t) e^{-i2\pi  \gth(t)}.
\end{equation}
where \(p,q,c,r\) are  real constants.
Then, from \eqref{defFCT_freq}, we have 
\begin{align}
   \mathcal{D}_x^{g}(t,\eta,\gga)=\wh{x}(\eta)  \int_{-\infty}^{\infty} e^{-{(p+q \eta+ i2\pi (c+r \eta))\xi +(q+i2\pi  r )\xi^2    }   } g(\xi) e^{i2\pi  \xi t} e^{i\pi  \gga \xi^2} d\xi. \label{FCT_polynomials} 
\end{align}
Taking the partial derivative of \eqref{FCT_polynomials} with respect to \(\eta\), we get
\begin{align}
  \label{partial_eta}   \partial_\eta\mathcal{D}_x^{g}(t,\eta,\gga) &= -(p+q \eta+ i2\pi (c+r \eta) )  \mathcal{D}_x^{g}(t,\eta,\gga)-( q +i2\pi  r )   \mathcal{D}_x^{\xi g}(t,\eta,\gga).
\end{align}
Next, differentiating both sides of \eqref{partial_eta} with respect to  \(t\), we obtain:
\begin{align*}
   \partial_t \partial_\eta\mathcal{D}_x^{g}(t,\eta,\gga) &= -(p+q \eta+ i2\pi (c+r \eta) )  \partial_t\mathcal{D}_x^{g}(t,\eta,\gga)-( q +i2\pi  r )  \partial_t \mathcal{D}_x^{\xi g}(t,\eta,\gga).
\end{align*}
Substituting \eqref{partial_t} into \eqref{partial_eta}, we ultimately get
\begin{align}
 \label{partial_eta_xig}   \partial_\eta\mathcal{D}_x^{\xi g}(t,\eta,\gga) &= -(p+q \eta+ i2\pi (c+r \eta) ) \mathcal{D}_x^{\xi g}(t,\eta,\gga)-( q +i2\pi  r )  \mathcal{D}_x^{\xi^2 g}(t,\eta,\gga).  
\end{align}

The derived relationships \eqref{partial_eta} and \eqref{partial_eta_xig} can be   represented in matrix form as:
\begin{align*}
    \begin{bmatrix}
        \partial_\eta \mathcal{D}_x^{g} \\[0.5em]
        \partial_\eta \mathcal{D}_x^{\xi g} \\[0.5em]
    \end{bmatrix}
    &=
   \begin{bmatrix}
        \mathcal{D}_x^{g} & \mathcal{D}_x^{\xi g}  \\[0.5em]
        \mathcal{D}_x^{\xi g} & \mathcal{D}_x^{\xi^2 g} \\[0.5em]
    \end{bmatrix}
    \begin{bmatrix}
       -\left(p+q \eta+ i2\pi (c+r \eta) \right) \\[0.5em]
       -\left(q + i2\pi  r\right) \\[0.5em]
    \end{bmatrix}.
\end{align*}
Next,  define the matrices $E_0$, $E_1$, and $E_2$ as follows:
\begin{align}\label{matrix_def}
E_0 = \begin{bmatrix}
      \mathcal{D}_x^{g} & \mathcal{D}_x^{\xi g}  \\
      \mathcal{D}_x^{\xi g} & \mathcal{D}_x^{\xi^2 g}
    \end{bmatrix}, \quad
E_1 = \begin{bmatrix}
      \partial_\eta  \mathcal{D}_x^{g} &  \mathcal{D}_x^{\xi g}  \\
      \partial_\eta   \mathcal{D}_x^{\xi g} &  \mathcal{D}_x^{\xi^2 g}
    \end{bmatrix}, \quad
E_2 = \begin{bmatrix}
      \mathcal{D}_x^{g} & \partial_\eta\mathcal{D}_x^{g}  \\
      \mathcal{D}_x^{\xi g} & \partial_\eta \mathcal{D}_x^{\xi g}
    \end{bmatrix}.
\end{align}
Let \(\epsilon > 0\) be a given threshold and define the region \(E_\epsilon\)  as:
\begin{align} \label{region_E}
  E_\epsilon = \left\{ (t,\eta,\gga) : \left| \det(E_0) \right| > \epsilon \right\}
\end{align}
where \(\det(\cdot)\) denotes the determinant of a square matrix.
Thus, for any \((t, \eta, \gga) \in E_\epsilon\), by Cramer's rule, we  obtain
\begin{align*} 
p+q \eta+ i2\pi (c+r \eta) =-\frac{\det(E_1)}{\det(E_0)},\quad  q + i2\pi  r =-\frac{\det(E_2)}{\det(E_0)}.
\end{align*}
Subsequently, 
\begin{align*} 
c+r \eta =-\frac{1}{2\pi}{\rm Im}\Big(\frac{\det(E_1)}{\det(E_0)}\Big), \quad r=-\frac{1}{2\pi}{\rm Im}\Big(\frac{\det(E_2)}{\det(E_0)}\Big).
\end{align*}
For a general signal $x(t)$, we may define the ideal  reference functions for GD and GDD as 
 \begin{equation}
\label{def_gO1rev}
\widehat{t}(t,\eta,\gga):=-\frac{1}{2\pi}{\rm Im}\Big(\frac{\det(E_1)}{\det(E_0)}\Big), \quad \widehat{r}(t,\eta,\gga):=-\frac{1}{2\pi}{\rm Im}\Big(\frac{\det(E_2)}{\det(E_0)}\Big).
\end{equation}

By substituting the window function \( g(\xi) \) in \eqref{partialeta} 
with \( \xi g(\xi) \), we  obtain the following result:
\begin{align}
\partial_\eta \mathcal{D}_x^{\xi g}(t,\eta,\gga) = -i2\pi  t \mathcal{D}_x^{\xi g}(t,\eta,\gga) - i2\pi  \gga \mathcal{D}_x^{\xi^2 g}(t,\eta,\gga) - \mathcal{D}_x^{\xi g'}(t,\eta,\gga) - \mathcal{D}_x^{g}(t,\eta,\gga).
\label{partialxieta}  
\end{align}
Using \eqref{partialeta} and \eqref{partialxieta}, one can calculate $\widehat{t}(t,\eta,\gga)$ and   $\widehat{r}(t,\eta,\gga)$ by the following formulas: 
\begin{eqnarray}
\label{def_gO1}
\widehat{t}(t,\eta,\gga)
&=& t+\frac{1}{2\pi}\operatorname{Im}\Big(
\frac{\mathcal{D}_x^{\xi^{2}g}\mathcal{D}_x^{g'}
      -\mathcal{D}_x^{\xi g}\mathcal{D}_x^{\xi g'}
      -\mathcal{D}_x^{\xi g}\mathcal{D}_x^{g}}
     {\mathcal{D}_x^{\xi^{2}g}\mathcal{D}_x^{g}
      -\mathcal{D}_x^{\xi g}\mathcal{D}_x^{\xi g}}
\Big), \\[4pt]
\label{def_gL1}
\widehat{r}(t,\eta,\gga)
&=& \gga+\frac{1}{2\pi}\operatorname{Im}\Big(
\frac{\mathcal{D}_x^{g}\mathcal{D}_x^{\xi g'}
      -\mathcal{D}_x^{\xi g}\mathcal{D}_x^{g'}
      +\mathcal{D}_x^{g}\mathcal{D}_x^{g}}
     {\mathcal{D}_x^{\xi^{2}g}\mathcal{D}_x^{g}
      -\mathcal{D}_x^{\xi g}\mathcal{D}_x^{\xi g}}
\Big).
\end{eqnarray}
When \(\gga = 0\), the FCT degenerates to the STFT. In this case, the  GD reference function \eqref{def_gO1} is simpler than the second-order GD estimator proposed in \cite{fourer2019second,he2020gaussian}. 
Furthermore, \cite{bao2021generalized,meignen2022analysis} introduced a high-order GD estimator in the time-frequency plane, which could conceivably be extended to formulate high-order GD and GDD reference functions (we provide the explicit mathematical expressions to Appendix A and omit further detailed discussion herein).

With  the GD and GDD reference functions obtained  above, we now define the 
time-reassigned synchrosqueezing frequency-domain chirplet transform (TSFCT) as follows.
\begin{mdef}\label{definition_TSFCT}
For a signal \( x(t) \in L^2(\mathbb{R}) \) and a chosen frequency-domain window function \( g(\xi)\), 
let $\mathcal{D}_x^g(t, \eta, \gga)$ denote its FCT with window $g(\xi)$, as defined in \eqref{defFCT_freq}.
With a threshold \( \epsilon > 0 \), using the GD and GDD reference functions \( \widehat{t}(t,\eta,\gga) \) and \( \widehat{r}(t,\eta,\gga) \) given in \eqref{def_gO1} and \eqref{def_gL1}, respectively, we define the TSFCT of a signal $x(t)$ as
\begin{align}
\mathbb{D}^g_x(\tau,\eta,u): = \iint\limits_{\{(t,\gga) : (t,\eta,\gga) \in E_\epsilon\}}
\mathcal{D}_x^g(t,\eta,\gga) 
\delta\bigl(u - \widehat{r}(t,\eta,\gga)\bigr) 
\delta\bigl(\tau - \widehat{t}(t,\eta,\gga)\bigr)  dt  d\gga,
\end{align}
where the integration domain \( E_\epsilon \) is defined in \eqref{region_E}.
\end{mdef}

\begin{remark}
    In Definition \ref{definition_TSFCT}, $g(\xi)$ is such a window function that 
\(g(\xi) \in L^2(\mathbb{R})\) 
and the FCTs \(\mathcal{D}_x^{g'}(t,\eta,\gga)\), \(\mathcal{D}_x^{\xi g}(t,\eta,\gga)\), and \(\mathcal{D}_x^{\xi^2 g}(t,\eta,\gga)\) are well-defined. 
\end{remark}

\subsection{Error analysis for  GD and GDD reference functions}
As derived in the above subsection, for a generalized linear chirp signal in the frequency domain given in \eqref{def_chirpf},  GD and GDD reference functions--$\widehat{t}(t,\eta,\gga)$ and $\widehat{r}(t,\eta,\gga)$--defined by  \eqref{def_gO1} and \eqref{def_gL1} satisfy 
$$
\widehat{t}(t,\eta,\gga)=\gth'(\eta), \quad  \widehat{r}(t,\eta,\gga)=\gth''(\eta). 
$$
This subsection analyzes the error bounds of the  GD and GDD reference functions for a general sginal given by \eqref{Frequency-domain_model}
 relative to the ground truth $\gth'_k(\eta)$ and $\gth''_k(\eta)$.
To this regard, we introduce the signal class $\mathcal{B}_{\epsilon_1,\epsilon_2}$ to characterize multicomponent signals with specific regularity conditions in the frequency domain.

\begin{mdef}[Class $\mathcal{B}_{\epsilon_1,\epsilon_2}$]
\label{def:multicomponent_class}
Let $\epsilon_1, \epsilon_2, \Delta_1, \Delta_2$ be  small positive constants. A multicomponent signal $x(t) = \sum_{k=1}^K x_k(t)$ with Fourier transform $\widehat{x}(\eta)= \sum_{k=1}^K \widehat{x}_k(t) = \sum_{k=1}^K B_k(\eta)e^{-i2\pi \theta_k(\eta)}$ is said to belong to the class $\mathcal{B}_{\epsilon_1,\epsilon_2}$ if the following conditions hold: for all $\eta \in \mathbb{R}$ and  $k = 1, \dots, K$,
\begin{itemize} 
    \item {Spectral amplitude condition}: $B_k(\eta) \in L^\infty(\mathbb{R})\cap C^1(\mathbb{R})$, with $B_k(\eta) > 0$ and $|B'_k(\eta)| \leq \epsilon_1$.  
    \item {Spectral phase condition}: $\theta_k(\eta) \in C^3(\mathbb{R})$, with $\theta_k'(\eta) > 0$, $\|\theta_k^{(j)}(\eta)\|_\infty < \infty$ for $j = 1, 2, 3$, and $|\theta_k^{(3)}(\eta)| \leq \epsilon_2$.
    \item {Separation condition}: For any $j \neq k$, either $|\theta_k'(\eta) - \theta_j'(\eta)| > 2\Delta_1$ or $|\theta_k''(\eta) - \theta_j''(\eta)| > 2\Delta_2$ holds.
\end{itemize}
\end{mdef}

For a window function $g$, we introduce the following notations.
\begin{align}
  \label{notation1}  &\mathcal{C}(g)(t,\gga) := \int_{\mathbb{R}} g(\xi) e^{-i 2 \pi \xi t} e^{-i \pi \gga \xi^2} d\xi,   \\
  \label{notation2}  &M(\eta) := \sum_{l=1}^{K} B_l(\eta), \quad    I_m := \int_{\mathbb{R}} |\tau^m g(\tau) d\tau|,\quad m =0, 1,\cdots.  
\end{align}
\begin{align}
 \label{piml}
     &  \Pi_{m,l}:= \epsilon_1  I_{m+1} + \epsilon_2 \frac{\pi}{3} B_l(\eta)  I_{m+3}.  \\
 \label{pim}
   & \Pi_{m}=\sum_{l=1}^{K}\Pi_{m,l}=\epsilon_1 K I_{m+1} + \epsilon_2\frac{\pi}{3} M(\eta) I_{m+3}.
\end{align}
Denote the region $Z_k$ as
\begin{equation}
\label{eq:Z_k_definition}   
Z_k := \left\{ (t, \eta, \gga) : \left|t - \theta'_k(\eta)\right| < \Delta_1 \text{ and } \left|\gga - \theta''_k(\eta)\right| < \Delta_2, \, \eta \in \mathbb{R} \right\}.
\end{equation}
For \(m=0,1,2\), let \(\Upsilon_{m,k}(\eta)\) be functions that satisfy  
\begin{align}\label{Upsilonmk}
   \sup_{(t,\eta,\gga)\notin Z_k} \bigl|\mathcal{C}(\eta^m g)\bigl(\theta_k'(\eta)-t,\theta_k''(\eta)-\gga\bigr)\bigr| \le \Upsilon_{m,k}(\eta). 
\end{align}
See Appendix \ref{section error functions} for the discussion on quantities $\Upsilon_{m,k}(\eta)$ when $g$ is the Gaussian function: 
\begin{align}
g_\sigma(\xi) := \frac{1}{\sigma \sqrt{2\pi}} e^{-\frac{\xi^2}{2\sigma^2}} \quad (\sigma > 0).  \label{Gaussian_function}
\end{align}
We also denote 
\begin{align}
    \gga_m&:=B_k(\eta)I_m + \sum_{l\neq k} B_l(\eta) \Upsilon_{m,l}(\eta)+\Pi_m, 
 \label{modulus_Cx}
\\
\label{Lambdamk}
\Lambda_{m,k} &:= \epsilon_1 K I_m + \epsilon_1 \epsilon_2 \pi K I_{m+3} + \epsilon_2 \pi M(\eta) I_{m+2}\\
&\quad + \sum_{l\neq k} 2\pi \Big( \left|\theta'_k(\eta) - \theta'_l(\eta)\right| \big( B_l(\eta) \Upsilon_{m,l}(\eta) + \Pi_{m,l} \big)+ \left|\theta''_k(\eta) - \theta''_l(\eta)\right| \big( B_l(\eta) \Upsilon_{m+1,l}(\eta) + \Pi_{m+1,l} \big) \Big)
 \nonumber
\end{align}

The next theorem provides error bounds for the  GD reference function \(\widehat{t}(t, \eta, \gga) \) and GDD reference function \(\widehat{r}(t, \eta, \gga) \) relative to their ground truth values  
\(\theta_k'(\eta)\) and \(\theta_k''(\eta)\).
\begin{theo}
\label{theorem_eta_lambda}
Let \(x(t) \in \mathcal{B}_{\epsilon_1, \epsilon_2}\) be a multicomponent signal with \(K\) modes for sufficiently small \(\epsilon_1, \epsilon_2 > 0\). 
Consider its FCT \(\mathcal{D}_x^{\xi^m g}(t,\eta,\gga)\) with window \(\xi^m g(\xi)\). 
Suppose \((t, \eta, \gga) \in Z_k\) and  the matrix \(E_0\) defined in \eqref{matrix_def} satisfies \(|\det(E_0)| > \epsilon_0^{-1}\). 
Then, the reference functions \(\widehat{t}(t, \eta, \gga)\) and \(\widehat{r}(t, \eta, \gga)\) in \eqref{def_gO1} and \eqref{def_gL1} satisfy:
\begin{align}
   \label{estimator_t_2} \left|\widehat{t}(t, \eta, \gga) - \theta_k'(\eta)\right| &\leq \frac{\epsilon_0}{2\pi} \left( \Lambda_{0,k} \gga_2 + \Lambda_{1,k} \gga_1 \right), \\
   \label{estimator_r_2} \left|\widehat{r}(t, \eta, \gga) - \theta_k''(\eta)\right| &\leq \frac{\epsilon_0}{2\pi} \left( \Lambda_{0,k}\gga_1  + \Lambda_{1,k}  \gga_0 \right),
\end{align}
where \(\gga_m\) and \(\Lambda_{m,k}\) are defined in \eqref{modulus_Cx} and \eqref{Lambdamk}, respectively.
\end{theo}

The proof of Theorem \ref{theorem_eta_lambda} is postponed to Appendix \ref{proof of theorem} and is based on two lemmas to be established therein.

\section{Frequency-domain signal separation operation scheme}
In this section, we propose a  frequency-domain group signal separation operation  (FGSSO) scheme based on the TSFCT. 
This scheme is specifically designed to efficiently recover individual modes from multicomponent signals exhibiting  intersecting GD curves. 

The three-dimensional CT- or WCT-based signal separation operation (SSO) scheme \cite{chui2021time,chui2023analysis} rely on accurate IF and chirprate estimation to extract crossing IF components.
However, the slow decay of CT and WCT along the chirprate dimension results in  residual interference from adjacent modes. The group SSO  algorithm proposed in \cite{li2022chirplet} successfully suppresses such inter-mode interference, outperforming conventional SSO methods.
Similarly, the FCT is limited by its slow decay along the GDD direction. To address this issue, we extend the group SSO framework to the TF-GDD space, thereby significantly enhancing the accuracy of mode retrieval. 
Furthermore, we provide a rigorous error analysis of the proposed FGSSO method, demonstrating that it depends on the accuracy of the estimated GD and GDD curves. 

Assuming $x(t) \in \mathcal{B}_{\epsilon_1,\epsilon_2}$ and that $\epsilon_1$ and $\epsilon_2$ are sufficiently small, we obtain from Lemma \ref{approxC_x} (see Appendix \ref{proof of theorem}) that
\begin{align}
\label{approx_multicomponent}  
\mathcal{D}_{x}^{g}(t,\eta,\gga) \approx \sum_{k=1}^{K} \widehat{x}_k(\eta) \mathcal{C}(g)(\theta_k'(\eta)-t, \theta_k''(\eta)-\gga), 
\end{align}
where $\mathcal{C}(g)$ is defined in \eqref{notation1}.
By applying the TSFCT to the multicomponent signal $x(t)$, we can generate a highly concentrated TF-GDD representation.  
This representation enables precise characterization of signal components and allows for the extraction of GD and GDD curves.  
The ridge curves $\check{\tau}_{k}(\eta)$ and $\check{\gga}_{k}(\eta)$, obtained via a three-dimensional ridge-extraction method \cite{zhu2020frequency,zhang2022two}, provide accurate approximations
\[
\check{\tau}_k(\eta)\approx\theta_k'(\eta),\qquad
\check{\gga}_k(\eta)\approx\theta_k''(\eta),\qquad
k=1,2,\dots,K.
\]

Then, when \( t = \check{\tau}_k(\eta) \) and \( \gga = \check{\gga}_k(\eta) \) for each \( k = 1, 2, \ldots, K \), Eq.~\eqref{approx_multicomponent} can be rewritten in matrix form as:
\begin{align}
	\begin{bmatrix}
		\mathcal{D}^{g}_x(\check{\tau}_1(\eta),\eta,\check{\gga}_1(\eta)) \\
		 \mathcal{D}^{g}_x(\check{\tau}_2(\eta),\eta,\check{\gga}_2(\eta)) \\
	\vdots \\
	   \mathcal{D}^{g}_x(\check{\tau}_K(\eta),\eta,\check{\gga}_K(\eta))
	\end{bmatrix}
	\approx
	\begin{bmatrix}
	a_{1,1} & a_{1,2} & \cdots & a_{1,K} \\
	a_{2,1} & a_{2,2} & \cdots & a_{2,K} \\
	\vdots & \vdots & \ddots & \vdots \\
	a_{K,1} & a_{K,2} & \cdots & a_{K,K}
	\end{bmatrix}
	\begin{bmatrix}
	\widehat{x}_1(\eta) \\
	\widehat{x}_2(\eta) \\
	\vdots \\
	\widehat{x}_K(\eta)
	\end{bmatrix}, \label{linear_system}
 \end{align}
where 
 \begin{align*}
  a_{k,l}:=\mathcal{C}(g)(\check{\tau}_l(\eta)-\check{\tau}_k(\eta),\check{\gga}_l(\eta)-\check{\gga}_k(\eta)). 
 \end{align*}
 By solving the linear system, the approximate solution \(\widetilde{\widehat{x}}_k(\eta)\) is given by
 \begin{equation}
\label{recover}
	 \begin{bmatrix}
	 \widetilde{\widehat{x}}_1(\eta) \\
	\widetilde{\widehat{x}}_2(\eta) \\
	\vdots \\
	\widetilde{\widehat{x}}_K(\eta)
	 \end{bmatrix}
	 =
	 \begin{bmatrix}
		 a_{1,1} & a_{1,2} & \cdots & a_{1,K} \\
		 a_{2,1} & a_{2,2} & \cdots & a_{2,K} \\
		 \vdots & \vdots & \ddots & \vdots \\
		 a_{K,1} & a_{K,2} & \cdots & a_{K,K}
	 \end{bmatrix}
	 ^{-1} 
	 \begin{bmatrix}
		 \mathcal{D}^{g}_x(\check{\tau}_1(\eta),\eta,\check{\gga}_1(\eta)) \\
		 \mathcal{D}^{g}_x(\check{\tau}_2(\eta),\eta,\check{\gga}_2(\eta)) \\
		 \vdots \\
		 \mathcal{D}^{g}_x(\check{\tau}_K(\eta),\eta,\check{\gga}_K(\eta))
	 \end{bmatrix}.
 \end{equation}
The matrix $A := [a_{k,l}]_{1 \leq k,l \leq K}$ is assumed to be nonsingular; otherwise, $A^{-1}$ refers to its pseudo-inverse.

In the next theorem,  we establish the approximation error between the reconstructed mode $\widetilde{\widehat{x}}_k(\eta)$ and the original mode's Fourier transform $\widehat{x}_k(\eta)$.
\begin{theo} \label{theorem_recover}
Let $\widetilde{\widehat{x}}_k(\eta), 1\le k\le K$ be the estimate of $\wh x_k(\eta)$ given by \eqref{recover}. Then 
\begin{align}
 \label{error_recov}    \left| \widetilde{\widehat{x}}_k(\eta)-\widehat{x}_k(\eta)\right| \leq \gO_0(\eta)\sum_{l=1}^{K} \left|b_{l,k}\right|,
\end{align}
and \(b_{l,k}\) denote the \((l,k)\)-th entry  of  the matrix \(A^{-1}\) defined in \eqref{recover}, and 
\begin{align}
\label{def_Oml}
\Omega_0(\eta):=\epsilon_1 K I_{1} + \epsilon_2\frac{\pi}{3} M(\eta) I_{3} + \sum_{l=1}^{K}  B_l(\eta) \left(2\pi \left|\check{\tau}_l(\eta)-\theta'_l(\eta)\right|I_1 + \pi \left|\check{\gga}_l(\eta)-\theta''_l(\eta)\right|I_2\right).
\end{align}
\end{theo}

\begin{proof} Observe that
\begin{align*}
	\begin{bmatrix}
		\sum_{l=1}^{K}a_{1,l}\widehat{x}_l(\eta) \\
		\sum_{l=1}^{K}a_{2,l}\widehat{x}_l(\eta) \\
	\vdots \\
	   \sum_{l=1}^{K}a_{K,l}\widehat{x}_l(\eta)
	\end{bmatrix}
	=
	\begin{bmatrix}
	a_{1,1} & a_{1,2} & \cdots & a_{1,K} \\
	a_{2,1} & a_{2,2} & \cdots & a_{2,K} \\
	\vdots & \vdots & \ddots & \vdots \\
	a_{K,1} & a_{K,2} & \cdots & a_{K,K}
	\end{bmatrix}
	\begin{bmatrix}
	\widehat{x}_1(\eta) \\
	\widehat{x}_2(\eta) \\
	\vdots \\
	\widehat{x}_K(\eta)
	\end{bmatrix}.
 \end{align*}
Combining this with \eqref{linear_system}, we obtain
 \begin{equation}
	 \begin{bmatrix}
	 \widetilde{\widehat{x}}_1(\eta)-\widehat{x}_1(\eta) \\
	\widetilde{\widehat{x}}_2(\eta)-\widehat{x}_2(\eta) \\
	\vdots \\
	\widetilde{\widehat{x}}_K(\eta)-\widehat{x}_K(\eta)
	 \end{bmatrix}
	 =
	 \begin{bmatrix}
		 a_{1,1} & a_{1,2} & \cdots & a_{1,K} \\
		 a_{2,1} & a_{2,2} & \cdots & a_{2,K} \\
		 \vdots & \vdots & \ddots & \vdots \\
		 a_{K,1} & a_{K,2} & \cdots & a_{K,K}
	 \end{bmatrix}
	 ^{-1} 
	 \begin{bmatrix}
		 \mathcal{D}^{g}_x(\check{\tau}_1(\eta),\eta,\check{\gga}_1(\eta))-\sum_{l=1}^{K}a_{1,l}\widehat{x}_l(\eta) \\
		 \mathcal{D}^{g}_x(\check{\tau}_2(\eta),\eta,\check{\gga}_2(\eta))-\sum_{l=1}^{K}a_{2,l}\widehat{x}_l(\eta)  \\
		 \vdots \\
		 \mathcal{D}^{g}_x(\check{\tau}_K(\eta),\eta,\check{\gga}_K(\eta))-\sum_{l=1}^{K}a_{K,l}\widehat{x}_l(\eta) 
	 \end{bmatrix}. \label{recover_error_bound}
 \end{equation}
As for the k-th row of the system \eqref{recover_error_bound}, we have
\begin{align*}
    &\mathcal{D}^{g}_x(\check{\tau}_k(\eta),\eta,\check{\gga}_k(\eta)) - \sum_{l=1}^{K} a_{k,l} \widehat{x}_l(\eta) \\
    &= \mathcal{D}^{g}_x(\check{\tau}_k(\eta),\eta,\check{\gga}_k(\eta)) - \sum_{l=1}^{K} \wh{x}_l(\eta) \mathcal{C}(g)\bigl(\theta'_l(\eta) - \check{\tau}_k(\eta), \theta''_l(\eta) - \check{\gga}_k(\eta)\bigr) \\
    &\quad + \sum_{l=1}^{K} \wh{x}_l(\eta) \Bigl( \mathcal{C}(g)\bigl(\theta'_l(\eta) - \check{\tau}_k(\eta), \theta''_l(\eta) - \check{\gga}_k(\eta)\bigr)-\mathcal{C}(g)\bigl(\check{\tau}_l(\eta) - \check{\tau}_k(\eta), \check{\gga}_l(\eta) - \check{\gga}_k(\eta)\bigr) \Bigr).
\end{align*}
Note that,
\begin{align*}
     &\left|\mathcal{C}(g)\bigl(\theta'_l(\eta) - \check{\tau}_k(\eta), \theta''_l(\eta) - \check{\gga}_k(\eta)\bigr)-\mathcal{C}(g)\bigl(\check{\tau}_l(\eta) - \check{\tau}_k(\eta), \check{\gga}_l(\eta) - \check{\gga}_k(\eta)\bigr)\right| \\
     &=\left|\int_{\mathbb{R}} g(\xi) e^{i 2 \pi \xi \check{\tau}_k(\eta)} e^{i \pi \xi^2 \check{\gga}_k(\eta)} \left( e^{-i 2 \pi \xi \check{\tau}_l(\eta)}e^{-i \pi \xi^2 \check{\gga}_l(\eta)}-e^{-i 2 \pi \xi \theta'_l(\eta)}e^{-i \pi \xi^2 \theta''_l(\eta)}\right) d\xi\right|,    \\
     &\leq 2\pi \left|\check{\tau}_l(\eta)-\theta'_l(\eta)\right|I_1 + \pi \left|\check{\gga}_l(\eta)-\theta''_l(\eta)\right|I_2.
\end{align*}
Additionally, according to \eqref{approx_allk} in Appendix \ref{proof of theorem}, when \(m=0 , t =\check{\tau}_k(\eta), \gamma=\check{\gga}_k(\eta)\), then
\[\left| \mathcal{D}^{g}_x(\check{\tau}_k(\eta),\eta,\check{\gga}_k(\eta)) - \sum_{l=1}^{K} \wh{x}_l(\eta) \mathcal{C}(g)\bigl(\theta'_l(\eta) - \check{\tau}_k(\eta), \theta''_l(\eta) - \check{\gga}_k(\eta)\bigr) \right|\leq \Pi_0=\epsilon_1 K I_{1} + \epsilon_2\frac{\pi}{3} M(\eta) I_{3},\]
Furthermore, we have the bound
\begin{align*}
    &\left|\mathcal{D}^{g}_x(\check{\tau}_k(\eta),\eta,\check{\gga}_k(\eta)) 
     - \sum_{l=1}^{K} a_{k,l} \widehat{x_l}(\eta)\right| \nonumber \\
    &\leq \left| \mathcal{D}^{g}_x(\check{\tau}_k(\eta),\eta,\check{\gga}_k(\eta)) 
     - \sum_{l=1}^{K} \wh{x}_l(\eta) \mathcal{C}(g)\bigl(\theta'_l(\eta) - \check{\tau}_k(\eta), 
     \theta''_l(\eta) - \check{\gga}_k(\eta)\bigr) \right| \nonumber \\
    &\quad + \sum_{l=1}^{K} \left| \wh{x}_l(\eta)\right|   \left|\Bigl( \mathcal{C}(g)\bigl(\theta'_l(\eta) - \check{\tau}_k(\eta), 
     \theta''_l(\eta) - \check{\gga}_k(\eta)\bigr) \right. \left. - \mathcal{C}(g)\bigl(\check{\tau}_l(\eta) - \check{\tau}_k(\eta), 
     \check{\gga}_l(\eta) - \check{\gga}_k(\eta)\bigr) \Bigr)\right| \nonumber \\
    &= \Omega_0(\eta),
\end{align*}
where $\Omega_0(\eta)$ is defined by \eqref{def_Oml}.  

Thus, from  Eq.~\eqref{recover_error_bound}, 
\begin{align*}
   \left| \widetilde{\widehat{x}}_k(\eta)-\widehat{x}_k(\eta)\right| \leq \Omega_0(\eta) \sum_{l=1}^{K} \left|b_{l,k}\right| ,
\end{align*}
as desired. 
\end{proof}



When considering equivalence under scalar scaling, time-domain shifting, and frequency-domain modulation, the Gaussian function is the only window function that achieves optimal time-frequency resolution.
In this paper, we will use the Gaussian function \(g_\sigma(\xi)\) given in \eqref{Gaussian_function} 
as the window function of FCT. In this case,  the \(a_{k,l}\) in \eqref{recover} is given by
\begin{align*}
  a_{k,l}= \frac{1}{\sqrt{1+i2\pi \sigma^2  (\check{\gga}_k(\eta)-\check{\gga}_l(\eta))}}      e^{\frac{-2\pi^2 \sigma^2 (\check{\tau}_k(\eta)-\check{\tau}_l(\eta))^2}{1+i2\pi  \sigma^2 (\check{\gga}_k(\eta)-\check{\gga}_l(\eta)) }}.
 \end{align*} 
Furthermore, the reconstruction error in \eqref{error_recov} is directly related to the values of $I_m$. 
As \eqref{Im} (in Appendix \ref{section error functions}) shows, a smaller $\sigma$
 leads to a reduction in the values of \(I_m\)
 and thus a smaller reconstruction error.

\section{Experimental results} 
\subsection{Implementation}

The  parameter $\sigma$ in the Gaussian window function $g_\sigma(\xi)$ affects the energy concentration of the TF-GDD representation. In this paper, we employ Rényi entropy to determine the value of the parameter \( \sigma \).
For the FCT, we define the Rényi entropy as:
\begin{equation}
\label{def_renyi_entropy_spec}
E_{\sigma} := \frac{1}{1-\ell} \log_2 \left(
\frac{\iiint_{\mathbb{R}^3} |\mathcal{D}^{g_\sigma}_x(t,\eta,\lambda)|^{2\ell} dt\, d\eta\, d\lambda}
{\left(\iiint_{\mathbb{R}^3} |\mathcal{D}^{g_\sigma}_x(t,\eta,\lambda)|^2 dt\, d\eta\, d\lambda\right)^\ell}
\right),
\end{equation}
where $ \ell > 1$ controls the entropy sensitivity. 
The Rényi entropy provides a quantitative measure of concentration for TF--GDD representations, where lower entropy values correspond to more concentrated energy distributions. 
To determine the optimal window parameter $\sigma$, we define the optimization problem:
\begin{align}
  \sigma_{\mathrm{opt}} := \operatorname*{arg\,min}_{\sigma > 0} \, E_{\sigma}(\ell), 
\label{opt_sigma}
\end{align}
where $E_{\sigma}(\ell)$ represents the Rényi entropy of order $\ell$.
The order parameter \(\ell\) is set to 2.5 throughout this paper.

Once the optimal window parameter $\sigma_{\mathrm{opt}}$ is obtained, we focus on the implementation of the FCT. 
Since the Gaussian function \(g_\sigma(\xi)\) is real-valued, the FCT admits the following equivalent forms:
\begin{align}
\mathcal{D}_{x}^{g_\sigma}(t, \eta, \gga):&= \int_{\mathbb{R}} \widehat{x}(\xi + \eta) g_\sigma(\xi)  e^{i2\pi \xi t} e^{i\pi  \gga \xi^2} \, d\xi \nonumber \\
                             &= \int_{\mathbb{R}} \widehat{x}(\xi + \eta) e^{i2\pi \xi t}   \overline{ g_\sigma(\xi)  e^{-i\pi  \gga \xi^2}} \, d\xi  \nonumber\\
                             &= \int_{\mathbb{R}} {x}(t-\tau) e^{i2\pi \eta (t-\tau)}   \overline{ \mathcal{C}(g_\sigma)(\tau,\gga)} \, d\tau  \nonumber\\
                              &= \int_{\mathbb{R}} {x}(\tau) e^{i2\pi \eta \tau}   \overline{ \mathcal{C}(g_\sigma)(t-\tau,\gga)} \, d\tau,  \label{defFCT_implement}
\end{align}
where \(\mathcal{C}(g)(t,\gga) \) was defined in \eqref{notation1}.

This work adopts  \eqref{defFCT_implement} for the FCT implementation.
Once  the corresponding window functions are obtained, the implementation of the FCT is identical to that of the STFT. 
Assume that the input signal \( x(t) \) is uniformly discretized at the points
\begin{align*}
t_n =  n\Delta t, \quad n = 0, 1, \cdots, N-1,
\end{align*}
where \( \Delta t \) represents the sampling interval (time step), and \( N \) stands for the total number of sampling points. For the window function, its corresponding time variable is defined as \( \tau_k = k \Delta t \), where \( k=0,1 \dots, N-1 \).
The parameter $\gga$ is defined over the interval $[-R_0, R_0]$, and is discretized into $L = 2\lfloor N/2 \rfloor + 1$ equidistant points. For $l = 1,\dots,L$, the $l$-th sample point and corresponding sampling interval are given by:
\begin{align*}
	\gga_l &= -R_0 + (l-1) \Delta \gga, \quad \Delta \gga = {2R_0}/{(L-1)},
\end{align*}
This discretization ensures that $\gga_1 = -R_0$, $\gga_L = R_0$, and provides a uniformly spaced grid across the interval \([-R_0, R_0]\).
Additionally, choosing $L$ as an odd number  guarantees symmetric sampling about $\gga = 0$.
The frequency variable $\eta$ is discretized into $N$ points with resolution $\Delta\eta = 1/(N\Delta t)$:
\begin{equation*}
\eta_j = \begin{cases} 
j\Delta\eta, & 0 \leq j \leq \lfloor N/2 \rfloor \\
(j-N)\Delta\eta, & \lfloor N/2 \rfloor < j < N
\end{cases}
\end{equation*}
For the non-negative frequency indices ($0 \leq j \leq \lfloor N/2 \rfloor$), the product $\eta_j\tau_k$ simplifies to $k j/N$. This leads to the discrete FCT:
\begin{align}
\mathcal{D}_x^{g_\sigma}(t_n,\eta_j,\gga_l) &= \sum_{k=0}^{N-1} x(\tau_k)\text{conj}(\mathcal{C}(g_\sigma)(t_n-\tau_k,\gga_l))e^{-i2\pi \eta_j\tau_k} \nonumber \\
&= \sum_{k=0}^{N-1} x(\tau_k)\text{conj}(\mathcal{C}(g_\sigma)(t_n-\tau_k,\gga_l))e^{-i2\pi kj/N}
\label{eq:fct_fft}
\end{align}
The second form in \eqref{eq:fct_fft}  reduces computational complexity by leveraging FFT compatibility.
The total computational cost, amounting to \( LN\log_2 N \) operations, is  equivalent to that of the CT.

For $\mathcal{D}_{x}^{\xi^n g_\sigma}(t, \eta, \gga)$ with $n = 0, 1, 2$, we have:
\begin{align*}
\mathcal{D}_x^{\xi^n g_\sigma}(t,\eta,\gga) = \int_\mathbb{R} x(\tau) e^{-i2\pi  \eta \tau} \overline{ \mathcal{C}(\xi^n g_\sigma)(t-\tau,\gga)}  d\tau,
\end{align*}
where the corresponding window functions $\mathcal{C}(\xi^n g_\sigma)(t,\gga)$ are provided in Appendix \ref{section error functions} (specifically, see Eqs.~\eqref{Cg}, \eqref{Cxig}, and \eqref{Cxi2g}).

Building upon the transform $\mathcal{D}_{x}^{\xi^n g}(t, \eta, \gga)$, we first construct the discrete time-reassigned operators $\widehat{t}(t_n,\eta_j,\gga_l)$ and $\widehat{r}(t_n,\eta_j,\gga_l)$ using the relations defined in \eqref{def_gO1} and \eqref{def_gL1}. 
This allows us to obtain the discrete  synchrosqueezing representation $\mathbb{D}^g_x(\tau_p,\eta_j,u_q)$
\begin{align}
    \label{eq: synchrosqueezing_transform}
    \mathbb{D}^g_x(\tau_p,\eta_j, u_q) = \sum_{(n,l) \in O_{\epsilon}} \mathcal{D}_x^{g}(t_n,\eta_j,\gga_l),
\end{align}
where the selection set $O_{\epsilon}$ is defined by:
\begin{equation}
O_{\epsilon} := \left\{ (n, l) : 
\begin{aligned}
&|\widehat{t}(t_n,\eta_j,\gga_l) -\tau_p| \leq \tfrac{1}{2} \Delta t, \quad |\widehat{r}(t_n,\eta_j,\gga_l)- u_q| \leq \tfrac{1}{2} \Delta\gga, \\
&|\det(E_0)(t_n,\eta_j,\gga_l)| > \epsilon, 
\end{aligned}
\right\},
\end{equation}
where \( \det(E_0)(t_n,\eta_j,\gga_l) \) is the discrete representation of \( \det(E_0)(t,\eta,\gga) \), defined in \eqref{matrix_def}, and \( \epsilon \) denotes a predefined threshold.

\subsection{Numerical validation}

In this part, we will present numerical examples to verify the effectiveness of the proposed TSFCT.
First, we consider a multicomponent signal $x(t)$ comprising two generalized  frequency-domain chirp signals whose GD curves intersect in the time-frequency plane. The frequency-domain representation is given by:
\begin{equation}
\widehat{x}(\eta) = e^{-0.00002(\eta-256)^2} e^{-i2\pi  (0.0003\eta^2+0.1\eta)} + e^{-0.00003(\eta-256)^2} e^{-i2\pi  (0.0002\eta^2+0.356\eta)} \label{equation:example_signal}
\end{equation}
with $\eta \in [0, 512)$ Hz and $t \in [0, 0.5)$ s.
Fig.~\ref{figure:GDs and GDDs of signal $x(t)$} displays the actual GD and GDD profiles of signal $x(t)$.
A key feature is the intersection of the GD trajectories of the two modes  at \((t, \eta) = (0.2536\ \text{s}, 256\ \text{Hz})\) in the time-frequency plane.

\begin{figure}[H]
    \centering
    \begin{tabular}{cc}
        \subfloat[GDs of the signal $x(t)$]{%
            \includegraphics[width=0.32\textwidth]{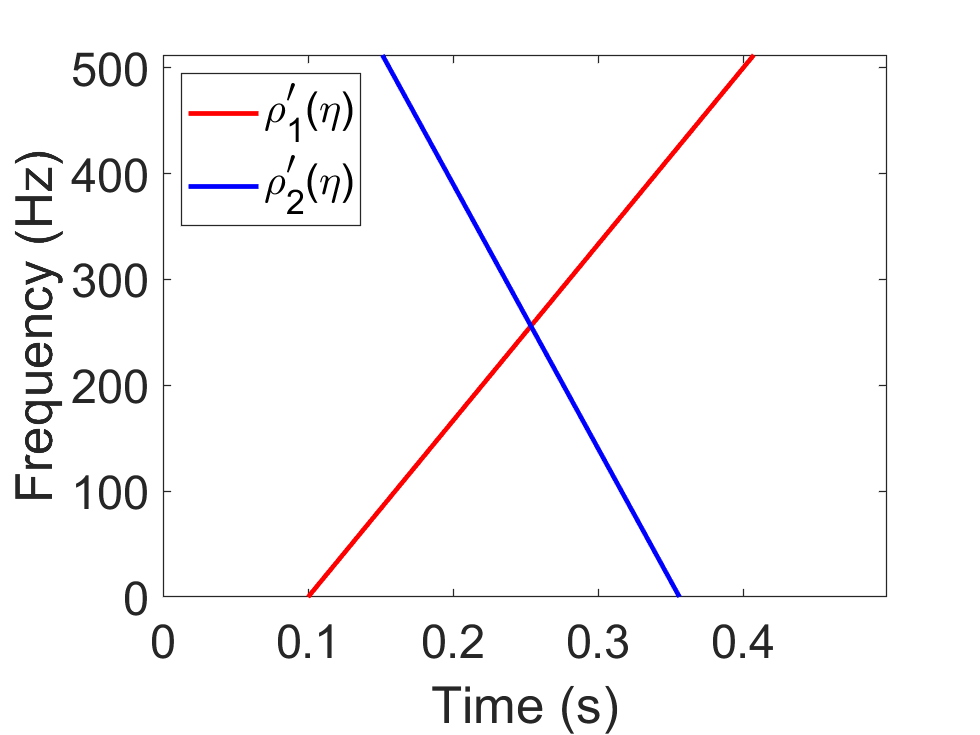}%
        } &
        \subfloat[GDDs of the signal $x(t)$]{%
            \includegraphics[width=0.32\textwidth]{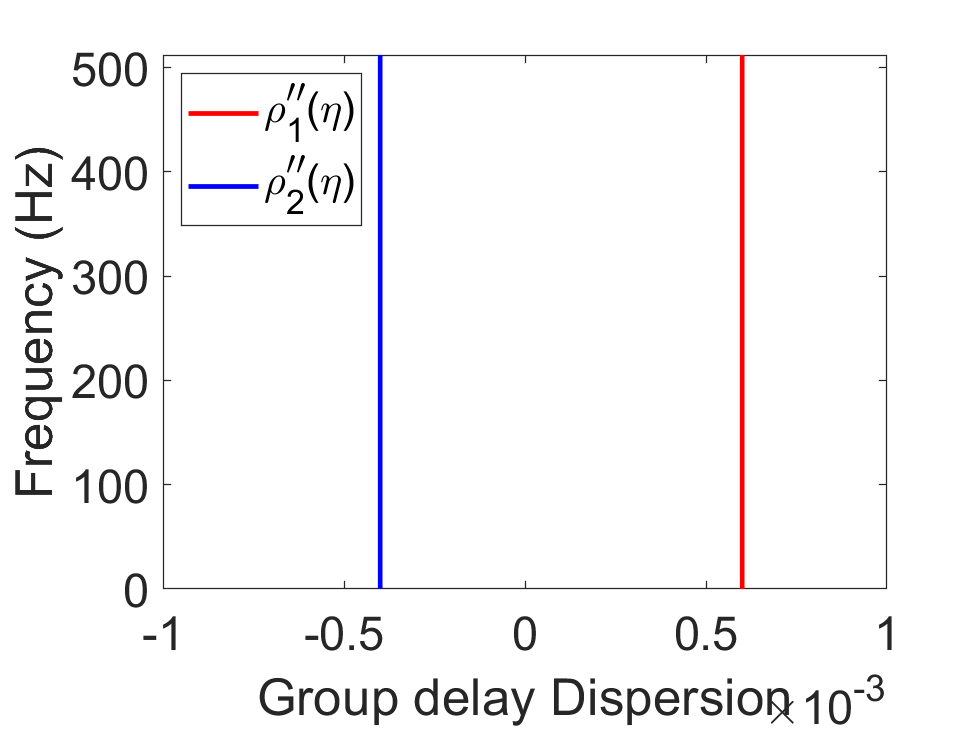}%
        } 
    \end{tabular}
    \caption{\small GDs and GDDs of the signal $\wh{x}(t)$}
    \label{figure:GDs and GDDs of   signal $x(t)$}
\end{figure}

\begin{figure}[H]
    \centering
    \begin{tabular}{cccc}
         \begin{subfigure}[t]{0.22\textwidth}
            \centering
            \resizebox{\linewidth}{!}{\includegraphics{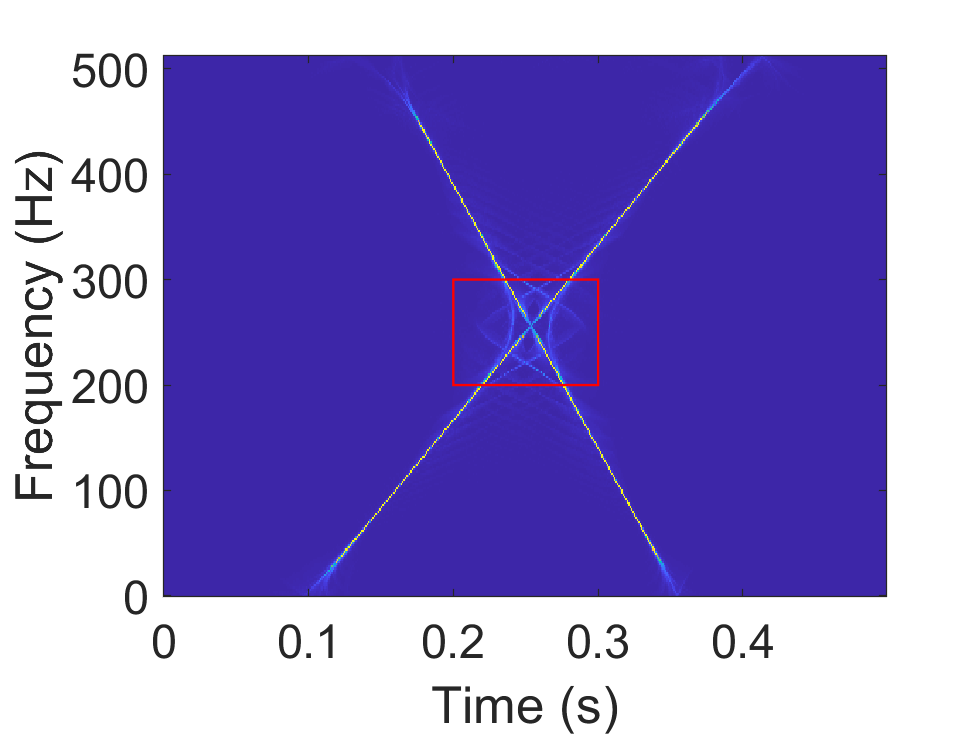}} 
        \end{subfigure} &
        \begin{subfigure}[t]{0.22\textwidth}
            \centering
            \resizebox{\linewidth}{!}{\includegraphics{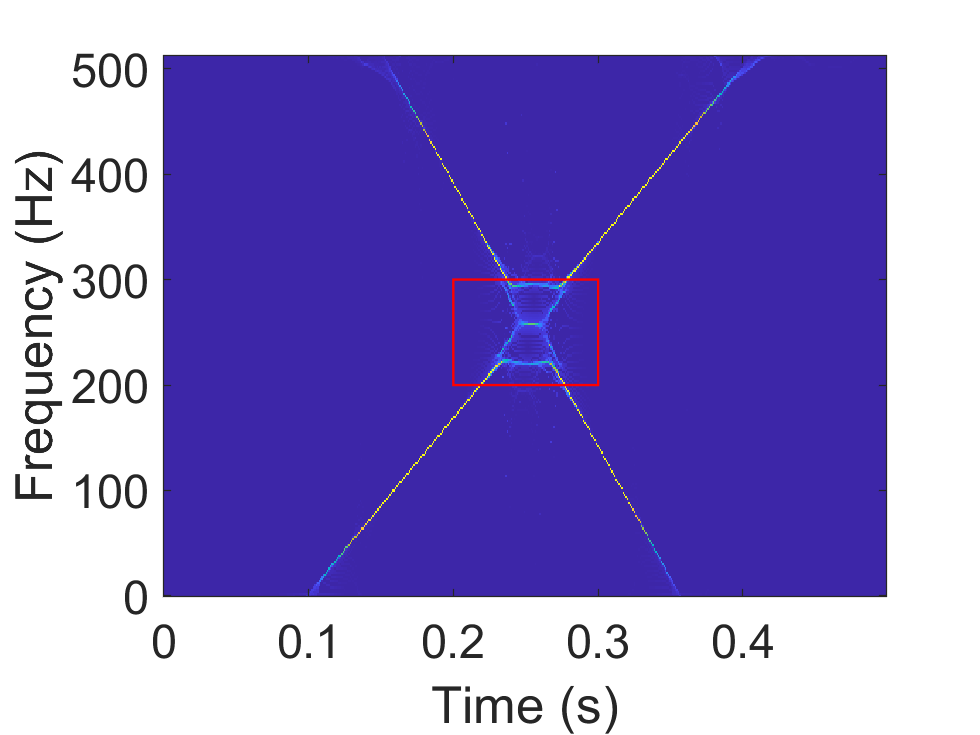}}
        \end{subfigure} &
        \begin{subfigure}[t]{0.22\textwidth}
            \centering
            \resizebox{\linewidth}{!}{\includegraphics{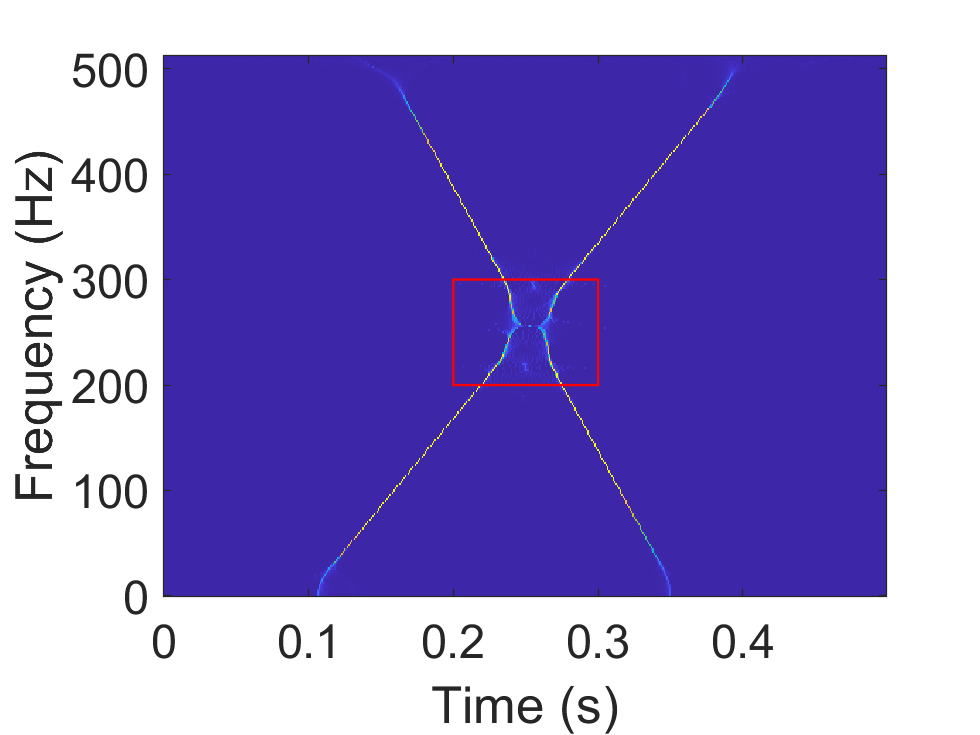}}
        \end{subfigure} &
        \begin{subfigure}[t]{0.22\textwidth}
            \centering
            \resizebox{\linewidth}{!}{\includegraphics{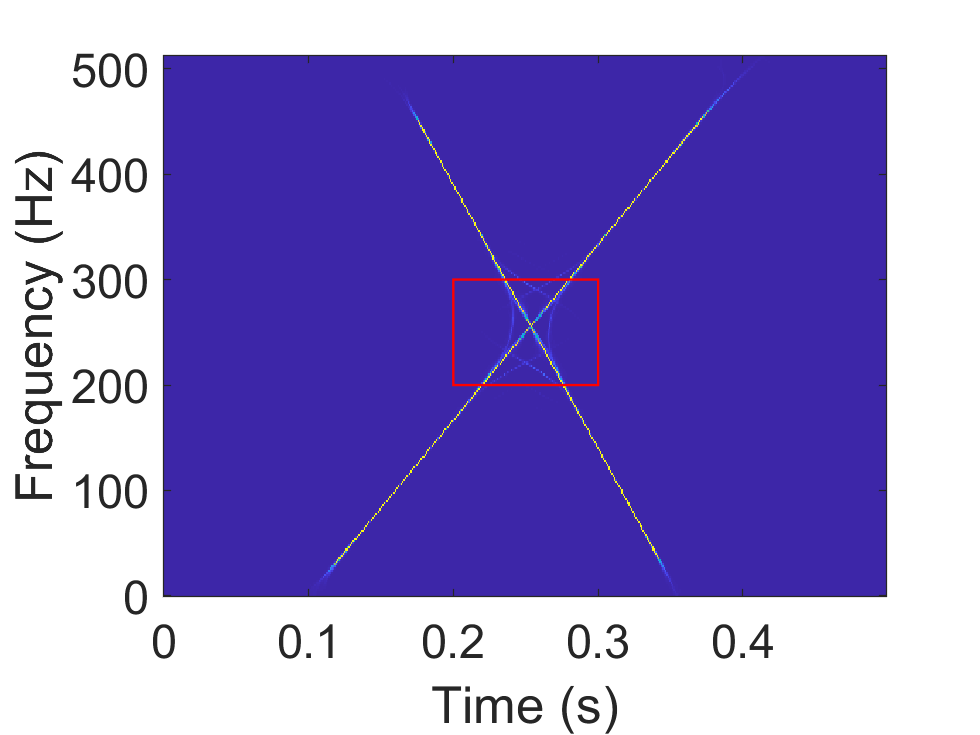}}
        \end{subfigure}  \\
         \begin{subfigure}[t]{0.22\textwidth}
            \centering
            \resizebox{\linewidth}{!}{\includegraphics{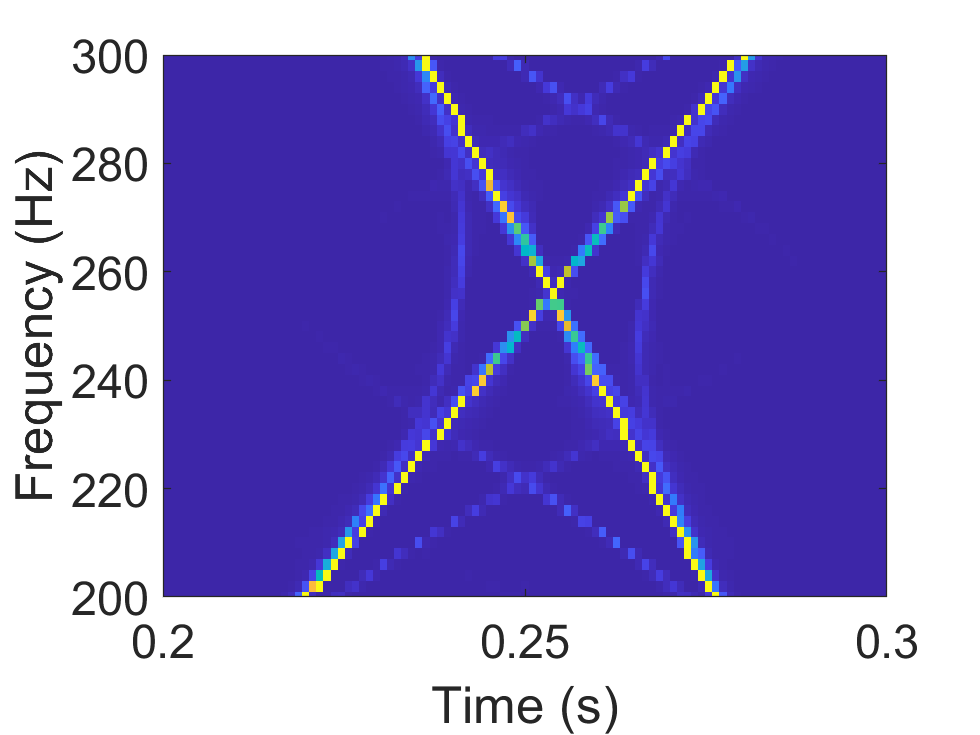}}    
        \end{subfigure}&
        \begin{subfigure}[t]{0.22\textwidth}
            \centering
            \resizebox{\linewidth}{!}{\includegraphics{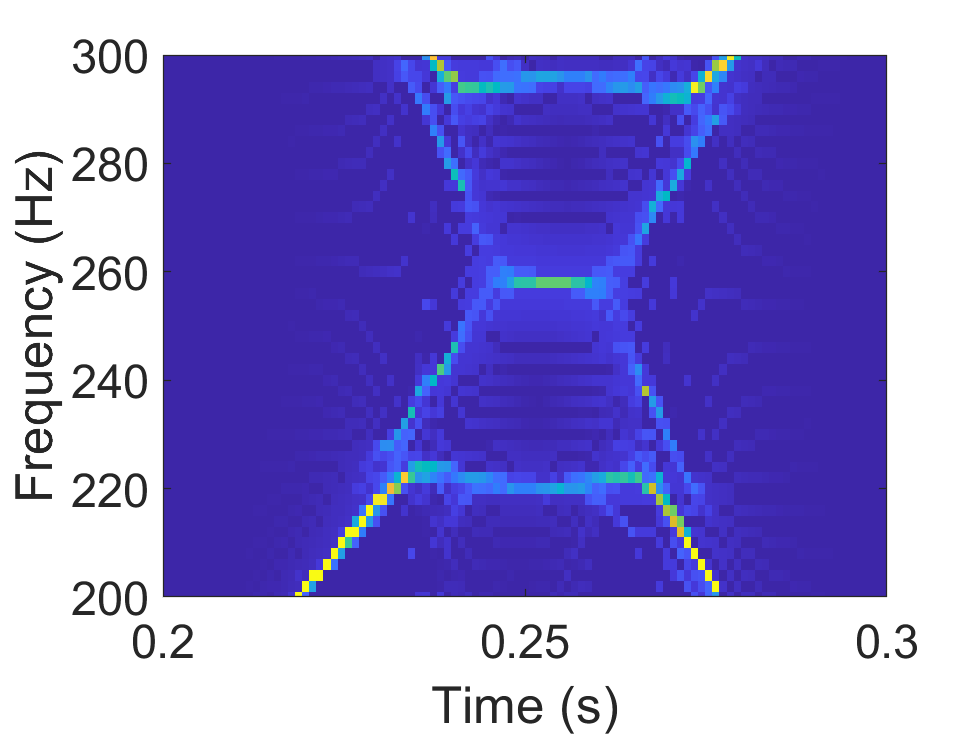}} 
        \end{subfigure} &
        \begin{subfigure}[t]{0.22\textwidth}
            \centering
            \resizebox{\linewidth}{!}{\includegraphics{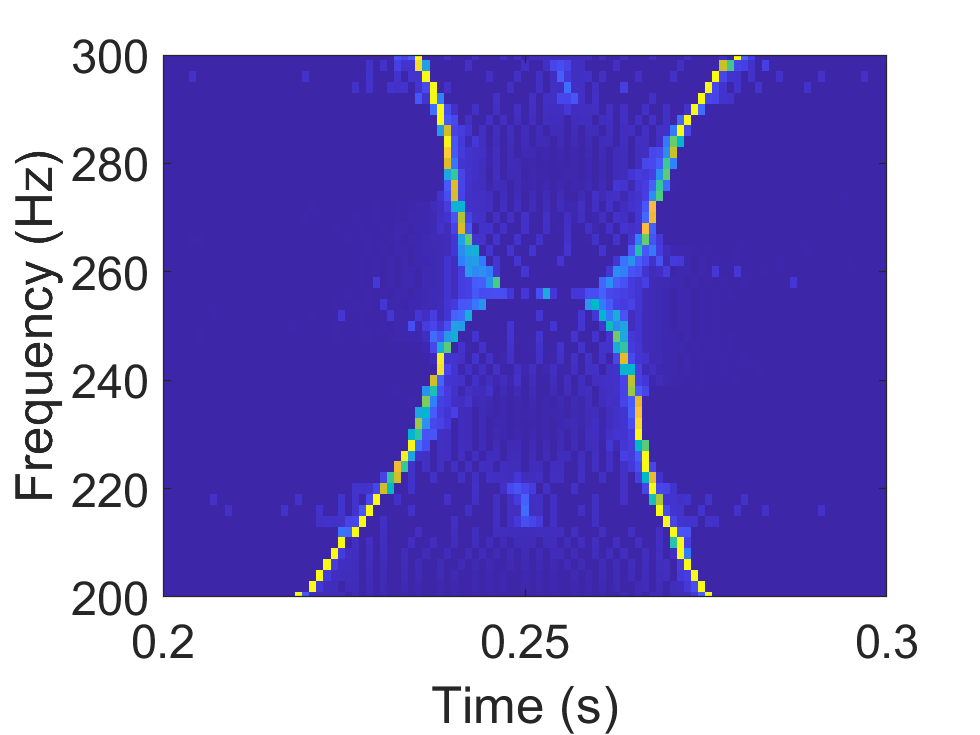}}    
        \end{subfigure} &
        \begin{subfigure}[t]{0.22\textwidth}
            \centering
            \resizebox{\linewidth}{!}{\includegraphics{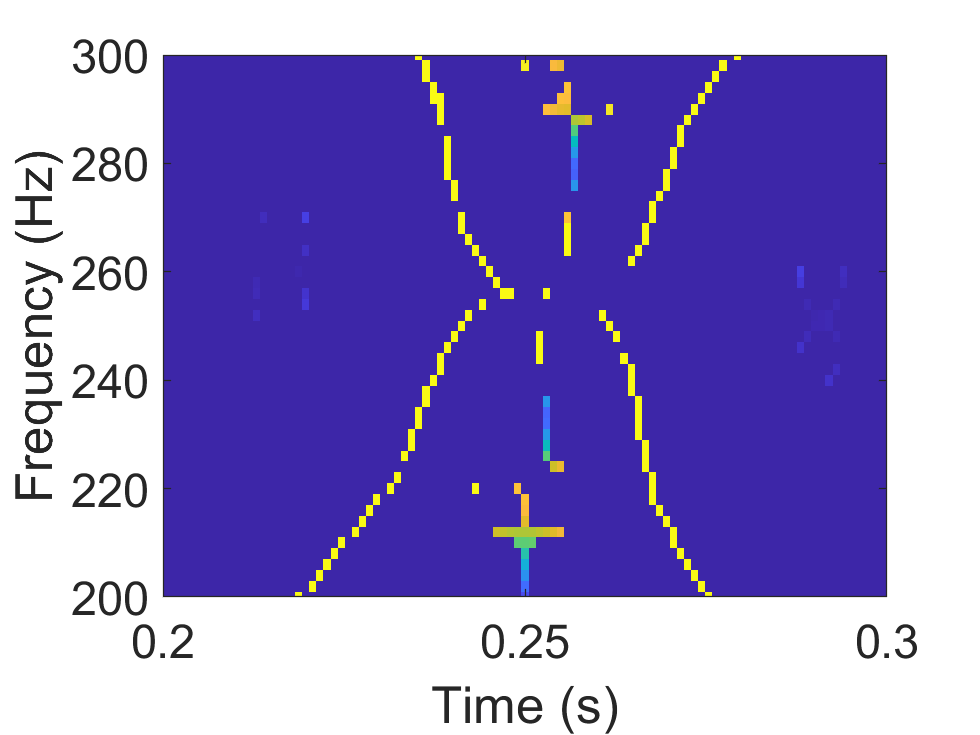}}    
        \end{subfigure}  \\        
    \end{tabular}
\caption{\small TFRs of $x(t)$. 
First row (from left to right): TSFCT, SST-2nd, TSST-2nd, TET-2nd; 
Second row: Corresponding local zoomed views of the TFRs in the first row.}
\end{figure}

Applying the Rényi entropy minimization method to  \eqref{opt_sigma}, the TSFCT optimal parameter $\sigma = 25$.
 For comparison purposes, we  also present the TFRs of several advanced time-frequency methods in Fig.~2,
including the second-order synchrosqueezing transform (SST-2nd) \cite{oberlin2017second},
the second-order time-reassigned synchrosqueezing transform (TSST-2nd) \cite{he2020gaussian}, and the second-order transient-extracting transform (TET-2nd) \cite{yu2021second}. Here we consider the TFR $\mathfrak{T}_{x}^{g}(\tau, \eta)$ generated by the TSFCT to be the following quantity obtained by integrating the squared modulus of the TF-GDD representation along the group delay dimension:
\begin{equation}
\mathfrak{T}_{x}^{g}(\tau, \eta) := \int_{\mathbb{R}} \left|\mathbb{D}_x^{g}(\tau,\eta,u) \right|^2 du.
\label{eq:tfr_definition}
\end{equation}

Most time-frequency methods operating in the time-frequency plane suffer from significant energy dispersion and spectral overlap near the intersection point.
In contrast, although the \(\mathfrak{T}_{x}^{g}(\tau, \eta)\) generated by the TSFCT shows some blurring around the intersection, 
it successfully captures the signal's time-frequency characteristics without apparent aliasing.


\begin{figure}[H]
    \centering
    \begin{tabular}{cccc}
        \begin{subfigure}[t]{0.22\textwidth}
            \centering
            \resizebox{\linewidth}{!}{\includegraphics{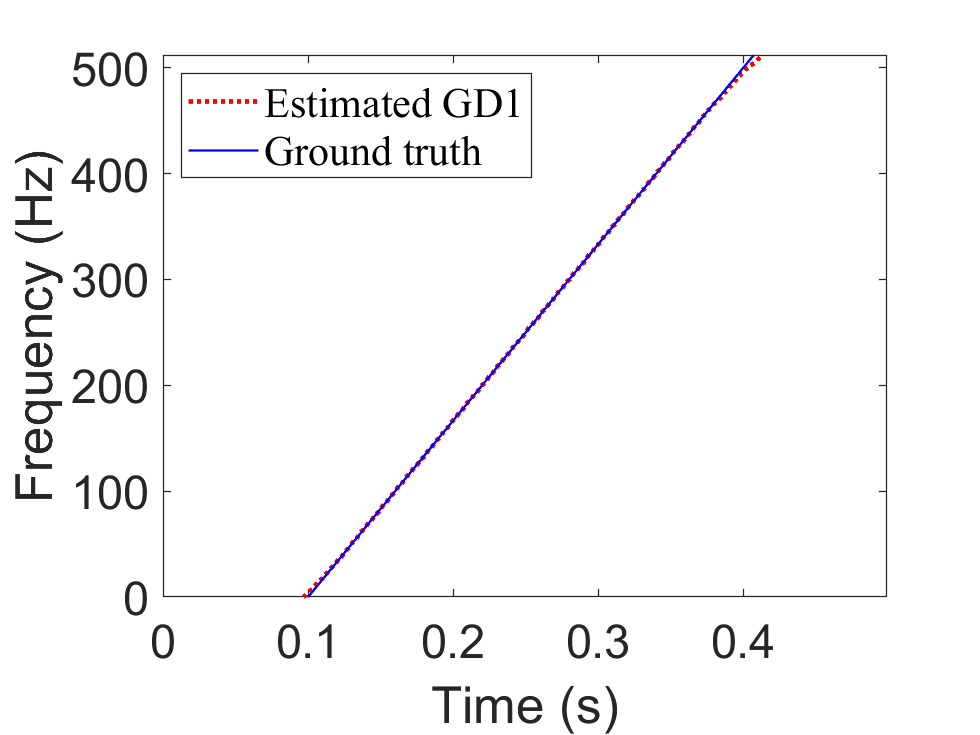}}
        \end{subfigure} &
        \begin{subfigure}[t]{0.22\textwidth}
            \centering
            \resizebox{\linewidth}{!}{\includegraphics{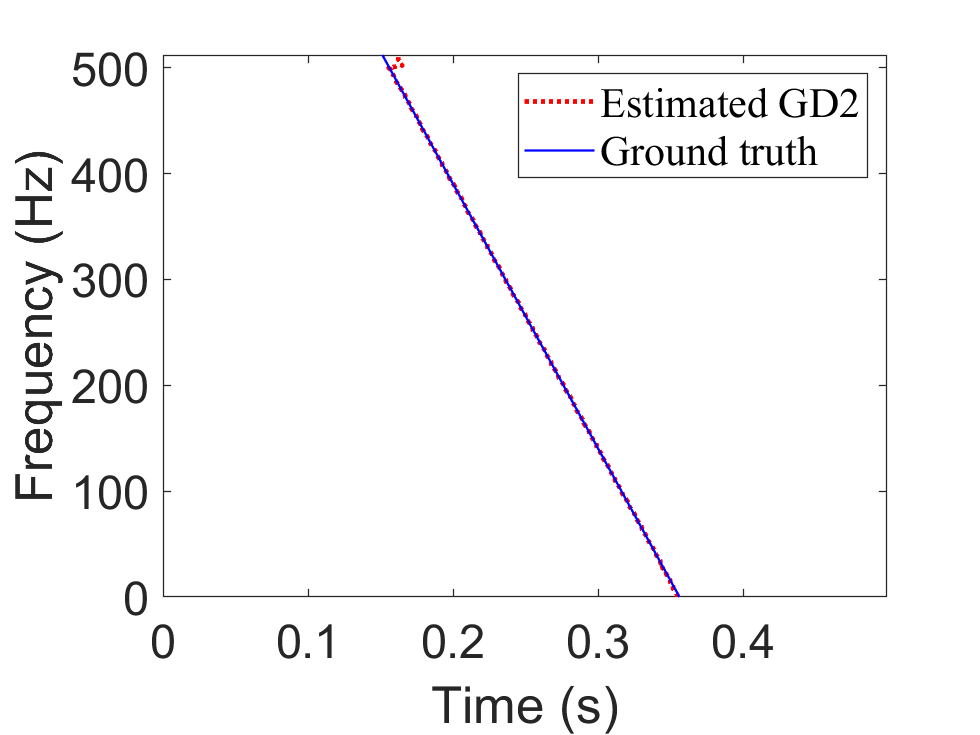}}
        \end{subfigure} &
        \begin{subfigure}[t]{0.22\textwidth}
            \centering
            \resizebox{\linewidth}{!}{\includegraphics{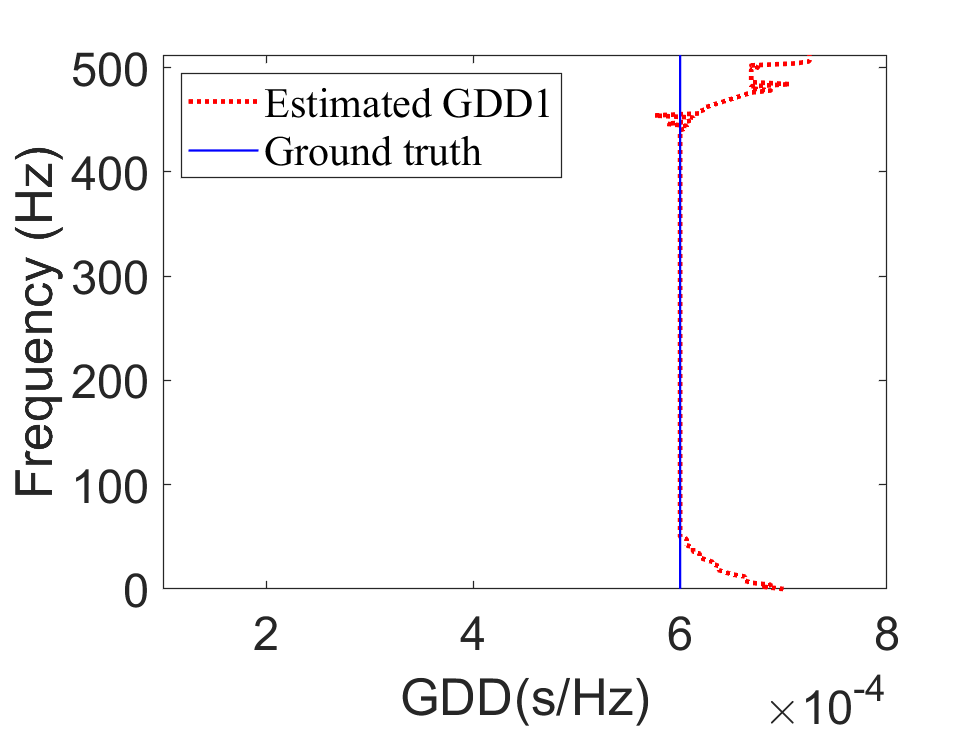}}
        \end{subfigure} &
        \begin{subfigure}[t]{0.22\textwidth}
            \centering
            \resizebox{\linewidth}{!}{\includegraphics{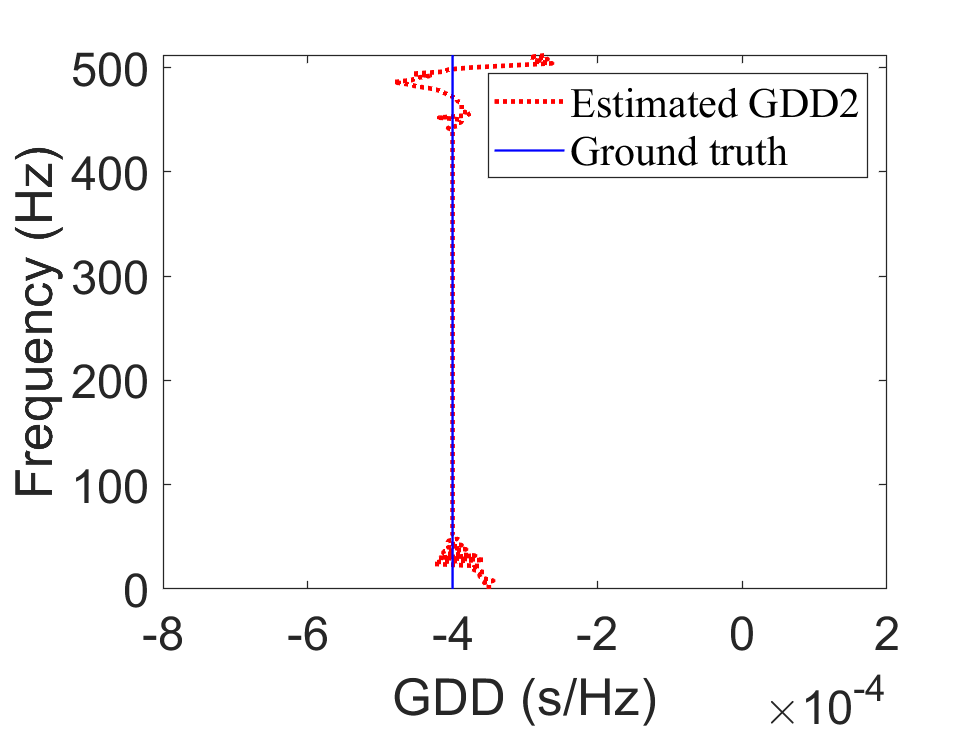}} 
        \end{subfigure} \\   
    \end{tabular}
\caption{GD and GDD estimation of $x(t)$ by  TSFCT.}
\label{figure:gdd_estimation}
\end{figure}


The TSFCT accurately captures both the GDs and GDDs of $x(t)$, except at the boundaries (see Fig.~\ref{figure:gdd_estimation}). These estimates are then used to reconstruct the components via the FGSSO scheme. The resulting recovery errors for both the retrieved modes $\widehat{x}_{1}(\eta)$, $\widehat{x}_{2}(\eta)$ and their time-domain counterparts $x_{1}(t)$, $x_{2}(t)$ are presented in Fig.~\ref{figure:recover_errors_of_modes$x(t)$}.
The FGSSO scheme successfully recovers the frequency-domain signals $\widehat{y}_1(\eta)$ and $\widehat{y}_2(\eta)$, except near the boundaries. However, these boundary errors may propagate through the inverse Fourier transform, leading to minor errors in the reconstructed time-domain signals.

\begin{figure}[H]
    \centering
    \begin{subfigure}[t]{0.22\textwidth}
        \centering
        \includegraphics[width=\linewidth]{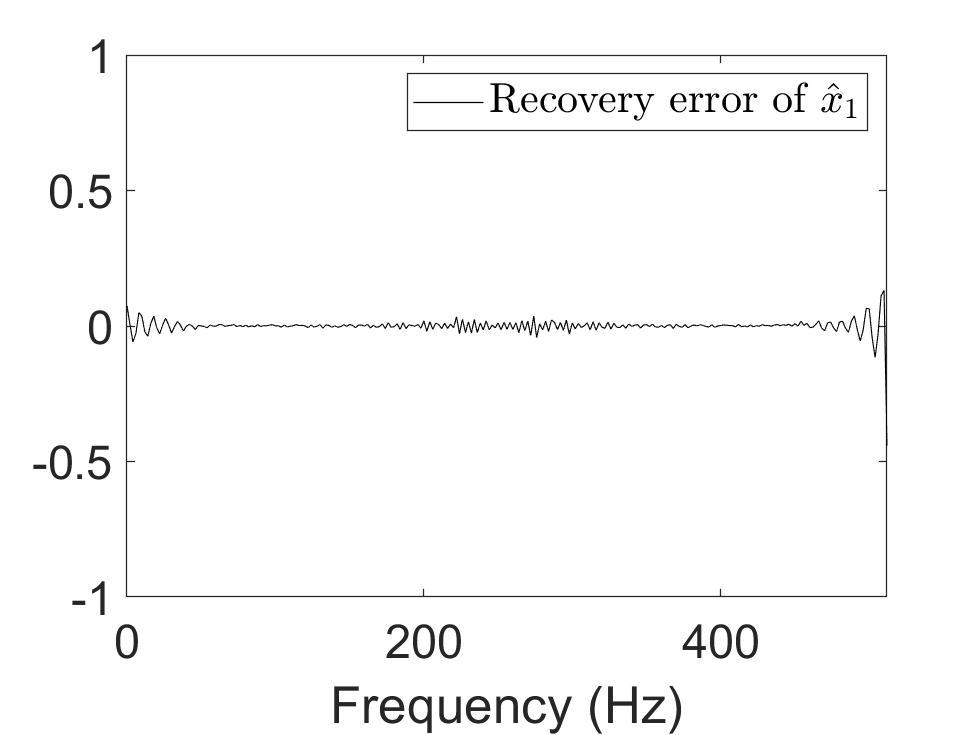}
        \caption{}
    \end{subfigure}
    \hfill
    \begin{subfigure}[t]{0.22\textwidth}
        \centering
        \includegraphics[width=\linewidth]{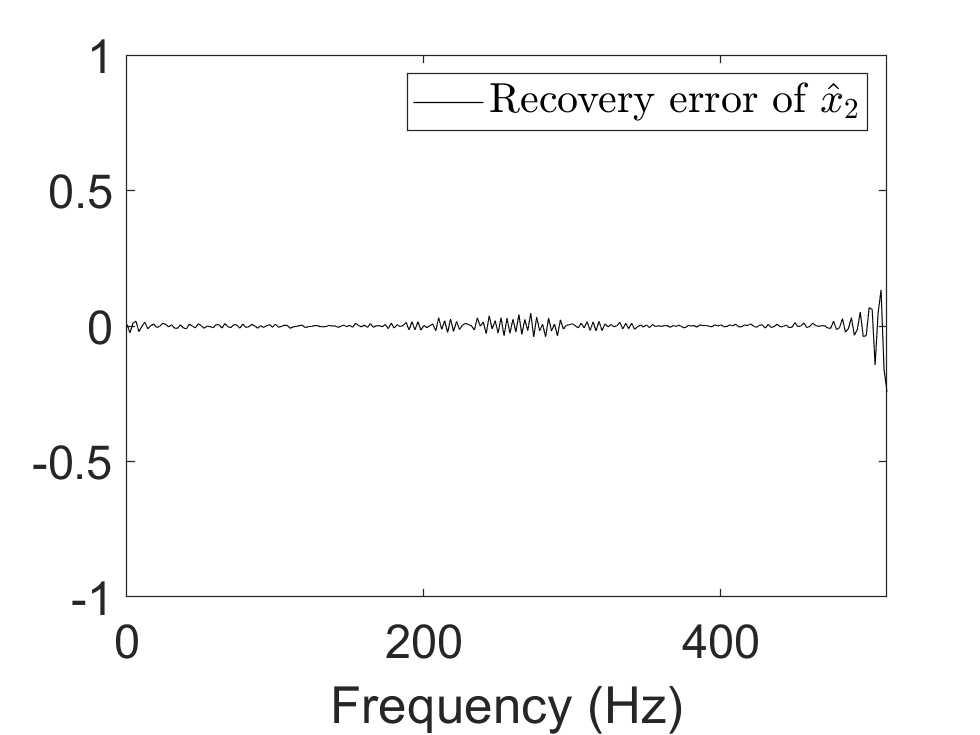}
        \caption{}
    \end{subfigure}
    \hfill
    \begin{subfigure}[t]{0.22\textwidth}
        \centering
        \includegraphics[width=\linewidth]{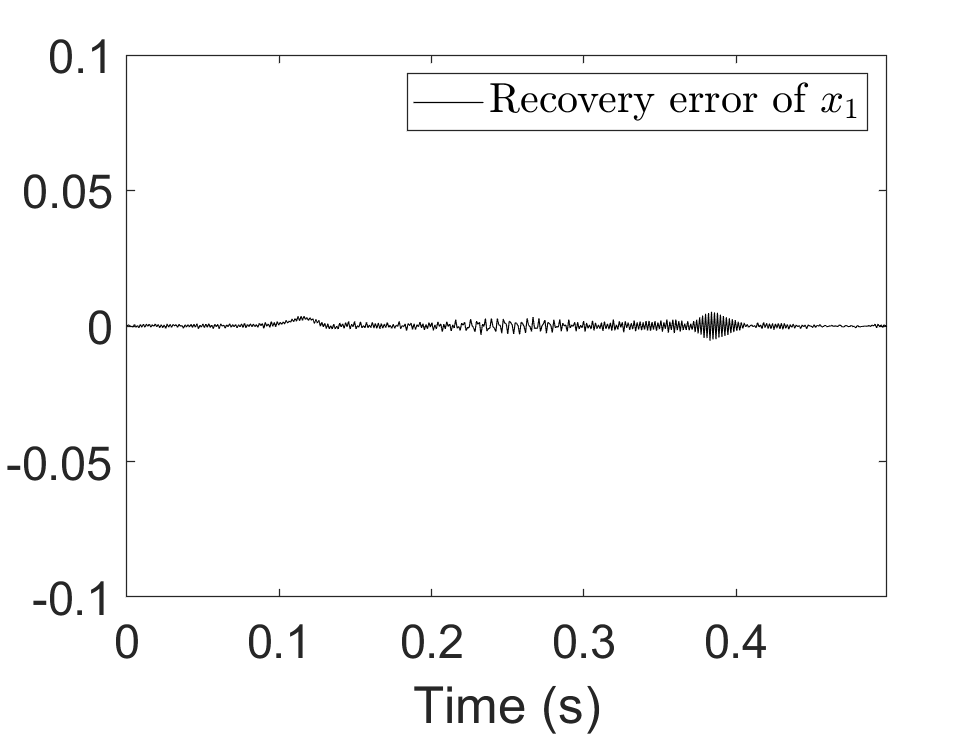}
        \caption{}
    \end{subfigure}
    \hfill
    \begin{subfigure}[t]{0.22\textwidth}
        \centering
        \includegraphics[width=\linewidth]{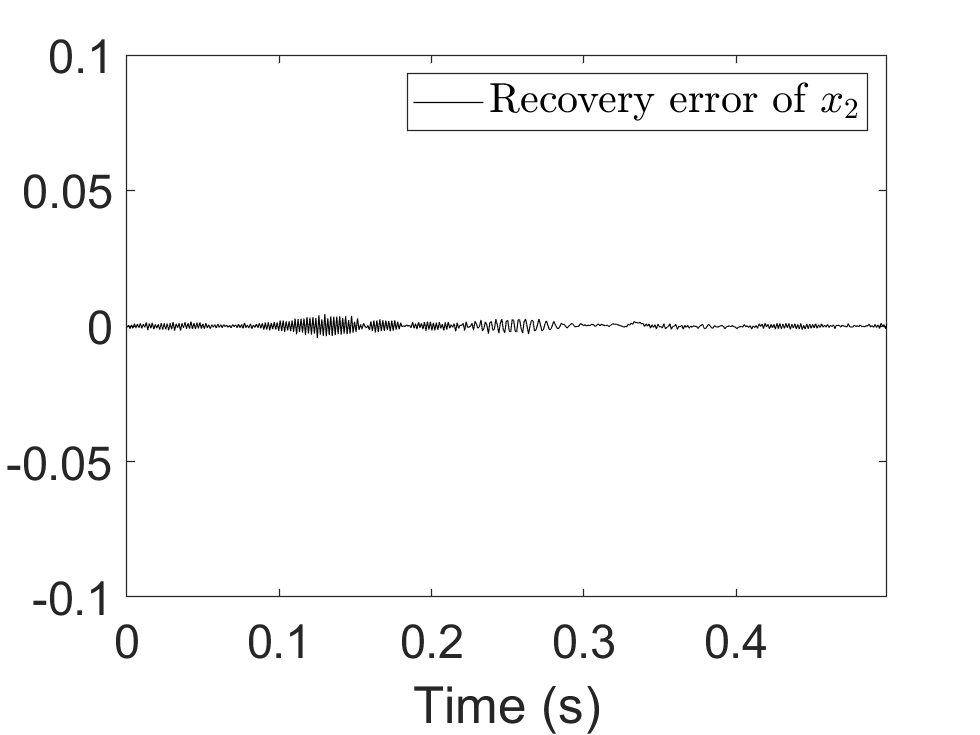}
        \caption{}
    \end{subfigure}
    \caption{Recovery errors(real part) of $y(t)$ using the frequency-domain SSO algorithm. (a) Recovery error of $\widehat{x}_1(\eta)$; (b) recovery error of $\widehat{x}_2(\eta)$; (c) recovery error of $x_1(t)$; (d) recovery error of $x_2(t)$.}
    \label{figure:recover_errors_of_modes$x(t)$}
\end{figure}

Moreover, the reconstruction error is directly related to the inverse coefficient matrix $A^{-1}$ (defined in \eqref{recover}), as established in Theorem~\ref{theorem_recover}. 
This dependency motivates the analysis of the infinity norm $\|A^{-1}\|_\infty$. 
To further understand the stability of the linear system in \eqref{recover}, we also examine 
the 2-norm condition number \cite{golub2013matrix} $\kappa_2(A) = \|A\|_2 \|A^{-1}\|_2$.
These two quantities are shown in Fig.~\ref{fig:inf_norm_inverse_matrix}.

\begin{figure}[H]
    \centering
    \begin{tabular}{cc}
        \includegraphics[width=0.32\textwidth]{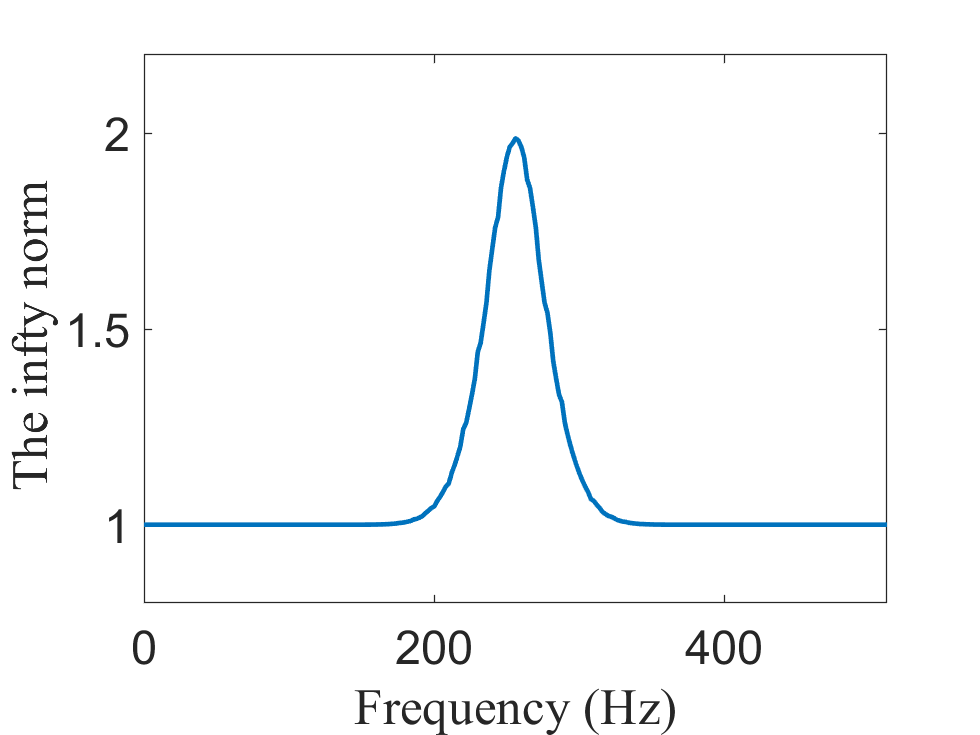} &
        \includegraphics[width=0.32\textwidth]{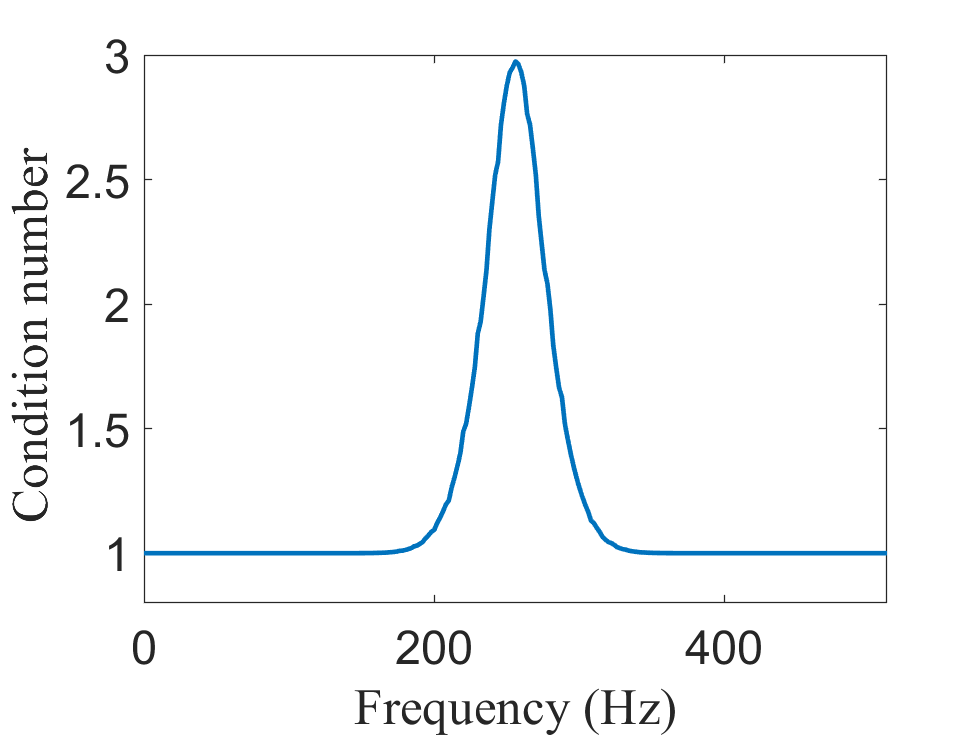}
    \end{tabular}
    \caption{The infinity norm of the inverse coefficient matrix (left) and 2-norm  condition number (right).}
    \label{fig:inf_norm_inverse_matrix}
\end{figure}
It is evident that the infinity norm of the inverse coefficient matrix is significantly smaller when $\eta$ is far from the frequency at which the GDs intersect ($\eta = 256$ Hz). As $\eta$ approaches this intersection frequency, the norm increases a little bit. 
This trend is also reflected in Panels (a) and (b) of Fig.~\ref{figure:recover_errors_of_modes$x(t)$}, where the reconstruction error peaks at the intersection point (except at the boundaries).
Besides, the condition number $\kappa_2(A)$ also exhibits a similar pattern, indicating that the linear system in \eqref{recover} becomes not as stable  as $\eta$ approaches the intersection frequency, but overall the linear system is quite stable since the condition number is quite small for any $\eta$.

To further validate our method, we consider a more complex scenario involving a multicomponent signal $y(t)$ with crossing modes that possess distinct GD and GDD properties. 
Let $y(t)$ be a signal whose Fourier transform is given by:
\begin{equation}
\widehat{y}(\eta) = \widehat{y}_1(\eta) + \widehat{y}_2(\eta) = e^{-0.00032(\eta-256)^2} e^{-i2\pi \theta_1(\eta)} + e^{-0.00025(\eta-256)^2} e^{-i2\pi \theta_2(\eta)},
\end{equation}
where $\eta \in [0, 512)$ Hz and $t \in [0, 0.5)$ seconds.
The phase spectrum  functions are defined as:
\begin{align*}
\theta_1(\eta) = -\frac{51.2}{\pi} \sin\left(\frac{\pi \eta}{256}\right) + 0.25\eta, \quad
\theta_2(\eta) = \frac{51.2}{\pi} \cos\left(\frac{\pi \eta}{256}\right) + 0.25\eta.
\end{align*}
The GDs given by the first derivatives, are:
\begin{align*}
\theta'_1(\eta) = -0.2 \cos\left(\frac{\pi \eta}{256}\right) + 0.25, \quad
\theta'_2(\eta) = 0.2 \cos\left(\frac{\pi \eta}{256}\right) + 0.25.
\end{align*}
The GDDs given by the second derivatives, are:
\begin{align*}
\theta''_1(\eta) = \frac{\pi}{1280} \sin\left(\frac{\pi \eta}{256}\right), \quad
\theta''_2(\eta) = -\frac{\pi}{1280} \sin\left(\frac{\pi \eta}{256}\right).
\end{align*}
The GD curves of the two modes intersect at $t = 0.25$ s, while their GDD curves crossing at zero, as shown in Fig.~\ref{figure:GDs and GDDs of signal $y(t)$}.
\begin{figure}[H]
    \centering
    \begin{tabular}{cc}
        \subfloat[GDs of the signal $y(t)$]{%
            \includegraphics[width=0.32\textwidth]{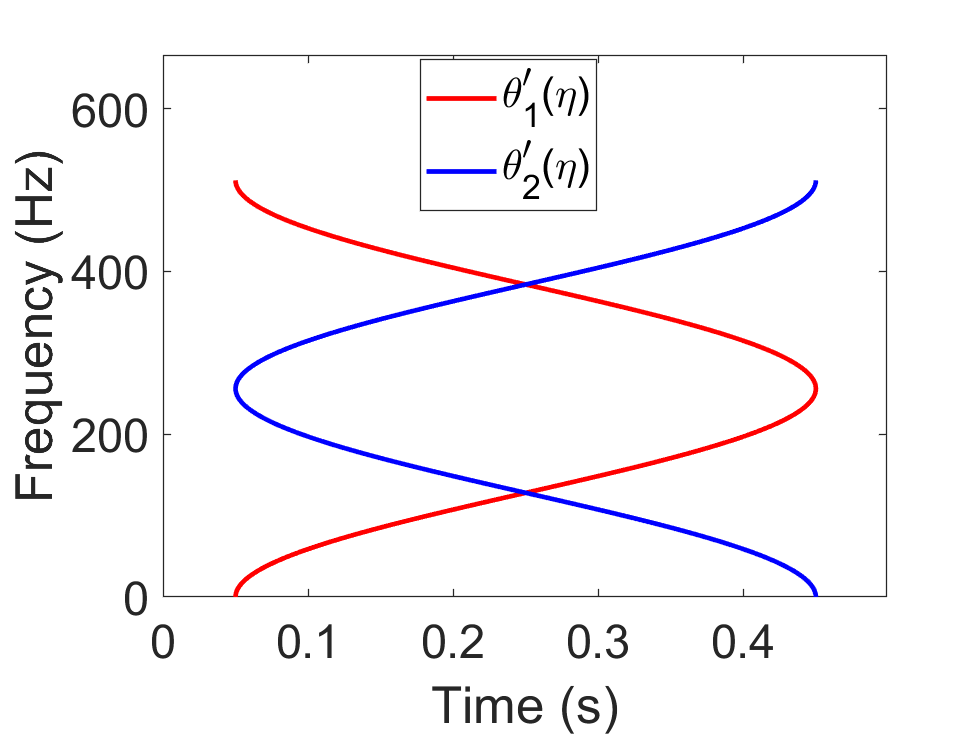}%
        } &
        \subfloat[GDDs of the signal $y(t)$]{%
            \includegraphics[width=0.32\textwidth]{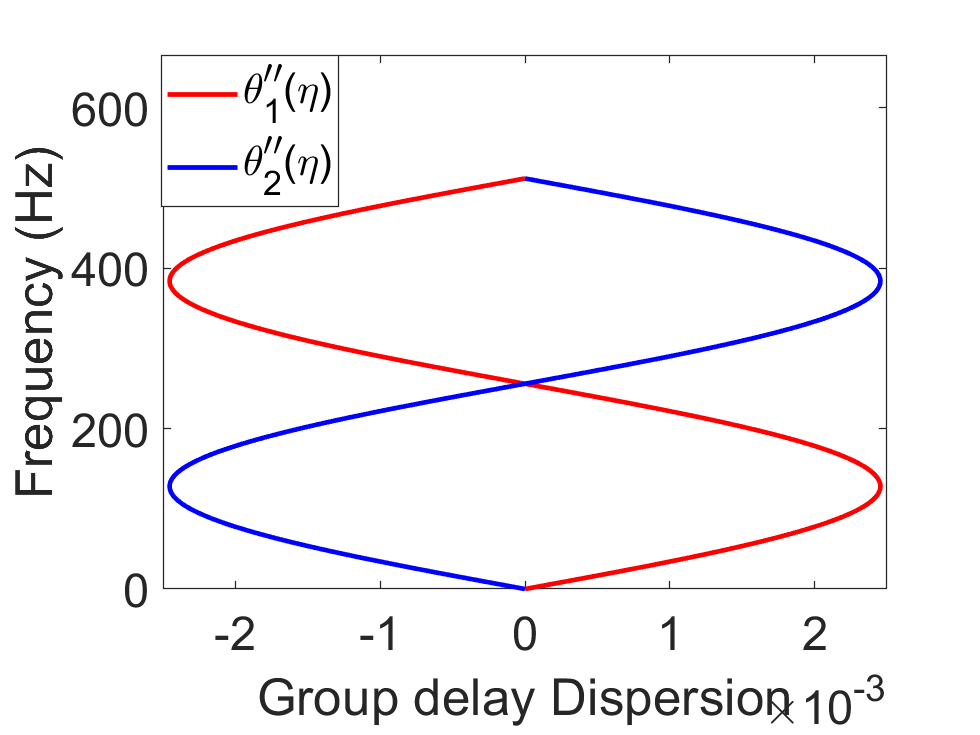}%
        } \\
    \end{tabular}
    \caption{\small GDs and GDDs of the signal $\wh{y}(t)$}
    \label{figure:GDs and GDDs of   signal $y(t)$}
\end{figure}

Indeed, as a three-dimensional TFC representation, the synchrosqueezing chirplet transform (SCT)~\cite{chen2023disentangling} could also achieve satisfactory GDD estimation for the signal in \eqref{equation:example_signal}; however, it fails to accurately estimate GDD values for the signal $y(t)$.
The key issue arises when $\eta = 256$, at this point $\theta''_1(\eta) = \theta''_2(\eta) = 0$, corresponding to infinite instantaneous chirprates.
In practical applications, \(\lambda\) is usually set within a reasonable range; that is, there is no such a \(\lambda\) that exactly matches the current chirprate.
The CT-based  and WCT-based methods, despite generating well-resolved TFC representations, exhibit the same limitation. 
In contrast, the FCT and TSFCT methods remain unaffected by this constraint.

Next, we elaborate on this point in greater detail. As noted in the introduction, a significant amount of recent literature has focused on methods based on the CT; moreover, these CT-based approaches can be naturally extended within the FCT framework. Therefore, we adopt the SCT as the primary benchmark for evaluating the proposed TSFCT.
To ensure a fair comparison, all methods utilize the same Gaussian window function as given in \eqref{Gaussian_function}, with the window parameter $\sigma$ optimized via Rényi entropy minimization. Specifically, during this minimization process (as defined in \eqref{def_renyi_entropy_spec}), the SCT replaces the TF-GDD representation $\mathcal{D}_{x}^{g}(t,\eta,\lambda)$ with ${Q}_{x}^{g}(t,\eta,\lambda)$. 
For consistency, the order parameter of the Rényi entropy for CT is also set to $\ell = 2.5$, and the frequency and chirprate bin sizes are kept identical with FCT.
The optimal $\sigma$ values obtained from the Rényi entropy minimization are 17.1 for the TSFCT and 0.007 for the SCT.

Fig.~\ref{cross-sections of $y(t)$} displays the cross-sections of $\left|\mathcal{D}_{x}^{g}(t, \eta, \gga)\right|$  and $\left|{Q}_{x}^{g}(t, \eta, \gl)\right|$  at the fixed frequency $\eta = 256$ Hz, where the GDDs of both signal modes vanish, i.e., $\theta''_1(\eta) = \theta''_2(\eta) = 0$.
Panel (a) shows that the FCT representation exhibits two distinct peaks around $(0.05, 0)$ and $(0.45, 0)$. In contrast, Panel (b) reveals that the CT produces two broad energy bands centered at $t = 0.05$ and $t = 0.45$, making it difficult to accurately identify the true chirprate locations.
\begin{figure}[H]
    \centering
    \begin{tabular}{cc}
        \subfloat[FCT \(\left|\mathcal{D}_{x}^{g}(t, 256, \gga)\right|\)]
        {%
            \includegraphics[width=0.32\textwidth]{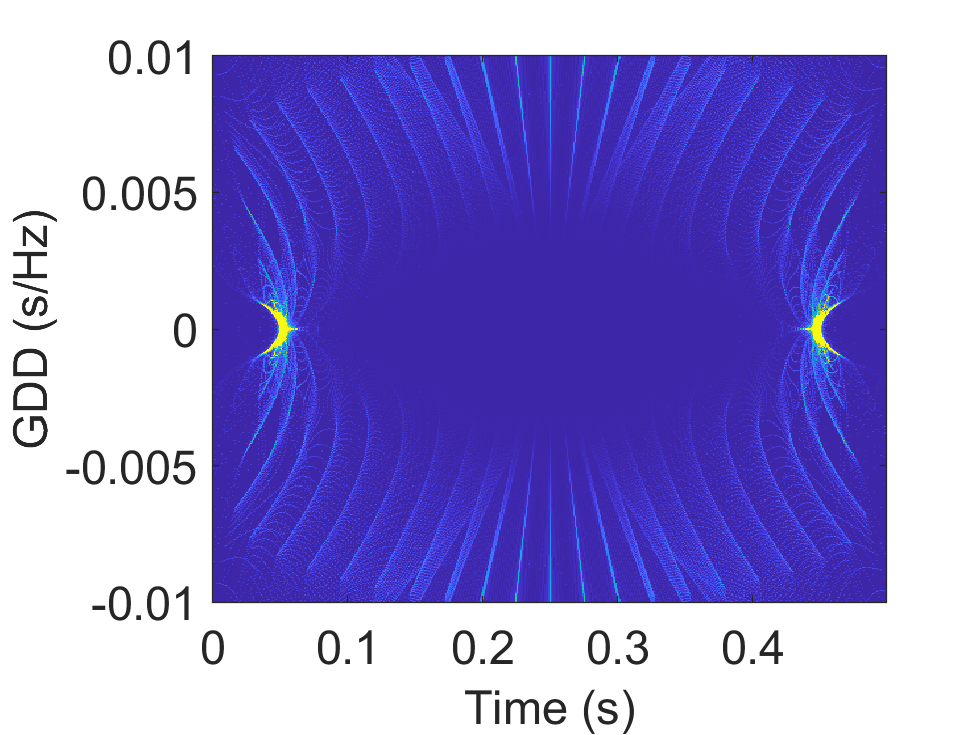}%
        } &
        \subfloat[ CT \(\left|{Q}_{x}^{g}(t, 256, \gga)\right|\)]
        {%
            \includegraphics[width=0.32\textwidth]{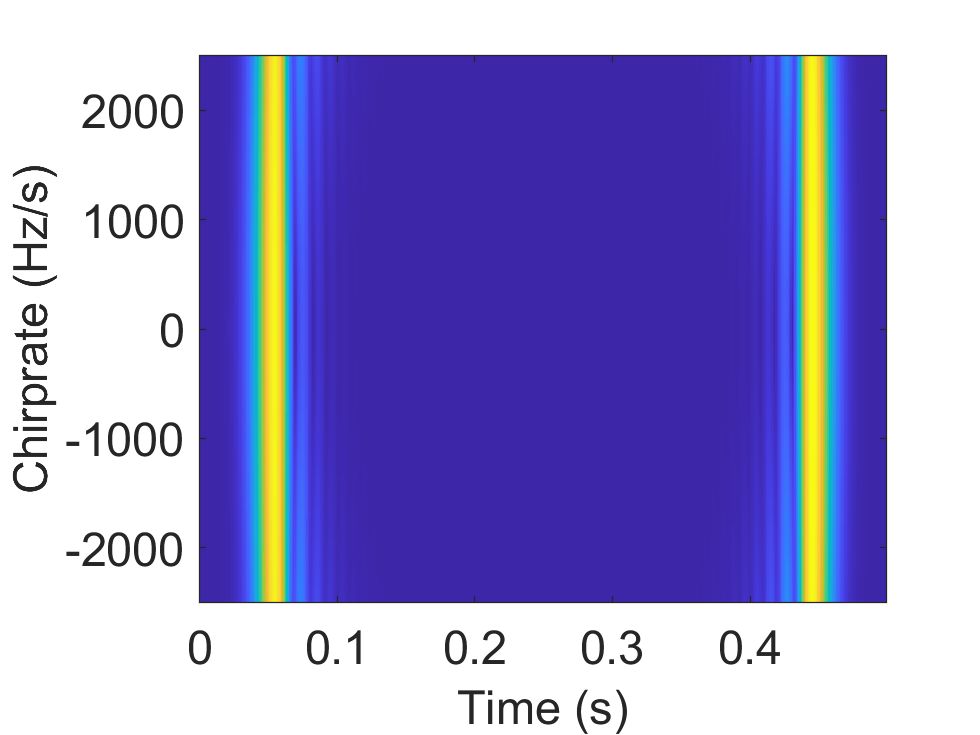}%
        } \\
    \end{tabular}
    \caption{\small Cross-sections of FCT $\mathcal{D}_{x}^{g}(t, \eta, \gga)$  and  CT ${Q}_{x}^{g}(t, \eta, \gga)$ at \(\eta=256\) }
    \label{cross-sections of $y(t)$}
\end{figure}

Subsequently, the three-dimensional views of the TSFCT magnitude  and the SCT magnitude  are plotted in Fig.~\ref{three-dimensional views of the TSFCT and SCT}.
Two continuous curves are observed in Panel (a). In contrast, both curves in Panel (b) exhibit a clear break at 
\(\eta=256\) Hz, which complicates the accurate extraction of the ridge curves.

\begin{figure}[H]
    \centering
    \begin{tabular}{cc}
        \subfloat[TSFCT magnitude]
        {%
            \includegraphics[width=0.32\textwidth]{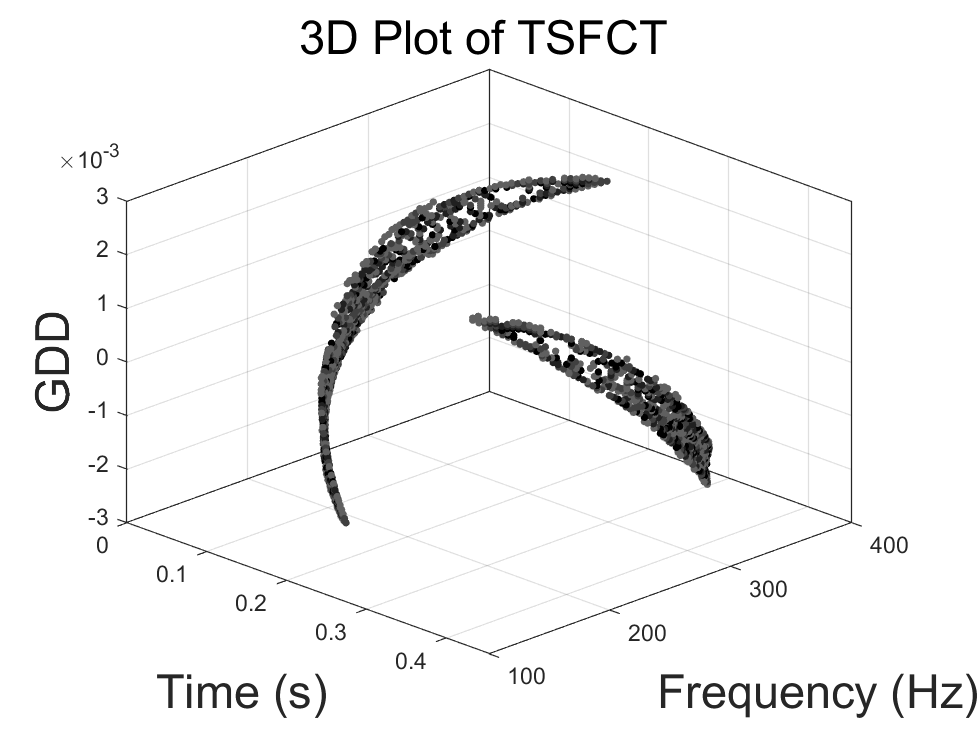}%
        } &
        \subfloat[SCT  magnitude]
        {%
            \includegraphics[width=0.32\textwidth]{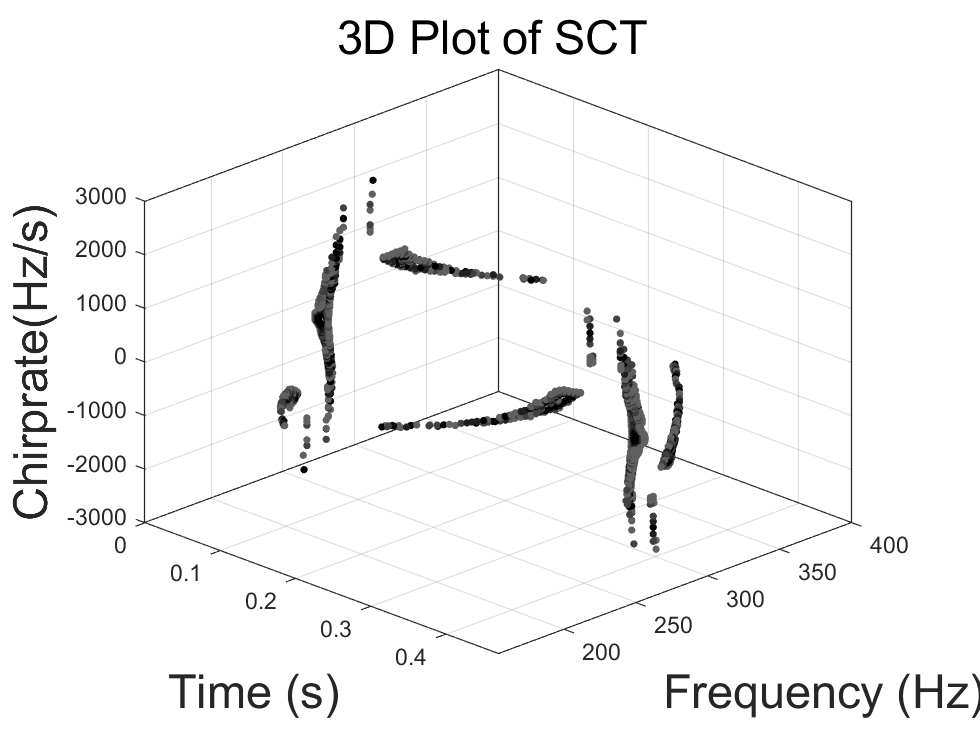}%
        } \\
    \end{tabular}
    \caption{\small TSFCT magnitude   and  SCT magnitude }
    \label{three-dimensional views of the TSFCT and SCT}
\end{figure}

Fig.~\ref{figure:GDs and GDDs estimation of $y(t)$} shows the GD and GDD estimation results of signal $y(t)$ by TSFCT. Except at the boundaries, the estimated GDs are virtually error-free. As for the GDD estimation, only minor deviations from the ground truth are observed, which fully demonstrates the capability of the TSFCT.
\begin{figure}[H]
    \centering
    \begin{tabular}{cccc}
        \begin{subfigure}[t]{0.22\textwidth}
            \centering
            \resizebox{\linewidth}{!}{\includegraphics{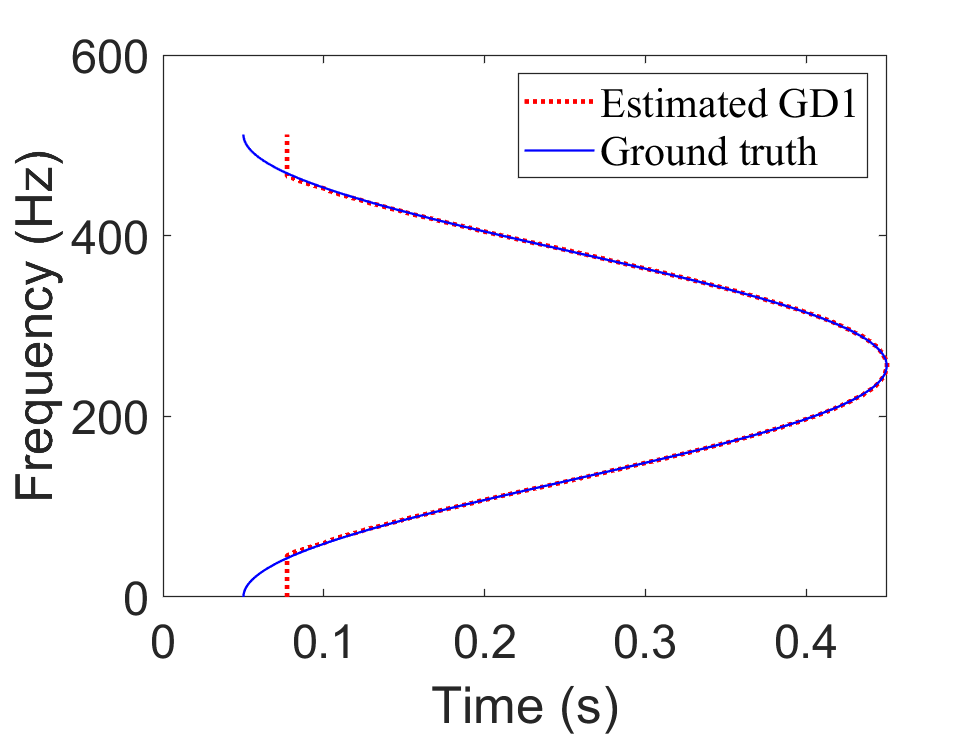}}
        \end{subfigure} &
        \begin{subfigure}[t]{0.22\textwidth}
            \centering
            \resizebox{\linewidth}{!}{\includegraphics{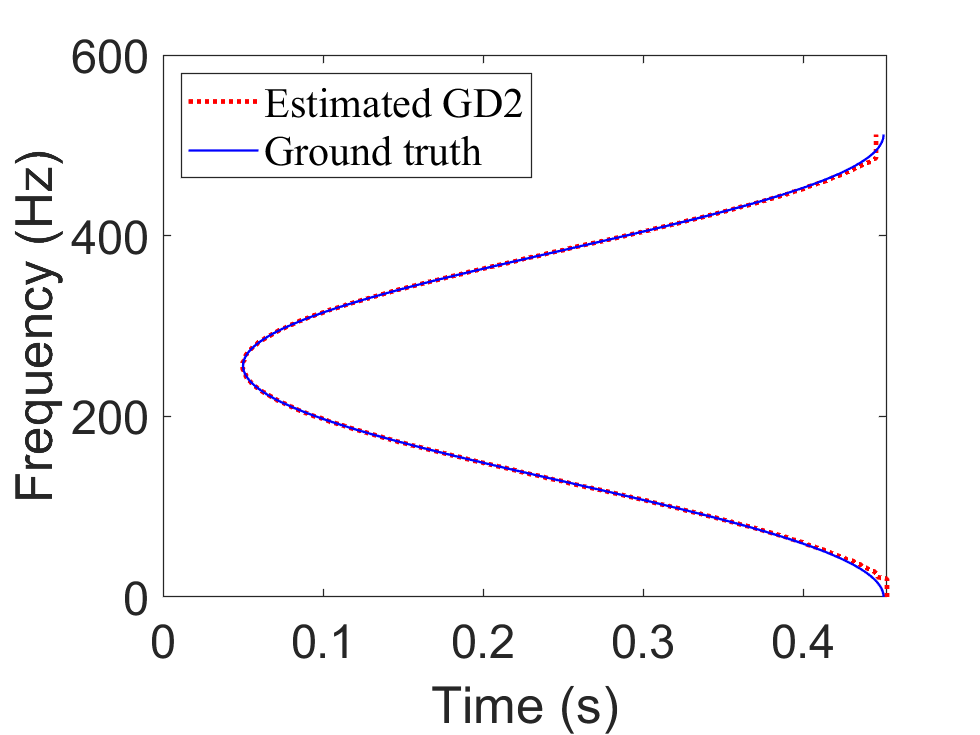}}
        \end{subfigure} &
        \begin{subfigure}[t]{0.22\textwidth}
            \centering
            \resizebox{\linewidth}{!}{\includegraphics{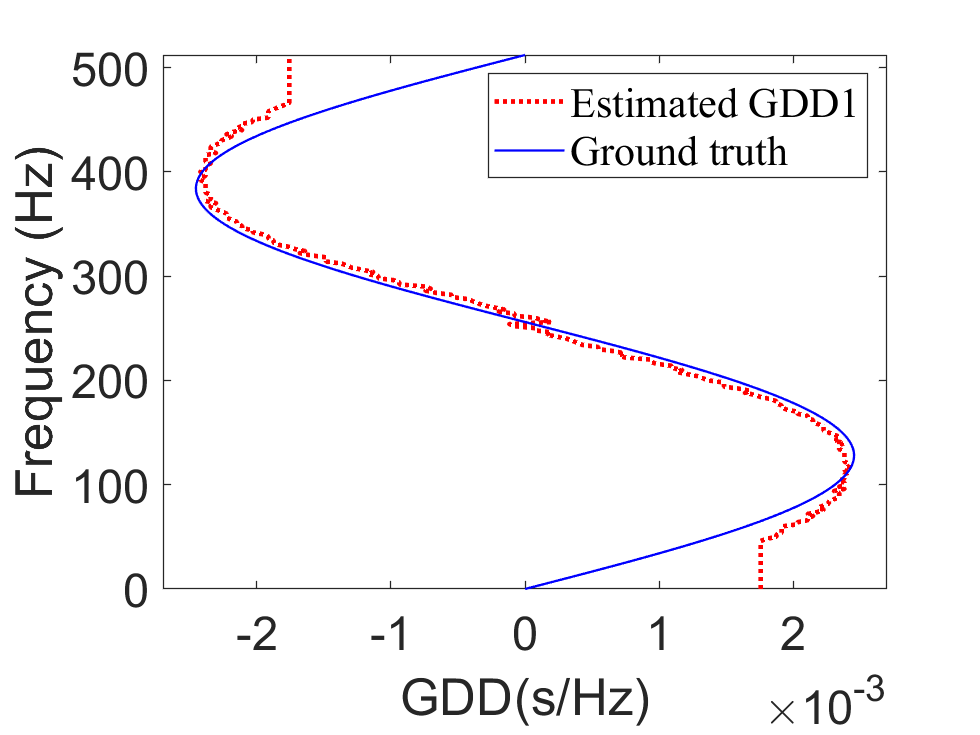}}
        \end{subfigure} &
        \begin{subfigure}[t]{0.22\textwidth}
            \centering
            \resizebox{\linewidth}{!}{\includegraphics{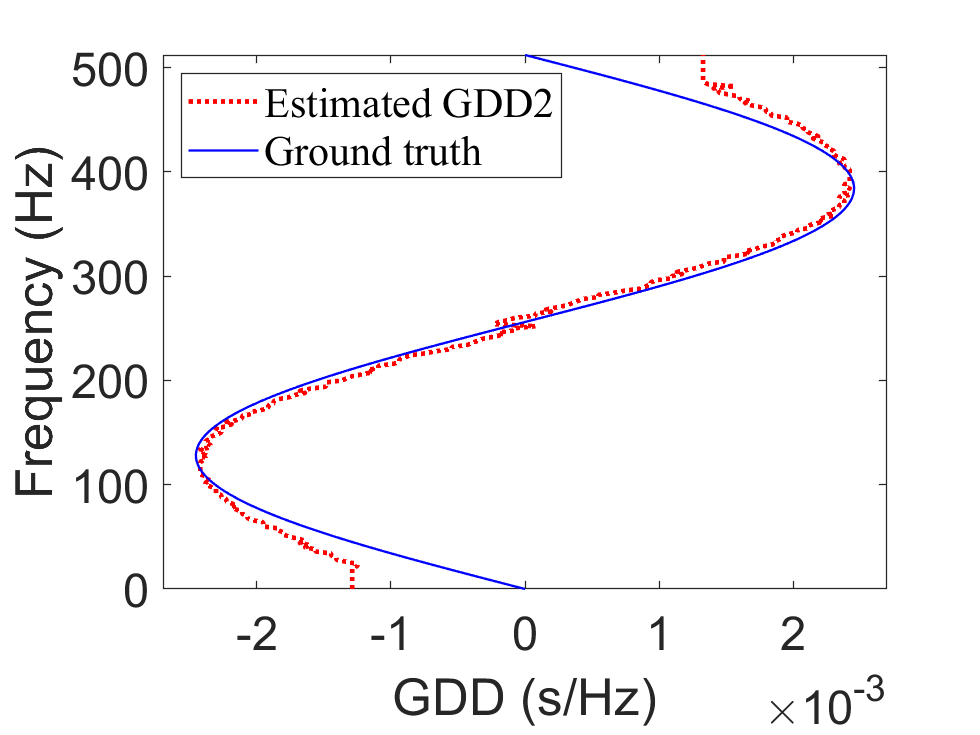}} 
        \end{subfigure}
    \end{tabular} 
    \small \caption{Estimation of GDs and GDDs for \(y(t)\) by TSFCT.} 
      \label{figure:GDs and GDDs estimation of $y(t)$}
\end{figure}

Finally, Fig.~\ref{figure:recover_errors_of_modes} presents the recovery errors of the modes using the FGSSO algorithm. The estimated GDs and GDDs ( Fig.~\ref{figure:GDs and GDDs estimation of $y(t)$} ), extracted from the TSFCT, are close to the ground truth.
The first row clearly reveals that while the FGSSO algorithm effectively recovers the modes $\widehat{y}_k(\eta)$ and ${y}_k(t)$  ($k=1,2$) , it introduces errors due to the complex sinusoidal structure of $\widehat{y}_k(\eta)$ itself.

As established in Appendix~\ref{section error functions} and Theorem~\ref{theorem_recover}, a smaller $\sigma$ yields a smaller value of $I_m$, which in turn could lead to a reduced reconstruction error. 
Therefore, we set $\sigma_0 = \sigma / 3 = 5.7$ to evaluate the reconstruction performance by rerunning the FGSSO algorithm using the same estimated GDs and GDDs curves, as presented in the second row.
The FGSSO algorithm with $\sigma_0$ demonstrates significantly improved recovery performance.

\begin{figure}[H]
    \centering
    \begin{tabular}{cccc}
         \begin{subfigure}[t]{0.22\textwidth}
            \centering
            \resizebox{\linewidth}{!}{\includegraphics{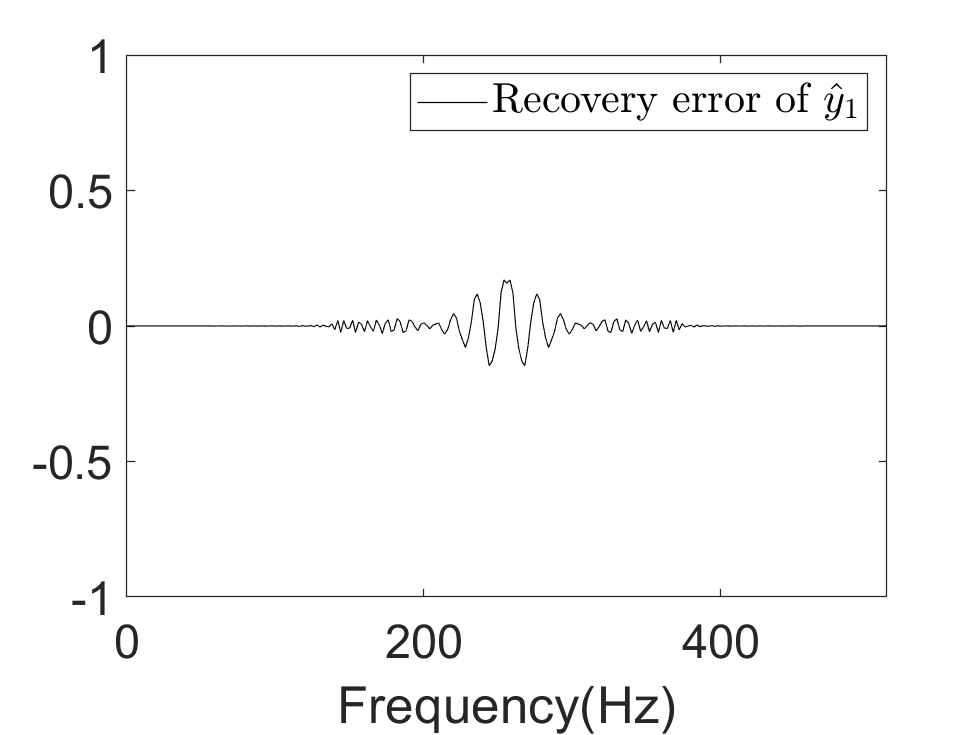}} 
        \end{subfigure} &
        \begin{subfigure}[t]{0.22\textwidth}
            \centering
            \resizebox{\linewidth}{!}{\includegraphics{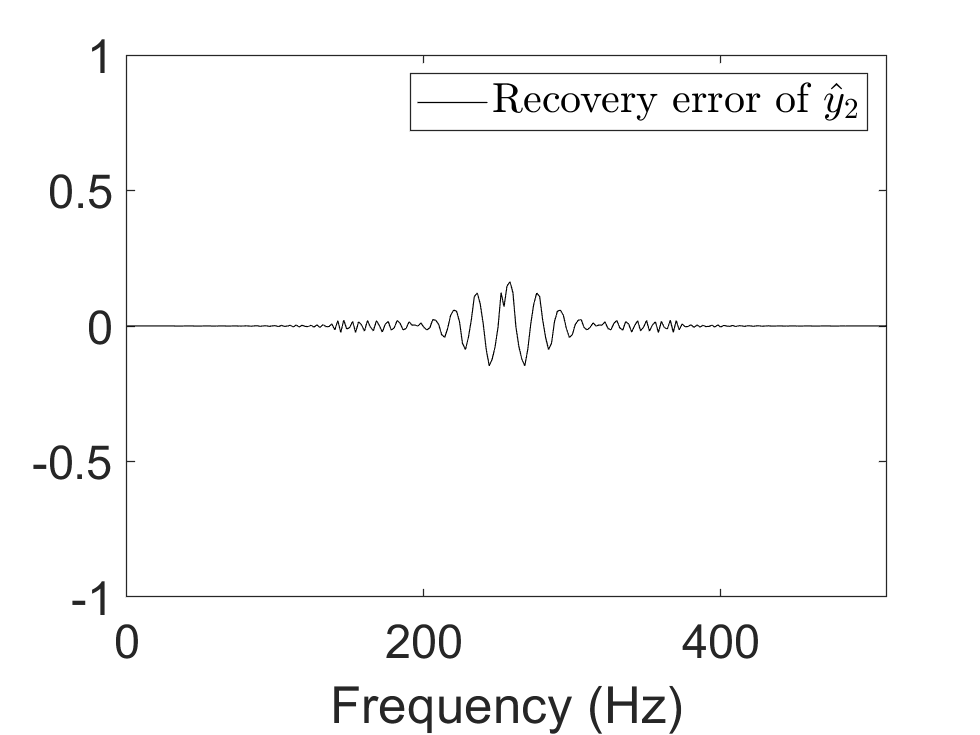}}
        \end{subfigure} &
        \begin{subfigure}[t]{0.22\textwidth}
            \centering
            \resizebox{\linewidth}{!}{\includegraphics{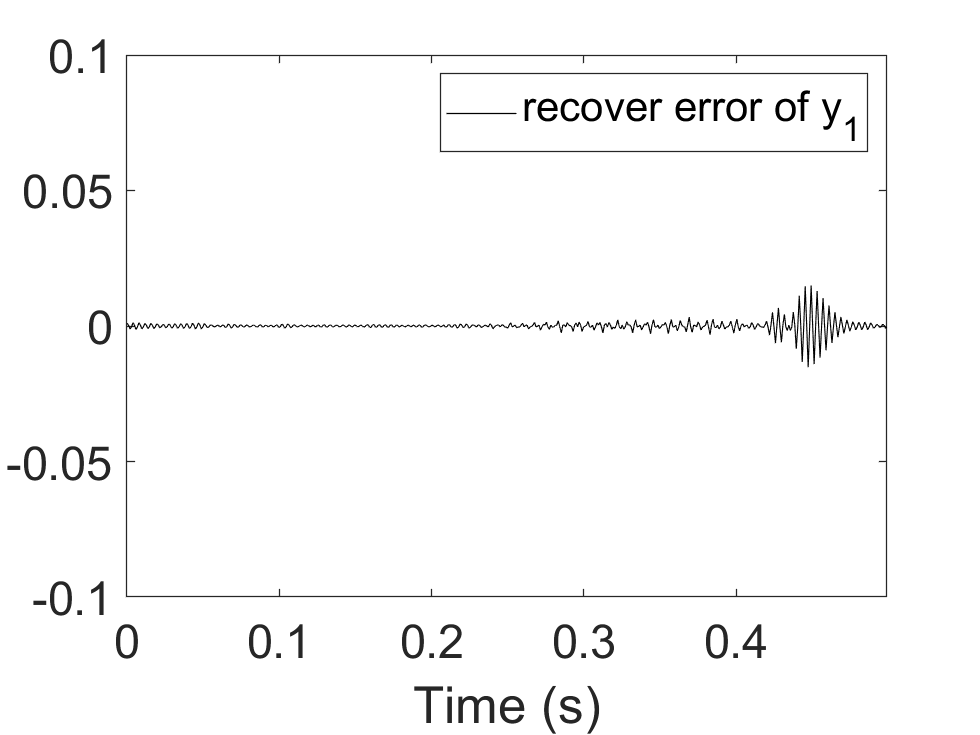}}
        \end{subfigure} &
        \begin{subfigure}[t]{0.22\textwidth}
            \centering
            \resizebox{\linewidth}{!}{\includegraphics{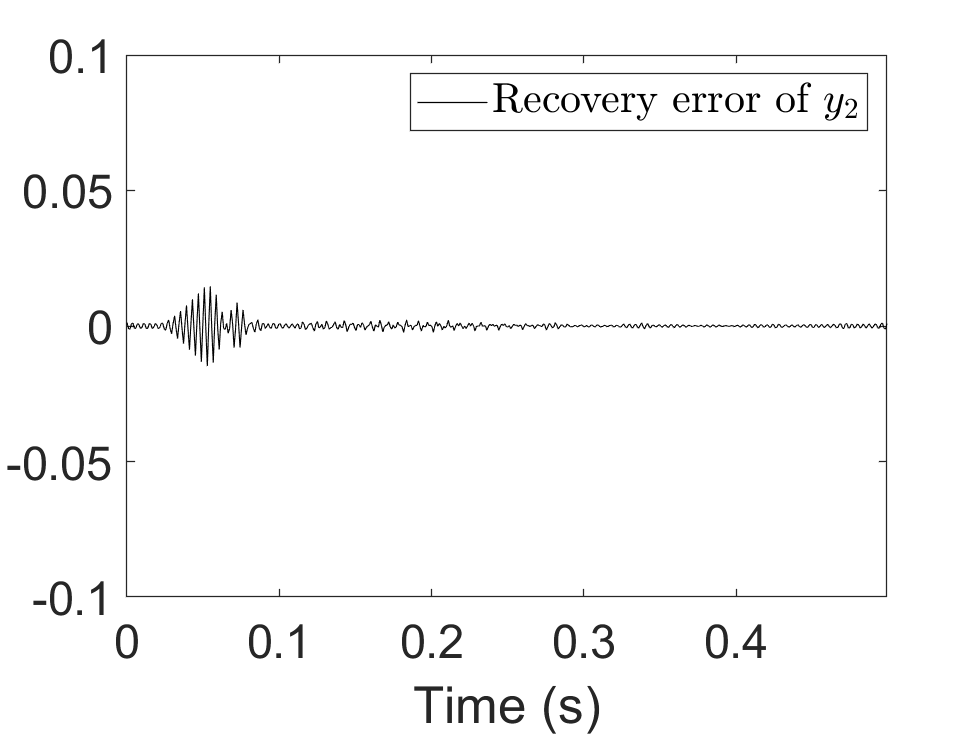}}
        \end{subfigure}  \\
         \begin{subfigure}[t]{0.22\textwidth}
            \centering
            \resizebox{\linewidth}{!}{\includegraphics{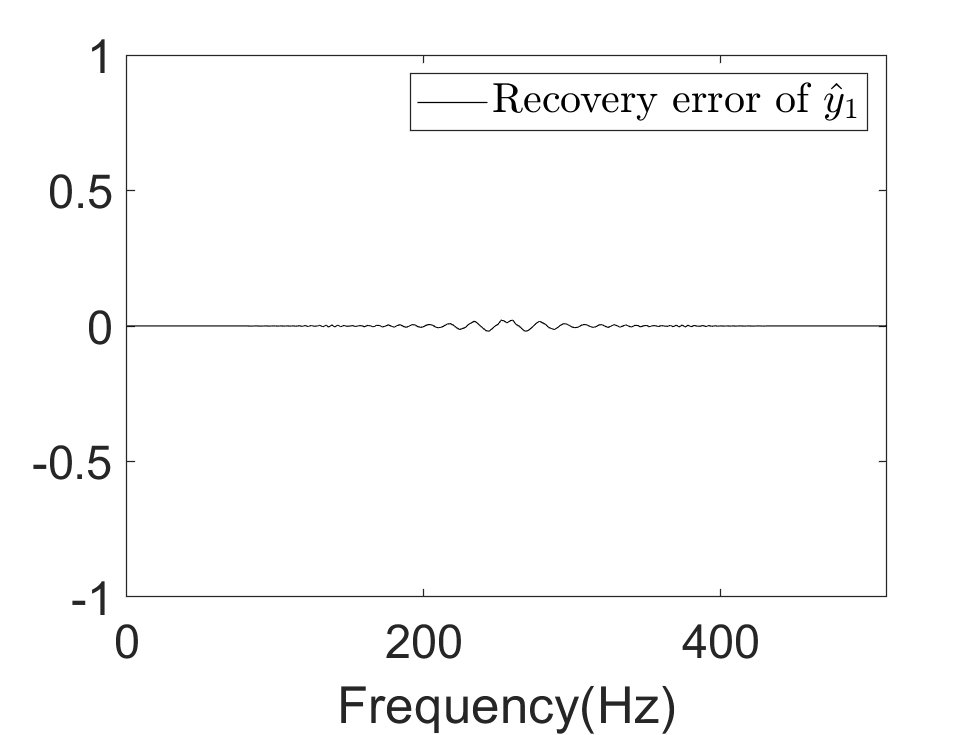}}    
        \end{subfigure}&
        \begin{subfigure}[t]{0.22\textwidth}
            \centering
            \resizebox{\linewidth}{!}{\includegraphics{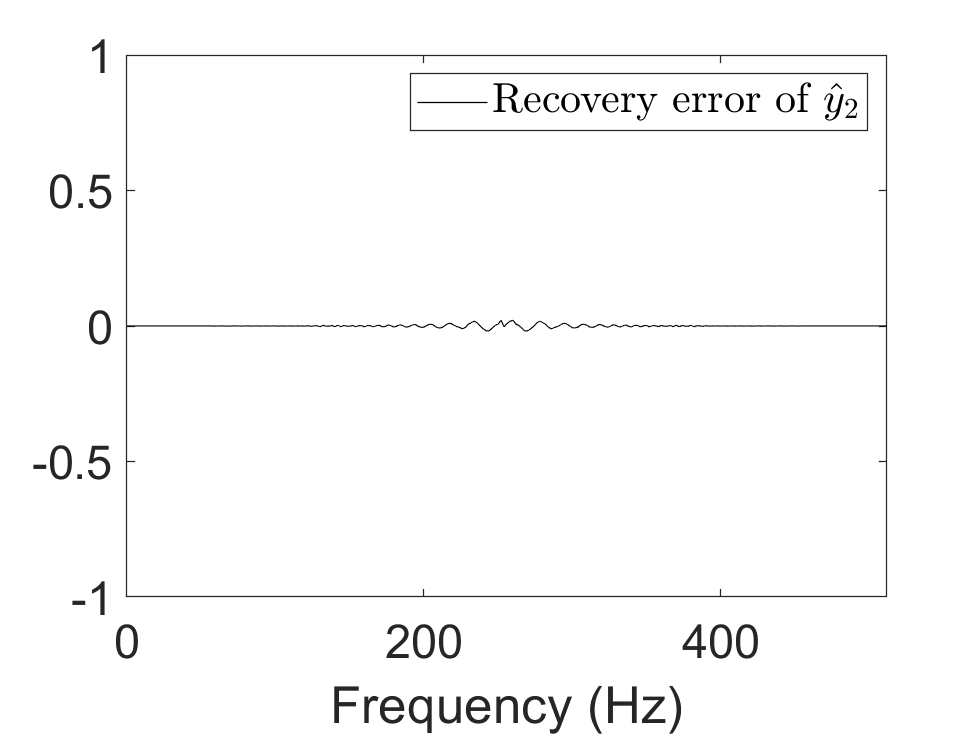}} 
        \end{subfigure} &
        \begin{subfigure}[t]{0.22\textwidth}
            \centering
            \resizebox{\linewidth}{!}{\includegraphics{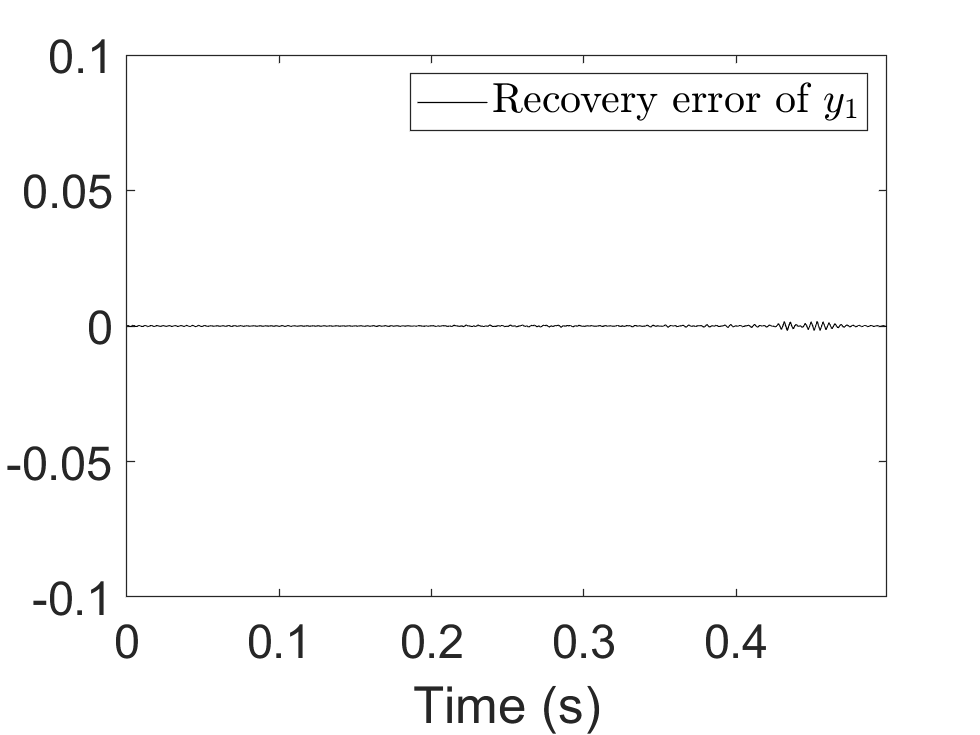}}    
        \end{subfigure} &
        \begin{subfigure}[t]{0.22\textwidth}
            \centering
            \resizebox{\linewidth}{!}{\includegraphics{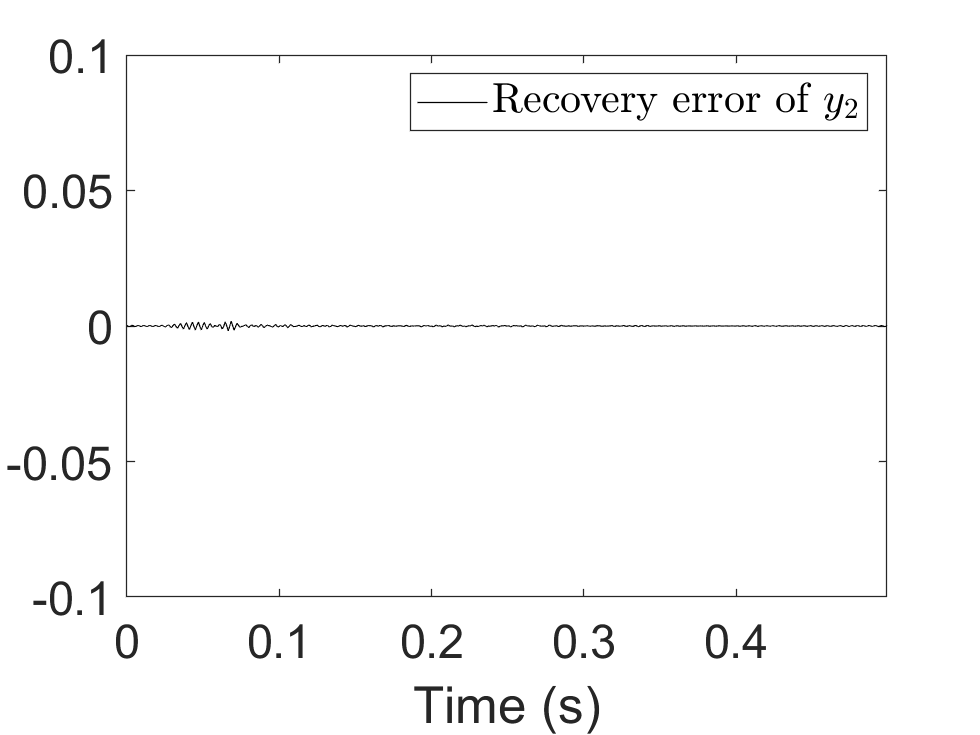}}    
        \end{subfigure}  \\        
    \end{tabular}
\caption{First row: Recovery errors (real part) of $y(t)$ using the FGSSO algorithm with $\sigma = 17.1$. From left to right: Recovery error of $\widehat{y}_1(\eta)$; recovery error of $\widehat{y}_2(\eta)$; recovery error of $y_1(t)$; recovery error of $y_2(t)$. Second row: Corresponding recovery errors using $\sigma_0 = \sigma/3 = 5.7$.}
    \label{figure:recover_errors_of_modes}
\end{figure}

\subsection{Audio recording of killer whale vocalizations}

The dataset\footnote{Data available at: \url{https://datadryad.org/dataset/doi:10.5061/dryad.1h46d}} comprises 142 audio recordings of killer whale (\textit{Orcinus orca}) vocalizations from off the Western Australian coast. These recordings capture a variety of sound types, including echolocation clicks used for navigation and foraging, whistles believed to facilitate social communication, and burst-pulse sounds which function as contact calls for group recognition and coordinating behavior.
A detailed assessment of these vocalizations can provide valuable insights into the population status, habitat usage, migration patterns, behavior, and acoustic ecology of this species in Australian waters, thus paving the way for further investigation of this poorly understood marine population \cite{wellard2015vocalisations}. However, the analysis of these signals is particularly challenging due to the complex underwater hydrological environment.

To validate the performance of the proposed TSFCT method on  complex, real-world signals, we applied it to one specific recording from this dataset (Sequence No. 115) for detailed examination. The original signal has a sampling rate of 96 kHz. 
To reduce computational complexity while preserving  relevant frequency content, we downsampled it by a factor of 3. Our analysis targeted a temporal segment ($0.6$ to $1.1$ s; duration: $0.5$ s) of this downsampled recording.

\begin{figure}[H]
    \centering
    \begin{tabular}{cc}
        \begin{subfigure}[t]{0.28\textwidth}
            \centering
            \resizebox{\linewidth}{!}{\includegraphics{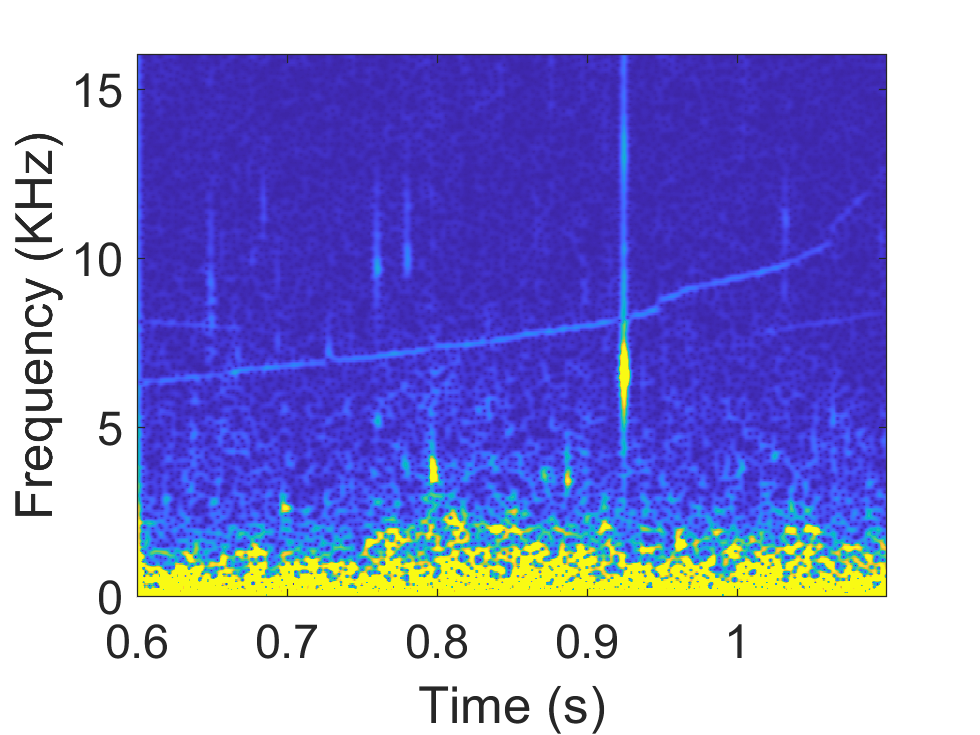}}
        \end{subfigure} &
        \begin{subfigure}[t]{0.28\textwidth}
            \centering
            \resizebox{\linewidth}{!}{\includegraphics{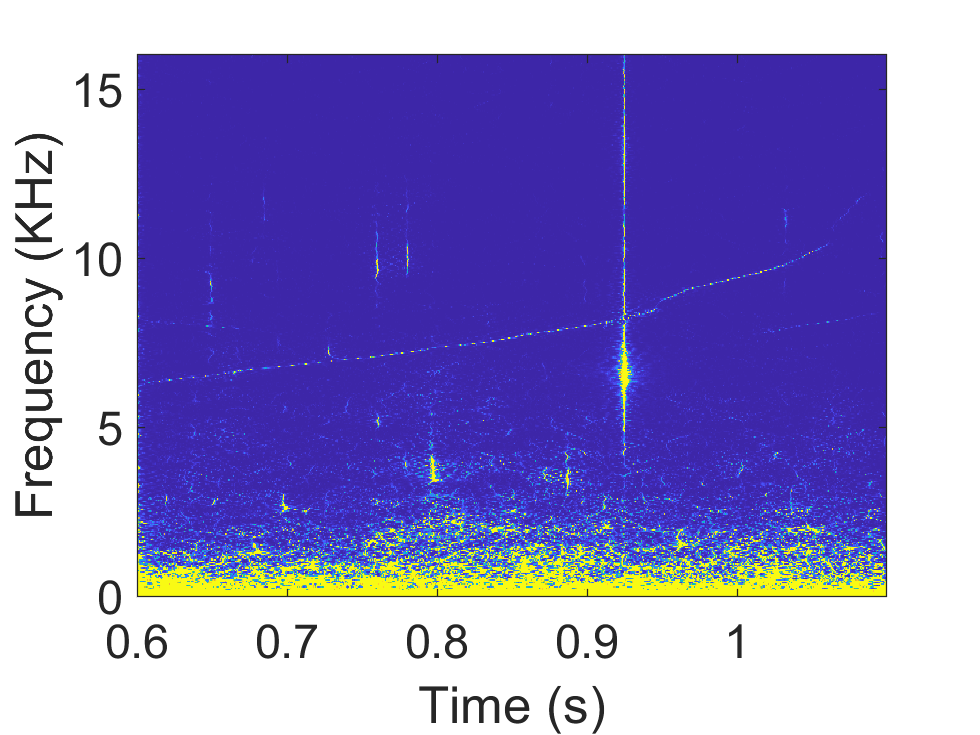}}
        \end{subfigure} 
    \end{tabular} 
 \small
\caption{TFRs of the audio recording. From left to right: STFT, 
TFR projection $\mathfrak{T}_{x}^{g}(\tau, \eta)$ (defined in \eqref{eq:tfr_definition}).}
      \label{figure:TFRs of killer whale}
\end{figure}

The TFR generated by the STFT method is shown in the left panel of Fig.~\ref{figure:TFRs of killer whale}. Although substantial background noise contaminates the 0-3 kHz frequency range, two distinct vocalization types are discernible. 
An echolocation click spanning 4-16 kHz occurs around 0.925 s, while a whistle signal appears from 0.6 s to 1.1 s.
Since echolocation clicks appear nearly perpendicular to the time axis in the time-frequency plane, the three-dimensional methods based on CT or WCT bases are unsuitable. Under these circumstances, FCT and TSFCT serve as valuable complements for analyzing such transient features.
Due to the presence of background noise that perturbs the energy distribution, displaying the complete three-dimensional TF-GDD representation becomes impractical. 
Therefore, we present 
TFR projection $\mathfrak{T}_{x}^{g}(\tau, \eta)$ (defined in \eqref{eq:tfr_definition}), which is shown in  the right panel.
The TFR generated by the TSFCT method effectively captures the time-frequency characteristics of both the echolocation click and the whistle, while reducing background noise interference.

\section{Conclusion}

In this paper, we propose a method based on the frequency-domain chirplet transform (FCT) to address multicomponent signals with intersecting group delay (GD) curves. By introducing a novel group delay dispersion (GDD) variable, the FCT extends the traditional time-frequency (TF) plane to a three-dimensional time-frequency-group delay dispersion (TF-GDD) space. Compared to the conventional chirplet transform (CT), the proposed FCT-based method effectively handles multicomponent signals featuring crossing or rapidly varying frequency ridges.
We further generalize the time-reassigned framework to this three-dimensional space and derive novel GD and GDD reference functions. This generalization yields significantly sharper TF-GDD distributions, which enables precise estimation of both GDs and GDDs. Furthermore, we propose a novel frequency-domain group signal separation operation (FGSSO) scheme to 
enable accurate recovery of transient signal modes--even in the cases with intersecting group GD curves.
Additionally, we establish theorems that characterize two key error metrics: first, the approximation error of the GD and GDD reference functions relative to the true GD and GDD of the signal; second, the error bounds for mode retrieval using FGSSO.

These advancements contribute to the progress of non-stationary signal processing. However, the analysis of complex signals containing both transient and harmonic components--which remains challenging for both CT and FCT--continues to be an open problem. Therefore, future research will focus on developing novel methods specifically tailored for such mixed signals.

\section*{Acknowledgments} 
This work was partially supported by the National Key Research and Development Program of China  under Grants 2022YFA1005703  and the National Natural Science Foundation of China under Grants U21A20455 and 12571109.

\begin{appendices}
  \numberwithin{equation}{section} 
  
\section{Expressions of the high-order GD and GDD reference functions}

For a signal $x(t)$, we define the following matrices:
\begingroup
\renewcommand{\arraystretch}{1.1}
\setlength{\arraycolsep}{4pt}
\begin{align*}
& E^N_{0} =
\begin{bmatrix}
\mathcal{D}_x^{g} & \mathcal{D}_x^{\xi g} & \cdots & \mathcal{D}_x^{\xi^{N-1} g} \\
\mathcal{D}_x^{\xi g} & \mathcal{D}_x^{\xi^2 g} & \cdots & \mathcal{D}_x^{\xi^{N} g} \\
\vdots & \vdots & \ddots & \vdots \\
\mathcal{D}_x^{\xi^{N-1} g} & \mathcal{D}_x^{\xi^{N} g} & \cdots & \mathcal{D}_x^{\xi^{2N-2} g}
\end{bmatrix},
&& E^N_{11} =
\begin{bmatrix}
\mathcal{D}_x^{g'} & \mathcal{D}_x^{\xi g} & \cdots & \mathcal{D}_x^{\xi^{N-1} g} \\
\mathcal{D}_x^{\xi g'} & \mathcal{D}_x^{\xi^{2} g} & \cdots & \mathcal{D}_x^{\xi^{N} g} \\
\vdots & \vdots & \ddots & \vdots \\
\mathcal{D}_x^{\xi^{N-1} g'} & \mathcal{D}_x^{\xi^{N} g} & \cdots & \mathcal{D}_x^{\xi^{2N-2} g}
\end{bmatrix},
\\[2ex]
& E^N_{12} =
\begin{bmatrix}
0 & \mathcal{D}_x^{\xi g} & \cdots & \mathcal{D}_x^{\xi^{N-1} g} \\
\mathcal{D}_x^{g} & \mathcal{D}_x^{\xi^{2} g} & \cdots & \mathcal{D}_x^{\xi^{N} g} \\
\vdots & \vdots & \ddots & \vdots \\
(N-1) \mathcal{D}_x^{\xi^{N-2}g} & \mathcal{D}_x^{\xi^{N} g} & \cdots & \mathcal{D}_x^{\xi^{2N-2} g}
\end{bmatrix},
&& E^N_{21} =
\begin{bmatrix}
\mathcal{D}_x^{g} & \mathcal{D}_x^{g'} & \cdots & \mathcal{D}_x^{\xi^{N-1} g} \\
\mathcal{D}_x^{\xi g} & \mathcal{D}_x^{\xi g'} & \cdots & \mathcal{D}_x^{\xi^{N} g} \\
\vdots & \vdots & \ddots & \vdots \\
\mathcal{D}_x^{\xi^{N-1} g} & \mathcal{D}_x^{\xi^{N-1} g'} & \cdots & \mathcal{D}_x^{\xi^{2N-2} g}
\end{bmatrix},
\\[2ex]
& E^N_{22} =
\begin{bmatrix}
\mathcal{D}_x^{g} & 0 & \cdots & \mathcal{D}_x^{\xi^{N-1} g} \\
\mathcal{D}_x^{\xi g} & \mathcal{D}_x^{g} & \cdots & \mathcal{D}_x^{\xi^{N} g} \\
\vdots & \vdots & \ddots & \vdots \\
\mathcal{D}_x^{\xi^{N-1} g} & (N-1) \mathcal{D}_x^{\xi^{N-2} g} & \cdots & \mathcal{D}_x^{\xi^{2N-2} g}
\end{bmatrix}.
\end{align*}
\endgroup

For $(t, \eta, \gga)$ with $\det(E^N_0) \neq 0$,  define 
\begin{align}
\label{NGD}\widehat{t}_{N}(t,\eta,\gga) &:= t + \frac{1}{2\pi}\operatorname{Im}\left(\frac{\det(E^N_{11}) + \det(E^N_{12})}{\det(E^N_0)}\right), \\[1ex]
\label{NGDD}\widehat{r}_{N}(t,\eta,\gga) &:= \lambda + \frac{1}{2\pi}\operatorname{Im}\left(\frac{\det(E^N_{21}) + \det(E^N_{22})}{\det(E^N_0)}\right).
\end{align}

One can verify that for a signal in the form of 
\[
\hat{x}(\eta)=\mathrm{e}^{-d(\eta)}\,\mathrm{e}^{-i 2\pi\theta(\eta)},
\]
where \(d(\eta)\) and \(\theta(\eta)\) are real-valued polynomials of degree~\(N\), the following hold: 
$$
\widehat{t}_{N}(t,\eta,\gga) =\theta'(\eta), \quad \widehat{r}_{N}(t,\eta,\gga) =\theta''(\eta).
$$
Thus the quantities on the righ-hand sides of \eqref{NGD} and \eqref{NGDD} can be used as the  $N$th-order GD and GDD reference functions.

\section{Proof of Theorem \ref{theorem_eta_lambda} }\label{proof of theorem}

The proof of Theorem \ref{theorem_eta_lambda}  is based on the two lemmas to be established below.

\begin{lem}\label{approxC_x}
  Let \( x(t) \in \mathcal{B}_{\epsilon_1,\epsilon_2} \) be a multicomponent signal 
for some \(\epsilon_1 > 0\) and \(\epsilon_2 > 0\). Then for \( m = 0,1,2,\cdots \), 
its FCT satisfies:
    \begin{align}
     \left|\mathcal{D}_{x_l}^{\xi^m g}(t,\eta,\gga) - \wh{x}_l(\eta) \mathcal{C}(\xi^m g)(\theta_l'(\eta)-t,\theta_l''(\eta)-\gga)\right| \leq \Pi_{m,l},
    \end{align}
 where   $\Pi_{m,l}$ is defined by \eqref{piml}. 
\end{lem}

\begin{proof}
For each mode $\widehat{x_l}(\eta)=B_k(\eta)e^ {-i2 \pi \theta_k(\eta)}$, we can write $ \widehat{x_l}(\eta+\xi)$ as:
 \begin{align*}
     \widehat{x_l}(\eta+\xi)= &B_l  (\eta+\xi)e^ {-i2 \pi \theta_l(\eta+\xi)}  \\
                =&B_l(\eta)(\eta)e^{-i 2 \pi \left(\theta_l(\eta)+\theta'_l(\eta)\xi+\frac{1}{2}\theta''_l(\eta) \xi^2 \right)} +\left(B_l(\eta+\xi)-B_l(\eta)\right)e^ {-i2 \pi \theta_l(\eta+\xi)} \\
                +&B_l(\eta)\left(e^ {-i2 \pi \theta_l(\eta+\xi)}- e^{-i 2 \pi \left(\theta_l(\eta)+\theta'_l(\eta)\xi+\frac{1}{2}\theta''_l(\eta) \xi^2 \right)}\right).           
 \end{align*}
 Then,
 \begin{align*}
    &\left|\mathcal{D}_{x_l}^{\xi^m g}(t,\eta,\gga)-\wh{x}_l(\eta) \mathcal{C}(\xi^m g)(\theta_l'(\eta)-t,\theta_l''(\eta)-\gga) \right|\\
    \leq& \left| \int_{\mathbb{R}} (B_l(\eta+\xi)-B_l(\eta)) e^{-i 2 \pi \theta_l(\eta+\xi)} \xi^m g(\xi) e^{i 2 \pi \xi t} e^{i \pi \gga \xi^2} d\xi \right|\\
    +&\left| \int_{\mathbb{R}} B_l(\eta)\left(e^ {-i2 \pi \theta_l(\eta+\xi)}- e^{-i 2 \pi \left(\theta_l(\eta)+\theta'_l(\eta)\xi+\frac{1}{2}\theta''_l(\eta) \xi^2 \right)}\right) \xi^m g(\xi) e^{i 2 \pi \xi t} e^{i \pi \gga \xi^2} d\xi \right|        
 \end{align*}

    Since $|B_l(\eta+\xi)-B_l(\eta)|\leq \epsilon_1 |\xi| $ and 
     \begin{align*}
         \left|e^ {-i2 \pi \theta_k(\eta+\xi)}-e^{-i 2 \pi \left(\theta_l(\eta)+\theta'_l(\eta)\xi+\frac{1}{2}\theta''_l(\eta) \xi^2 \right)} \right|\leq \frac{\pi}{3}\epsilon_2 |\xi|^{3},
     \end{align*}
    the remainder  can be bounded by \(\Pi_{m,l}\),
   \begin{align*}
   \left|\mathcal{D}_{x_l}^{\xi^m g}(t,\eta,\gga)-\wh{x}_l(\eta) \mathcal{C}(\xi^m g)(\theta_l'(\eta)-t,\theta_l''(\eta)-\gga) \right|\leq  \epsilon_1 I_{m+1} + \frac{\pi}{3} B_l(\eta)\epsilon_2 I_{m+3}=\Pi_{m,l}.
   \end{align*}
  \end{proof}

Furthermore, we have
\begin{align}\label{approx_allk}
\bigg| \mathcal{D}_{x}^{\xi^m g}(t,\eta,\gga) 
- \sum_{l=1}^{K} \widehat{x_l}(\eta) \mathcal{C}(\xi^m g)\bigl(\theta_l'(\eta)-t,\theta_l''(\eta)-\gga\bigr) \bigg| 
\leq \sum_{l=1}^{K} \Pi_{m,l} = \Pi_m.
\end{align}
If \((t,\eta,\gga) \in Z_k\)\((k\neq l)\), based on the Lemma \ref{approxC_x} and \eqref{Upsilonmk},
\begin{align}
     \left|\mathcal{D}_{x_l}^{\xi^m g}(t,\eta,\gga)\right| \leq B_l(\eta) \Upsilon_{m,l}(\eta) + \Pi_{m,l}. \label{modulus_not_in_Z_k}          
 \end{align}
Moreover, the $\gga_m$  is the  upper bound of the modulus  
for $\mathcal{D}_{x}^{\xi^{m}g}(t,\eta,\gga)$ in the region $Z_k$, where $ \gga_m$ is defined by  \eqref{modulus_Cx}. 
Indeed,
    \begin{align*}
     \left|\mathcal{D}_{x}^{\xi^m g}(t,\eta,\gga)\right| \leq&  \sum_{l=1}^{K}\left|\wh{x}_l(\eta) \mathcal{C}(\xi^m g)(\theta_l'(\eta)-t,\theta_l''(\eta)-\gga)\right|+\Pi_m  \nonumber\\
       \leq& B_k(\eta)I_m + \sum_{l\neq k} B_l(\eta) \Upsilon_{m,l}(\eta)+\Pi_m=\gga_m.         
 \end{align*}

Next, we analyze the error bound of the residual term 
\begin{align}\label{def_res}
\mathrm{Res}_{m,k} = \partial_\eta \mathcal{D}_x^{\xi^m g}(t,\eta,\gga) 
+ i2\pi  \theta'_k(\eta) \mathcal{D}_x^{\xi^{m} g}(t,\eta,\gga) 
+ i2\pi  \theta''_k(\eta) \mathcal{D}_x^{\xi^{m+1} g}(t,\eta,\gga)
\end{align}
in the \( Z_k \) regions, which plays a crucial role in proving  the main theorem.

\begin{lem}
    \label{difference_partial_eta}
    Assume that $x(t)$ satisfies the same conditions as in Lemma \ref{approxC_x}. If $(t,\eta,\gga) \in Z_k$, then for  the residual term $\mathrm{Res}_{m,k}$ defined in \eqref{def_res}, we have:
    \begin{align*}
    \left| \mathrm{Res}_{m,k}\right|\leq \Lambda_{m,k},
     \end{align*}
    where $\Lambda_{m,k}$ is defined by  \eqref{Lambdamk}.
\end{lem}

\begin{proof}
 First, we define 
\begin{align*} 
Q_{m,l} &= \partial_\eta \mathcal{D}_{x_l}^{\xi^m g}(t,\eta,\gga) + i2\pi  \theta'_l(\eta) \mathcal{D}_x^{\xi^{m} g}(t,\eta,\gga) + i2\pi  \theta''_l(\eta) \mathcal{D}_x^{\xi^{m+1} g}(t,\eta,\gga) \\
&= \int_{\mathbb{R}} B'_l(\eta+\xi) e^{-i 2 \pi \theta_l(\eta+\xi)} \xi^m g(\xi) e^{i 2 \pi \xi t} e^{i \pi \gga \xi^2} d\xi \\
&+ \int_{\mathbb{R}} \Bigl(B_l(\eta+\xi)-B_l(\eta)\Bigr) \Bigl(i2\pi  \theta'_l(\eta) + i2\pi  \theta''_l(\eta)\xi - i2\pi  \theta'_l(\eta+\xi) \Bigr) e^{-i 2 \pi \theta_l(\eta+\xi)} \xi^m g(\xi) e^{i 2 \pi \xi t} e^{i \pi \gga \xi^2} d\xi,\\
&+ \int_{\mathbb{R}} B_l(\eta) \Bigl(i2\pi  \theta'_l(\eta) + i2\pi  \theta''_l(\eta)\xi - i2\pi  \theta'_l(\eta+\xi) \Bigr) e^{-i 2 \pi \theta_l(\eta+\xi)} \xi^m g(\xi) e^{i 2 \pi \xi t} e^{i \pi \gga \xi^2} d\xi,
\end{align*}
Notice that,
\[\left|Q_{m,l}\right| \leq \epsilon_1 I_m +\epsilon_1\epsilon_2\pi I_{m+3}+\epsilon_2 \pi B_l(\eta)I_{m+2}\]
Hence,
  \begin{align*}
\mathrm{Res}_{m,k} &= \partial_\eta \mathcal{D}_x^{\xi^m g}(t,\eta,\gga)  + i2\pi  \theta'_k(\eta) \mathcal{D}_x^{\xi^{m} g}(t,\eta,\gga) + i2\pi  \theta''_k(\eta) \mathcal{D}_x^{\xi^{m+1} g}(t,\eta,\gga), \\
&= \sum_{l=1}^{K} \left(\partial_\eta \mathcal{D}_{x_l}^{\xi^m g} + i2\pi  \theta'_l(\eta) \mathcal{D}_{x_l}^{\xi^{m} g}+ i2\pi  \theta''_l(\eta) \mathcal{D}_{x_l}^{\xi^{m+1} g}\right) \\
&+ \sum_{l=1}^{K} i2\pi  \left(  (\theta'_k(\eta)-\theta'_l(\eta))\mathcal{D}_{x_l}^{\xi^{m} g}+ (\theta''_k(\eta)- \theta''_l(\eta)) \mathcal{D}_{x_l}^{\xi^{m+1} g}\right).
\end{align*}
 When \(\left(t,\eta,\gga\right) \in Z_k \), then from \eqref{modulus_not_in_Z_k}, we have
    \begin{align*}
        \left|\mathrm{Res}_{m,k} \right| & \leq \sum_{l=1}^{K}\left|Q_{m,l} \right| +\sum_{l\neq k} 2\pi \left(  \left|\theta'_k(\eta)-\theta'_l(\eta)\right| \left|\mathcal{D}_{x_l}^{\xi^{m} g}\right|+ \left|\theta''_k(\eta)- \theta''_l(\eta) \right| \left|\mathcal{D}_{x_l}^{\xi^{m+1} g}\right|\right)\\
                            &= \sum_{l\neq k} 2\pi \Big( \left|\theta'_k(\eta) - \theta'_l(\eta)\right| \big( B_l(\eta) \Upsilon_{m,l}(\eta) + \Pi_{m,l} \big)+ \left|\theta''_k(\eta) - \theta''_l(\eta)\right| \big( B_l(\eta) \Upsilon_{m+1,l}(\eta) + \Pi_{m+1,l} \big) \Big) \\
&\quad + \epsilon_1 K I_m + \epsilon_1 \epsilon_2 \pi K I_{m+3} + \epsilon_2 \pi M(\eta) I_{m+2}.
    \end{align*}  
\end{proof}

\n {\bf Proof of Theorem \ref{theorem_eta_lambda}}.
    When \(\left(t,\eta,\gga\right) \in Z_k\), according to Lemma \ref{difference_partial_eta}, we have
    \begin{align*}
        &  - \det(E_1)  - i 2 \pi \theta'_k(\eta) \det(E_0) =
        \begin{vmatrix}
           - \partial_\eta \mathcal{D}_x^{g} - i 2 \pi \theta'_k(\eta) \mathcal{D}_x^{g} & \mathcal{D}_x^{\eta g} \\
           - \partial_\eta \mathcal{D}_x^{\eta g} - i 2 \pi \theta'_k(\eta) \mathcal{D}_x^{\eta g} & \mathcal{D}_x^{\eta^2 g} \\
        \end{vmatrix} \\
        &= 
        \begin{vmatrix}
            {i2\pi  \theta'_k(\eta)}\mathcal{D}_x^{g}+{i2\pi  \theta''_k(\eta)}\mathcal{D}_x^{\eta g} - i 2 \pi \theta'_k(\eta) \mathcal{D}_x^{g} - \mathrm{Res}_{0,k} & \mathcal{D}_x^{\eta g} \\
             {i2\pi  \theta'_k(\eta)}\mathcal{D}_x^{\eta g}+{i2\pi  \theta''_k(\eta)}\mathcal{D}_x^{\eta^2 g}- i 2 \pi \theta'_k(\eta) \mathcal{D}_x^{\eta g} - \mathrm{Res}_{1,k} & \mathcal{D}_x^{\eta^2 g} \\
        \end{vmatrix} \\ 
        &=- 
        \begin{vmatrix}
            \mathrm{Res}_{0,k} & \mathcal{D}_x^{\eta g} \\
            \mathrm{Res}_{1,k} & \mathcal{D}_x^{\eta^2 g}\\
        \end{vmatrix} = -\mathrm{Res}_{0,k}\mathcal{D}_x^{\eta^2 g}+\mathrm{Res}_{1,k}\mathcal{D}_x^{\eta g}.  
    \end{align*}
    Note that, except for $ \mathrm{Res}_{0,k}$ and $ \mathrm{Res}_{1,k}$ in the first column, the other entries in the first column do not contribute to the determinant.
    Following similar arguments, we can derive analogous results:
\begin{align*}
     - \det(E_2)  - i 2 \pi \theta''_k(t) \det(E_0) = \mathrm{Res}_{0,k}\mathcal{D}_x^{\eta g} - \mathrm{Res}_{1,k}\mathcal{D}_x^{g}.
\end{align*}
    Thus, we have
    \begin{align*}
        |\wh{t}(t,\eta,\gga) - \theta_k'(\eta)| &\leq \left| -\det(E_1) - i 2 \pi \theta'_k(\eta) \det(E_0)\right|\left|2 \pi \det(E_0)\right|^{-1} \leq \frac{\epsilon_0}{2\pi} \left( \Lambda_{0,k} \gga_2 + \Lambda_{1,k} \gga_1 \right), \\
        |\wh{r}(t,\eta,\gga) - \theta_k''(\eta)| &\leq \left| -\det(E_2) - i 2 \pi \theta''_k(\eta) \det(E_0) \right|\left|2 \pi \det(E_0) \right|^{-1} \leq \frac{\epsilon_0}{2\pi} \left( \Lambda_{0,k} \gga_1 + \Lambda_{1,k} \gga_0 \right).
    \end{align*}
This completes the proof of Theorem \ref{theorem_eta_lambda}. 
\hfill $\Box$

\section{Error functions}\label{section error functions}
In this part, with $g(\xi)$ chosen as the Gaussian function \eqref{Gaussian_function}, we first present the corresponding window functions $\mathcal{C}\big(\xi^n g_\sigma\big)(t,\gga)$ for $n = 0, 1, 2$, respectively. A discussion of the quantities $\mathit{\Upsilon}_{n,k}(\eta)$ for the same values of $n$, as defined in \eqref{Upsilonmk}, will then follow.

Since \(g_\sigma(\xi)\) (given in \eqref{Gaussian_function}) is a positive even function, then
\[ 
I_m =\int_{-\infty}^{+\infty} \left|\xi^m g_\sigma(\xi)\right| \, d\xi= 2\int_{0}^{+\infty} \xi^m g_\sigma(\xi) \, d\xi= \frac{\sigma^l}{\sqrt{2\pi}} \cdot 2^{\frac{l+1}{2}} \Gamma\left( \frac{l+1}{2} \right),
\]
where $\Gamma(z)$ is Gamma function.
Then, we  have
\begin{align}\label{Im}
I_0 = 1; \quad
I_1 =  \frac{2\sigma}{\sqrt{2\pi}}; \quad
I_2 = {\sigma^2}; \quad
I_3 =  \frac {4 \sigma^3}{\sqrt{2\pi}}; \quad
I_4 = {3\sigma^4}; \quad
I_5=\frac{16\sigma^5}{\sqrt{2\pi}}.
\end{align}

Besides, when the  window function  is  \(g_\sigma(\xi)\), according to  \cite{chui2021time,li2022chirplet}, we have
\begin{align}
 &\mathcal{C}(g_\sigma)(\eta, \gga) =\frac 1{\sqrt{1+i2\pi \gs^2 \gl}} e^{-\frac{2\pi^2\gs^2 \eta^2}{1+i2\pi  \gs^2\gl}},\label{Cg}\\      
 &\mathcal{C}(\xi g_\sigma)(\eta, \gga)= \frac{-i2\pi  \sigma^2 \eta}{(1+i2\pi \sigma^2 \gga)^{\frac{3}{2}}} e^{-\frac{2\pi^2\sigma^2 \eta^2}{1+i2\pi  \sigma^2\gga}},\label{Cxig} \\
 &\mathcal{C}(\xi^2 g_\sigma)(\eta, \gga)= \sigma^2\left(\frac{1}{(1+i2\pi \sigma^2 \gga)^{\frac{3}{2}}} - \frac{(2\pi \sigma \eta)^2}{(1+i2\pi \sigma^2 \gga)^{\frac{5}{2}}}\right) e^{-\frac{2\pi^2\sigma^2 \eta^2}{1+i2\pi  \sigma^2\gga}} \label{Cxi2g}.
\end{align}
Note that 
\begin{equation*}
|\mathcal{C}(g_\sigma)(\eta, \gga)|= \frac{1}{(1 + 4\pi^2 \sigma^4 \gga^2)^{1/4}} e^{-\frac{2\pi^2 \sigma^2 \eta^2}{1 + 4\pi^2 \sigma^4 \gga^2}}.
\label{eq:42}
\end{equation*}
First, we can   obtain that
\begin{equation*}
|\mathcal{C}(g_\sigma)(\eta, \gga)| 
\leq 
\begin{cases} 
\displaystyle \frac{(1 + 4\pi^2 \sigma^4 \gga^2)^{3/4}}{2\pi^2 \sigma^2 \eta^2}, & \text{if }   2\pi^2 \sigma^2 \eta^2 \geq  1 + 4\pi^2 \sigma^4 \gga^2, \\[6pt]  
\displaystyle \frac{1}{(1 + 4\pi^2 \sigma^4 \gga^2)^{1/4}}, & \text{if } 2\pi^2 \sigma^2 \eta^2 <  1 + 4\pi^2 \sigma^4 \gga^2,.
\end{cases} \label{controlfunction_g}
\end{equation*}

Indeed, if \(x>1\), then \(e^{-x}< \frac{1}{x}\). Thus, if 
\(2\pi^2 \sigma^2 \eta^2 \geq 1 + 4\pi^2 \sigma^4 \gga^2\), then
\[
|\mathcal{C}(g_\sigma)(\eta, \gga)| \leq \frac{1}{(1 + 4\pi^2 \sigma^4 \gga^2)^{1/4}} \frac{1 + 4\pi^2 \sigma^4 \gga^2}{2\pi^2 \sigma^2 \eta^2} = \frac{(1 + 4\pi^2 \sigma^4\gga^2)^{3/4}}{2\pi^2 \sigma^2 \eta^2};
\]
Otherwise, for \(2\pi^2 \sigma^2 \eta^2 < 1 + 4\pi^2 \sigma^4 \gga^2\), we have
\[
|\mathcal{C}(g_\sigma)(\eta, \gga)| \leq \frac{1}{(1 + 4\pi^2 \gga^2)^{1/4}}.
\]
From the \eqref{Cxig} and \eqref{Cxi2g}, then we have 
\begin{equation*}
|\mathcal{C}(\xi g_\sigma)(\eta, \gga)| 
\leq 
\begin{cases} 
\displaystyle \frac{\sqrt{2}\sigma(1 + 4\pi^2 \sigma^4 \gga^2)^{1/4}}{\left(2\pi^2 \sigma^2 \eta^2\right)^{\frac{1}{2}}}, & \text{if }   2\pi^2 \sigma^2 \eta^2 \geq  1 + 4\pi^2 \sigma^4 \gga^2, \\[6pt]  
\displaystyle \frac{\sqrt{2}\sigma (2\pi^2 \sigma^2 \eta^2)^{\frac{1}{2}}}{(1 + 4\pi^2 \sigma^4 \gga^2)^{3/4}}, & \text{if } 2\pi^2 \sigma^2 \eta^2 <  1 + 4\pi^2 \sigma^4 \gga^2,.
\end{cases} \label{controlfunction_tg}
\end{equation*}
\begin{equation*}
|\mathcal{C}(\xi^2 g_\sigma)(\eta, \gga)| 
\leq 
\begin{cases} 
\displaystyle \frac{(1 + 4\pi^2 \sigma^4 \gga^2)^{1/4}}{2\pi^2 \sigma^2 \eta^2}+\frac{2}{\left(1+4\pi^2 \sigma^4\gga^2\right)^{\frac{1}{4}}}, & \text{if }   2\pi^2 \sigma^2 \eta^2 \geq  1 + 4\pi^2 \sigma^4 \gga^2, \\[6pt]  
\displaystyle \frac{1}{(1 + 4\pi^2 \sigma^4 \gga^2)^{3/4}}+\frac{4 \pi^2 \sigma^2 \eta^2}{\left(1+4\pi^2 \sigma^4 \gga^2\right)^{\frac{5}{4}}}, & \text{if } 2\pi^2 \sigma^2 \eta^2 <  1 + 4\pi^2 \sigma^4 \gga^2,.
\end{cases} \label{controlfunction_t^2g}
\end{equation*}
Thus, for any \(l=1,2,\cdots, K\),  the control functions \eqref {Upsilonmk} can be given as follows
\begin{align*}
    \Upsilon_{0,l}(\eta) &= \max\left( \frac{1}{2^{1/4}(\pi \sigma \Delta_1)^{1/2}}, \frac{1}{(1+4\pi^2 \sigma^4  \Delta_2^2)^{1/4}} \right);  \\
    \Upsilon_{1,l}(\eta) &= \max\left( \frac{2^{1/4}\sigma^{1/2}}{(\pi \Delta_1)^{1/2}}, \frac{2^{1/2}\sigma}{(1+4\pi^2 \sigma^4  \Delta_2^2)^{1/4}} \right);  \\
    \Upsilon_{2,l}(\eta) &= \max\Bigg( 
        \frac{1}{2^{3/4}(\pi \sigma \Delta_1)^{3/2}} + \frac{2}{(1+4\pi^2 \sigma^4  \Delta_2^2)^{1/4}}, \frac{1}{(1+4\pi^2 \sigma^4  \Delta_2^2)^{3/4}} + \frac{2}{(1+4\pi^2 \sigma^4  \Delta_2^2)^{1/4}} 
    \Bigg). 
\end{align*}

\end{appendices}





\bibliographystyle{elsarticle-num-names} 

\bibliography{IEEEabrv,ref_STFCT20251003}

\begin{thebibliography}{49}
\expandafter\ifx\csname natexlab\endcsname\relax\def\natexlab#1{#1}\fi
\providecommand{\url}[1]{\texttt{#1}}
\providecommand{\href}[2]{#2}
\providecommand{\path}[1]{#1}
\providecommand{\DOIprefix}{doi:}
\providecommand{\ArXivprefix}{arXiv:}
\providecommand{\URLprefix}{URL: }
\providecommand{\Pubmedprefix}{pmid:}
\providecommand{\doi}[1]{\href{http://dx.doi.org/#1}{\path{#1}}}
\providecommand{\Pubmed}[1]{\href{pmid:#1}{\path{#1}}}
\providecommand{\bibinfo}[2]{#2}
\ifx\xfnm\relax \def\xfnm[#1]{\unskip,\space#1}\fi
\bibitem[{Stankovi{\'c} et~al.(2013)Stankovi{\'c}, Dakovi{\'c}, and Thayaparan}]{stankovic2013time}
\bibinfo{author}{L.~Stankovi{\'c}}, \bibinfo{author}{M.~Dakovi{\'c}}, \bibinfo{author}{T.~Thayaparan}, \bibinfo{title}{Time-frequency Signal Analysis with Applications}, \bibinfo{publisher}{Artech House}, \bibinfo{year}{2013}.
\bibitem[{Daubechies(1992)}]{daubechies1992ten}
\bibinfo{author}{I.~Daubechies}, \bibinfo{title}{Ten Lectures on Wavelets}, \bibinfo{publisher}{SIAM}, \bibinfo{year}{1992}.
\bibitem[{Mallat(1999)}]{mallat1999wavelet}
\bibinfo{author}{S.~Mallat}, \bibinfo{title}{A Wavelet Tour of Signal Processing}, \bibinfo{publisher}{Elsevier}, \bibinfo{year}{1999}.
\bibitem[{Cohen(1995)}]{cohen1995time}
\bibinfo{author}{L.~Cohen}, \bibinfo{title}{Time-frequency Analysis}, volume \bibinfo{volume}{778}, \bibinfo{publisher}{Prentice Hall PTR New Jersey}, \bibinfo{year}{1995}.
\bibitem[{Hlawatsch and Boudreaux-Bartels(1992)}]{hlawatsch1992linear}
\bibinfo{author}{F.~Hlawatsch}, \bibinfo{author}{G.~F. Boudreaux-Bartels},
\newblock \bibinfo{title}{Linear and quadratic time-frequency signal representations},
\newblock \bibinfo{journal}{IEEE Signal Processing Magazine} \bibinfo{volume}{9} (\bibinfo{year}{1992}) \bibinfo{pages}{21--67}.
\bibitem[{Auger and Flandrin(1995)}]{auger1995improving}
\bibinfo{author}{F.~Auger}, \bibinfo{author}{P.~Flandrin},
\newblock \bibinfo{title}{Improving the readability of time-frequency and time-scale representations by the reassignment method},
\newblock \bibinfo{journal}{IEEE Transactions on Signal Processing} \bibinfo{volume}{43} (\bibinfo{year}{1995}) \bibinfo{pages}{1068--1089}.
\bibitem[{Daubechies and Maes(1996)}]{daubechies1996nonlinear}
\bibinfo{author}{I.~Daubechies}, \bibinfo{author}{S.~Maes},
\newblock \bibinfo{title}{A nonlinear squeezing of the continuous wavelet transform},
\newblock \bibinfo{journal}{Wavelets in Medicine and Biology}  (\bibinfo{year}{1996}) \bibinfo{pages}{527--546}.
\bibitem[{Daubechies et~al.(2011)Daubechies, Lu, and Wu}]{daubechies2011synchrosqueezed}
\bibinfo{author}{I.~Daubechies}, \bibinfo{author}{J.~Lu}, \bibinfo{author}{H.-T. Wu},
\newblock \bibinfo{title}{Synchrosqueezed wavelet transforms: {A}n empirical mode decomposition-like tool},
\newblock \bibinfo{journal}{Applied and Computational Harmonic Analysis} \bibinfo{volume}{30} (\bibinfo{year}{2011}) \bibinfo{pages}{243--261}.
\bibitem[{Oberlin et~al.(2015)Oberlin, Meignen, and Perrier}]{oberlin2015second}
\bibinfo{author}{T.~Oberlin}, \bibinfo{author}{S.~Meignen}, \bibinfo{author}{V.~Perrier},
\newblock \bibinfo{title}{Second-order synchrosqueezing transform or invertible reassignment? {T}owards ideal time-frequency representations},
\newblock \bibinfo{journal}{IEEE Transactions on Signal Processing} \bibinfo{volume}{63} (\bibinfo{year}{2015}) \bibinfo{pages}{1335--1344}.
\bibitem[{Oberlin and Meignen(2017)}]{oberlin2017second}
\bibinfo{author}{T.~Oberlin}, \bibinfo{author}{S.~Meignen},
\newblock \bibinfo{title}{The second-order wavelet synchrosqueezing transform},
\newblock in: \bibinfo{booktitle}{2017 IEEE International Conference on Acoustics, Speech and Signal Processing (ICASSP)}, \bibinfo{organization}{IEEE}, \bibinfo{year}{2017}, pp. \bibinfo{pages}{3994--3998}.
\bibitem[{Behera et~al.(2018)Behera, Meignen, and Oberlin}]{behera2018theoretical}
\bibinfo{author}{R.~Behera}, \bibinfo{author}{S.~Meignen}, \bibinfo{author}{T.~Oberlin},
\newblock \bibinfo{title}{Theoretical analysis of the second-order synchrosqueezing transform},
\newblock \bibinfo{journal}{Applied and Computational Harmonic Analysis} \bibinfo{volume}{45} (\bibinfo{year}{2018}) \bibinfo{pages}{379--404}.
\bibitem[{Pham and Meignen(2017)}]{pham2017high}
\bibinfo{author}{D.-H. Pham}, \bibinfo{author}{S.~Meignen},
\newblock \bibinfo{title}{High-order synchrosqueezing transform for multicomponent signals analysis—{W}ith an application to gravitational-wave signal},
\newblock \bibinfo{journal}{IEEE Transactions on Signal Processing} \bibinfo{volume}{65} (\bibinfo{year}{2017}) \bibinfo{pages}{3168--3178}.
\bibitem[{Yu et~al.(2017)Yu, Yu, and Xu}]{yu2017synchroextracting}
\bibinfo{author}{G.~Yu}, \bibinfo{author}{M.~Yu}, \bibinfo{author}{C.~Xu},
\newblock \bibinfo{title}{Synchroextracting transform},
\newblock \bibinfo{journal}{IEEE Transactions on Industrial Electronics} \bibinfo{volume}{64} (\bibinfo{year}{2017}) \bibinfo{pages}{8042--8054}.
\bibitem[{Yu et~al.(2018)Yu, Wang, and Zhao}]{yu2018multisynchrosqueezing}
\bibinfo{author}{G.~Yu}, \bibinfo{author}{Z.~Wang}, \bibinfo{author}{P.~Zhao},
\newblock \bibinfo{title}{Multisynchrosqueezing transform},
\newblock \bibinfo{journal}{IEEE Transactions on Industrial Electronics} \bibinfo{volume}{66} (\bibinfo{year}{2018}) \bibinfo{pages}{5441--5455}.
\bibitem[{Sheu et~al.(2017)Sheu, Hsu, Chou, and Wu}]{sheu2017entropy}
\bibinfo{author}{Y.-L. Sheu}, \bibinfo{author}{L.-Y. Hsu}, \bibinfo{author}{P.-T. Chou}, \bibinfo{author}{H.-T. Wu},
\newblock \bibinfo{title}{Entropy-based time-varying window width selection for nonlinear-type time--frequency analysis},
\newblock \bibinfo{journal}{International Journal of Data Science and Analytics} \bibinfo{volume}{3} (\bibinfo{year}{2017}) \bibinfo{pages}{231--245}.
\bibitem[{Berrian and Saito(2017)}]{berrian2017adaptive}
\bibinfo{author}{A.~Berrian}, \bibinfo{author}{N.~Saito},
\newblock \bibinfo{title}{Adaptive synchrosqueezing based on a quilted short-time {F}ourier transform},
\newblock in: \bibinfo{booktitle}{Wavelets and Sparsity XVII}, volume \bibinfo{volume}{10394}, \bibinfo{organization}{SPIE}, \bibinfo{year}{2017}, pp. \bibinfo{pages}{413--432}.
\bibitem[{Li et~al.(2020{\natexlab{a}})Li, Cai, and Jiang}]{li2020adaptive}
\bibinfo{author}{L.~Li}, \bibinfo{author}{H.~Cai}, \bibinfo{author}{Q.~Jiang},
\newblock \bibinfo{title}{Adaptive synchrosqueezing transform with a time-varying parameter for non-stationary signal separation},
\newblock \bibinfo{journal}{Applied and Computational Harmonic Analysis} \bibinfo{volume}{49} (\bibinfo{year}{2020}{\natexlab{a}}) \bibinfo{pages}{1075--1106}.
\bibitem[{Li et~al.(2020{\natexlab{b}})Li, Cai, Han, Jiang, and Ji}]{li2020adaptivestft}
\bibinfo{author}{L.~Li}, \bibinfo{author}{H.~Cai}, \bibinfo{author}{H.~Han}, \bibinfo{author}{Q.~Jiang}, \bibinfo{author}{H.~Ji},
\newblock \bibinfo{title}{Adaptive short-time {F}ourier transform and synchrosqueezing transform for non-stationary signal separation},
\newblock \bibinfo{journal}{Signal Processing} \bibinfo{volume}{166} (\bibinfo{year}{2020}{\natexlab{b}}) \bibinfo{pages}{107231}.
\bibitem[{Chui et~al.(2021)Chui, Jiang, Li, and Lu}]{chui2021time}
\bibinfo{author}{C.~K. Chui}, \bibinfo{author}{Q.~Jiang}, \bibinfo{author}{L.~Li}, \bibinfo{author}{J.~Lu},
\newblock \bibinfo{title}{Time-scale-chirp\_rate operator for recovery of non-stationary signal components with crossover instantaneous frequency curves},
\newblock \bibinfo{journal}{Applied and Computational Harmonic Analysis} \bibinfo{volume}{54} (\bibinfo{year}{2021}) \bibinfo{pages}{323--344}.
\bibitem[{Li et~al.(2022)Li, Han, Jiang, and Chui}]{li2022chirplet}
\bibinfo{author}{L.~Li}, \bibinfo{author}{N.~Han}, \bibinfo{author}{Q.~Jiang}, \bibinfo{author}{C.~K. Chui},
\newblock \bibinfo{title}{A chirplet transform-based mode retrieval method for multicomponent signals with crossover instantaneous frequencies},
\newblock \bibinfo{journal}{Digital Signal Processing} \bibinfo{volume}{120} (\bibinfo{year}{2022}) \bibinfo{pages}{103262}.
\bibitem[{Chui et~al.(2023)Chui, Jiang, Li, and Lu}]{chui2023analysis}
\bibinfo{author}{C.~K. Chui}, \bibinfo{author}{Q.~Jiang}, \bibinfo{author}{L.~Li}, \bibinfo{author}{J.~Lu},
\newblock \bibinfo{title}{Analysis of a direct separation method based on adaptive chirplet transform for signals with crossover instantaneous frequencies},
\newblock \bibinfo{journal}{Applied and Computational Harmonic Analysis} \bibinfo{volume}{62} (\bibinfo{year}{2023}) \bibinfo{pages}{24--40}.
\bibitem[{Zhu et~al.(2020)Zhu, Yang, Zhang, Gao, and Liu}]{zhu2020frequency}
\bibinfo{author}{X.~Zhu}, \bibinfo{author}{H.~Yang}, \bibinfo{author}{Z.~Zhang}, \bibinfo{author}{J.~Gao}, \bibinfo{author}{N.~Liu},
\newblock \bibinfo{title}{Frequency-chirprate reassignment},
\newblock \bibinfo{journal}{Digital Signal Processing} \bibinfo{volume}{104} (\bibinfo{year}{2020}) \bibinfo{pages}{102783}.
\bibitem[{Chen and Wu(2023)}]{chen2023disentangling}
\bibinfo{author}{Z.~Chen}, \bibinfo{author}{H.-T. Wu},
\newblock \bibinfo{title}{Disentangling modes with crossover instantaneous frequencies by synchrosqueezed chirplet transforms, from theory to application},
\newblock \bibinfo{journal}{Applied and Computational Harmonic Analysis} \bibinfo{volume}{62} (\bibinfo{year}{2023}) \bibinfo{pages}{84--122}.
\bibitem[{Chen et~al.(2024{\natexlab{a}})Chen, Xie, Cui, and Su}]{chen2024multiple}
\bibinfo{author}{T.~Chen}, \bibinfo{author}{L.~Xie}, \bibinfo{author}{M.~Cui}, \bibinfo{author}{H.~Su},
\newblock \bibinfo{title}{Multiple enhanced synchrosqueezing in the time--frequency--chirprate space},
\newblock \bibinfo{journal}{Signal Processing} \bibinfo{volume}{222} (\bibinfo{year}{2024}{\natexlab{a}}) \bibinfo{pages}{109541}.
\bibitem[{Chen et~al.(2024{\natexlab{b}})Chen, Zhang, and Yang}]{chen2024composite}
\bibinfo{author}{X.~Chen}, \bibinfo{author}{Z.~Zhang}, \bibinfo{author}{W.~Yang},
\newblock \bibinfo{title}{Composite signal detection using multisynchrosqueezing wavelet transform},
\newblock \bibinfo{journal}{Digital Signal Processing} \bibinfo{volume}{149} (\bibinfo{year}{2024}{\natexlab{b}}) \bibinfo{pages}{104482}.
\bibitem[{Jiang et~al.(2025)Jiang, Li, Chen, and Li}]{jiang2025synchrosqueezed}
\bibinfo{author}{Q.~Jiang}, \bibinfo{author}{S.~Li}, \bibinfo{author}{J.~Chen}, \bibinfo{author}{L.~Li},
\newblock \bibinfo{title}{Synchrosqueezed x-ray wavelet--chirplet transform for accurate chirp rate estimation and retrieval of modes from multicomponent signals with crossover instantaneous frequencies},
\newblock \bibinfo{journal}{Mechanical Systems and Signal Processing} \bibinfo{volume}{238} (\bibinfo{year}{2025}) \bibinfo{pages}{113193}.
\bibitem[{He et~al.(2019)He, Cao, Wang, and Chen}]{he2019time}
\bibinfo{author}{D.~He}, \bibinfo{author}{H.~Cao}, \bibinfo{author}{S.~Wang}, \bibinfo{author}{X.~Chen},
\newblock \bibinfo{title}{Time-reassigned synchrosqueezing transform: {T}he algorithm and its applications in mechanical signal processing},
\newblock \bibinfo{journal}{Mechanical Systems and Signal Processing} \bibinfo{volume}{117} (\bibinfo{year}{2019}) \bibinfo{pages}{255--279}.
\bibitem[{Yu et~al.(2020)Yu, Lin, Wang, and Li}]{yu2020time}
\bibinfo{author}{G.~Yu}, \bibinfo{author}{T.~Lin}, \bibinfo{author}{Z.~Wang}, \bibinfo{author}{Y.~Li},
\newblock \bibinfo{title}{Time-reassigned multisynchrosqueezing transform for bearing fault diagnosis of rotating machinery},
\newblock \bibinfo{journal}{IEEE Transactions on Industrial Electronics} \bibinfo{volume}{68} (\bibinfo{year}{2020}) \bibinfo{pages}{1486--1496}.
\bibitem[{Li et~al.(2022)Li, Zhang, Auger, and Zhu}]{li2022theoretical}
\bibinfo{author}{W.~Li}, \bibinfo{author}{Z.~Zhang}, \bibinfo{author}{F.~Auger}, \bibinfo{author}{X.~Zhu},
\newblock \bibinfo{title}{Theoretical analysis of time-reassigned synchrosqueezing wavelet transform},
\newblock \bibinfo{journal}{Applied Mathematics Letters} \bibinfo{volume}{132} (\bibinfo{year}{2022}) \bibinfo{pages}{108141}.
\bibitem[{Fourer and Auger(2019)}]{fourer2019second}
\bibinfo{author}{D.~Fourer}, \bibinfo{author}{F.~Auger},
\newblock \bibinfo{title}{Second-order time-reassigned synchrosqueezing transform: Application to {D}raupner wave analysis},
\newblock in: \bibinfo{booktitle}{2019 27th European Signal Processing Conference (EUSIPCO)}, \bibinfo{organization}{IEEE}, \bibinfo{year}{2019}, pp. \bibinfo{pages}{1--5}.
\bibitem[{He et~al.(2020)He, Tu, Bao, Hu, and Li}]{he2020gaussian}
\bibinfo{author}{Z.~He}, \bibinfo{author}{X.~Tu}, \bibinfo{author}{W.~Bao}, \bibinfo{author}{Y.~Hu}, \bibinfo{author}{F.~Li},
\newblock \bibinfo{title}{Gaussian-modulated linear group delay model: Application to second-order time-reassigned synchrosqueezing transform},
\newblock \bibinfo{journal}{Signal Processing} \bibinfo{volume}{167} (\bibinfo{year}{2020}) \bibinfo{pages}{107275}.
\bibitem[{Yu(2019)}]{yu2019concentrated}
\bibinfo{author}{G.~Yu},
\newblock \bibinfo{title}{A concentrated time--frequency analysis tool for bearing fault diagnosis},
\newblock \bibinfo{journal}{IEEE Transactions on Instrumentation and Measurement} \bibinfo{volume}{69} (\bibinfo{year}{2019}) \bibinfo{pages}{371--381}.
\bibitem[{Li et~al.(2023)Li, Auger, Zhang, and Zhu}]{li2023newton}
\bibinfo{author}{W.~Li}, \bibinfo{author}{F.~Auger}, \bibinfo{author}{Z.~Zhang}, \bibinfo{author}{X.~Zhu},
\newblock \bibinfo{title}{Newton time-extracting wavelet transform: An effective tool for characterizing frequency-varying signals with weakly-separated components and theoretical analysis},
\newblock \bibinfo{journal}{Signal Processing} \bibinfo{volume}{209} (\bibinfo{year}{2023}) \bibinfo{pages}{109017}.
\bibitem[{Yu and Lin(2021)}]{yu2021second}
\bibinfo{author}{G.~Yu}, \bibinfo{author}{T.~R. Lin},
\newblock \bibinfo{title}{Second-order transient-extracting transform for the analysis of impulsive-like signals},
\newblock \bibinfo{journal}{Mechanical Systems and Signal Processing} \bibinfo{volume}{147} (\bibinfo{year}{2021}) \bibinfo{pages}{107069}.
\bibitem[{He et~al.(2019)He, Tu, Bao, Hu, and Li}]{he2019second}
\bibinfo{author}{Z.~He}, \bibinfo{author}{X.~Tu}, \bibinfo{author}{W.~Bao}, \bibinfo{author}{Y.~Hu}, \bibinfo{author}{F.~Li},
\newblock \bibinfo{title}{Second-order transient-extracting transform with application to time-frequency filtering},
\newblock \bibinfo{journal}{IEEE Transactions on Instrumentation and Measurement} \bibinfo{volume}{69} (\bibinfo{year}{2019}) \bibinfo{pages}{5428--5437}.
\bibitem[{Bao et~al.(2022)Bao, Li, and Chen}]{bao2022generalized}
\bibinfo{author}{W.~Bao}, \bibinfo{author}{F.~Li}, \bibinfo{author}{Z.~Chen},
\newblock \bibinfo{title}{Generalized transient-squeezing transform: Algorithm and applications},
\newblock \bibinfo{journal}{IEEE Transactions on Instrumentation and Measurement} \bibinfo{volume}{71} (\bibinfo{year}{2022}) \bibinfo{pages}{1--10}.
\bibitem[{Bao et~al.(2021)Bao, Hu, and Li}]{bao2021generalized}
\bibinfo{author}{W.~Bao}, \bibinfo{author}{Y.~Hu}, \bibinfo{author}{F.~Li},
\newblock \bibinfo{title}{Generalized transient-extracting transform and its accurate signal reconstruction},
\newblock \bibinfo{journal}{IEEE Transactions on Industrial Electronics} \bibinfo{volume}{69} (\bibinfo{year}{2021}) \bibinfo{pages}{10552--10563}.
\bibitem[{Yang et~al.(2014)Yang, Peng, Zhang, and Meng}]{yang2014frequency}
\bibinfo{author}{Y.~Yang}, \bibinfo{author}{Z.~Peng}, \bibinfo{author}{W.~Zhang}, \bibinfo{author}{G.~Meng},
\newblock \bibinfo{title}{Frequency-varying group delay estimation using frequency domain polynomial chirplet transform},
\newblock \bibinfo{journal}{Mechanical Systems and Signal Processing} \bibinfo{volume}{46} (\bibinfo{year}{2014}) \bibinfo{pages}{146--162}.
\bibitem[{Chui and Mhaskar(2016)}]{chui2016signal}
\bibinfo{author}{C.~K. Chui}, \bibinfo{author}{H.~Mhaskar},
\newblock \bibinfo{title}{Signal decomposition and analysis via extraction of frequencies},
\newblock \bibinfo{journal}{Applied and Computational Harmonic Analysis} \bibinfo{volume}{40} (\bibinfo{year}{2016}) \bibinfo{pages}{97--136}.
\bibitem[{Li et~al.(2022)Li, Chui, and Jiang}]{li2022direct}
\bibinfo{author}{L.~Li}, \bibinfo{author}{C.~K. Chui}, \bibinfo{author}{Q.~Jiang},
\newblock \bibinfo{title}{Direct signal separation via extraction of local frequencies with adaptive time-varying parameters},
\newblock \bibinfo{journal}{IEEE Transactions on Signal Processing} \bibinfo{volume}{70} (\bibinfo{year}{2022}) \bibinfo{pages}{2321--2333}.
\bibitem[{Chui et~al.(2021{\natexlab{a}})Chui, Jiang, Li, and Lu}]{chui2021analysis}
\bibinfo{author}{C.~K. Chui}, \bibinfo{author}{Q.~Jiang}, \bibinfo{author}{L.~Li}, \bibinfo{author}{J.~Lu},
\newblock \bibinfo{title}{Analysis of an adaptive short-time {F}ourier transform-based multicomponent signal separation method derived from linear chirp local approximation},
\newblock \bibinfo{journal}{Journal of Computational and Applied Mathematics} \bibinfo{volume}{396} (\bibinfo{year}{2021}{\natexlab{a}}) \bibinfo{pages}{113607}.
\bibitem[{Chui et~al.(2021{\natexlab{b}})Chui, Jiang, Li, and Lu}]{chui2021signal}
\bibinfo{author}{C.~K. Chui}, \bibinfo{author}{Q.~Jiang}, \bibinfo{author}{L.~Li}, \bibinfo{author}{J.~Lu},
\newblock \bibinfo{title}{Signal separation based on adaptive continuous wavelet-like transform and analysis},
\newblock \bibinfo{journal}{Applied and Computational Harmonic Analysis} \bibinfo{volume}{53} (\bibinfo{year}{2021}{\natexlab{b}}) \bibinfo{pages}{151--179}.
\bibitem[{Mann and Haykin(1995)}]{mann1995chirplet}
\bibinfo{author}{S.~Mann}, \bibinfo{author}{S.~Haykin},
\newblock \bibinfo{title}{The chirplet transform: Physical considerations},
\newblock \bibinfo{journal}{IEEE Transactions on Signal Processing} \bibinfo{volume}{43} (\bibinfo{year}{1995}) \bibinfo{pages}{2745--2761}.
\bibitem[{Huang et~al.(2024)Huang, Zhao, and Cui}]{Huang2024horizontal}
\bibinfo{author}{X.~Huang}, \bibinfo{author}{D.~Zhao}, \bibinfo{author}{L.~Cui},
\newblock \bibinfo{title}{Horizontal rearrangement frequency domain chirplet transform: algorithm and applications},
\newblock \bibinfo{journal}{Measurement Science and Technology} \bibinfo{volume}{35} (\bibinfo{year}{2024}) \bibinfo{pages}{116125}.
\bibitem[{Zhao et~al.(2025)Zhao, Du, and Wang}]{zhao2025horizontal}
\bibinfo{author}{D.~Zhao}, \bibinfo{author}{S.~Du}, \bibinfo{author}{T.~Wang},
\newblock \bibinfo{title}{Horizontal local-squeezing frequency domain chirplet transform},
\newblock \bibinfo{journal}{IEEE Transactions on Instrumentation and Measurement}  (\bibinfo{year}{2025}).
\bibitem[{Meignen and Singh(2022)}]{meignen2022analysis}
\bibinfo{author}{S.~Meignen}, \bibinfo{author}{N.~Singh},
\newblock \bibinfo{title}{Analysis of reassignment operators used in synchrosqueezing transforms: With an application to instantaneous frequency estimation},
\newblock \bibinfo{journal}{IEEE Transactions on Signal Processing} \bibinfo{volume}{70} (\bibinfo{year}{2022}) \bibinfo{pages}{216--227}.
\bibitem[{Zhang et~al.(2022)Zhang, Liu, Tan, Yang, and Zhang}]{zhang2022two}
\bibinfo{author}{R.~Zhang}, \bibinfo{author}{X.~Liu}, \bibinfo{author}{Y.~Tan}, \bibinfo{author}{X.~Yang}, \bibinfo{author}{L.~Zhang},
\newblock \bibinfo{title}{Two dimensional local maximum synchroextracting chirplet transfrom and application of characterizing micro-{D}oppler signals},
\newblock \bibinfo{journal}{Signal Processing} \bibinfo{volume}{198} (\bibinfo{year}{2022}) \bibinfo{pages}{108598}.
\bibitem[{Golub and Van~Loan(2013)}]{golub2013matrix}
\bibinfo{author}{G.~H. Golub}, \bibinfo{author}{C.~F. Van~Loan}, \bibinfo{title}{Matrix Computations}, \bibinfo{publisher}{JHU press}, \bibinfo{year}{2013}.
\bibitem[{Wellard et~al.(2015)Wellard, Erbe, Fouda, and Blewitt}]{wellard2015vocalisations}
\bibinfo{author}{R.~Wellard}, \bibinfo{author}{C.~Erbe}, \bibinfo{author}{L.~Fouda}, \bibinfo{author}{M.~Blewitt},
\newblock \bibinfo{title}{Vocalisations of killer whales ({O}rcinus orca) in the {B}remer {C}anyon, {W}estern {A}ustralia},
\newblock \bibinfo{journal}{PLOS ONE} \bibinfo{volume}{10} (\bibinfo{year}{2015}) \bibinfo{pages}{e0136535}.

\end{thebibliography}

\end{document}